\documentclass[12pt]{article}

\usepackage{amssymb}
\usepackage{ textcomp }
\usepackage[margin=1in]{geometry}
\usepackage{complexity}
\usepackage{mathtools}
\usepackage{xspace}
\usepackage{rpmacros}
\RequirePackage[colorlinks=true]{hyperref}
\hypersetup{
  linkcolor=[rgb]{0,0,0.4},
  citecolor=[rgb]{0, 0.4, 0},
  urlcolor=[rgb]{0.6, 0, 0}
}
\usepackage{mathpazo}
\usepackage{bbm}
\usepackage{dsfont}
\usepackage{tikz}
\usetikzlibrary{arrows,shapes.geometric,positioning}

\usepackage{amsthm}
\usepackage{thmtools,thm-restate}

\numberwithin{equation}{section}
\declaretheoremstyle[bodyfont=\it,qed=\qedsymbol]{noproofstyle}

\declaretheorem[numberlike=equation]{observation}

\declaretheorem[name=Observation,numbered=no]{observation*}

\declaretheorem[numberlike=equation]{fact}

\declaretheorem[numberlike=equation]{problem}

\declaretheorem[numberlike=equation]{theorem}

\declaretheorem[name=Theorem,numbered=no]{theorem*}

\declaretheorem[numberlike=equation]{lemma}
\declaretheorem[name=Lemma,numbered=no]{lemma*}

\declaretheorem[numberlike=equation]{corollary}
\declaretheorem[name=Corollary,numbered=no]{corollary*}

\declaretheorem[numberlike=equation]{proposition}
\declaretheorem[name=Proposition,numbered=no]{proposition*}

\declaretheorem[numberlike=equation]{claim}
\declaretheorem[name=Claim,numbered=no]{claim*}

\declaretheorem[numberlike=equation]{conjecture}
\declaretheorem[name=Conjecture,numbered=no]{conjecture*}

\declaretheorem[name=Question,numbered=no]{question*}

\declaretheoremstyle[bodyfont=\it,qed=$\lozenge$]{defstyle} 

\declaretheorem[numberlike=equation,style=defstyle]{definition}
\declaretheorem[unnumbered,name=Definition,style=defstyle]{definition*}

\declaretheorem[unnumbered,name=Example,style=defstyle]{example*}

\declaretheorem[unnumbered,name=Notation=defstyle]{notation*}

\declaretheorem[unnumbered,name=Construction,style=defstyle]{construction*}

\declaretheorem[numberlike=equation,style=defstyle]{remark}
\declaretheorem[unnumbered,name=Remark,style=defstyle]{remark*}

\declaretheorem[numberlike=equation,style=defstyle]{assumption}
\declaretheorem[unnumbered,name=Assumption,style=defstyle]{assumption*}

\usepackage{nth}
\usepackage{intcalc}
\usepackage{etoolbox}
\usepackage{xstring}
\hypersetup{
}

\usepackage{ifpdf}
\ifpdf
\else
\usepackage[quadpoints=false]{hypdvips}
\fi

\newcommand{\ECCCurl}[2]{http://eccc.hpi-web.de/report/\ifnumcomp{#1}{>}{93}{19}{20}#1/#2/}

\newcommand{\shortECCC}[2]{\texttt{\href{http://eccc.hpi-web.de/report/\ifnumcomp{#1}{>}{93}{19}{20}#1/#2/}{eccc:TR#1-#2}}}

\newcommand{\parseECCC}[1]{
\StrSubstitute{#1}{TR}{}[\tmpstring]%
\IfSubStr{\tmpstring}{/}{ 
\StrBefore{\tmpstring}{/}[\ecccyear]%
\StrBehind{\tmpstring}{/}[\ecccreport]%
}{
\StrBefore{\tmpstring}{-}[\ecccyear]%
\StrBehind{\tmpstring}{-}[\ecccreport]%
}%
\shortECCC{\ecccyear}{\ecccreport}}

\usepackage{algorithmicx}
\usepackage{algorithm} 
\usepackage[noend]{algpseudocode}

\algrenewcommand\algorithmicindent{1.0em}%
\usepackage[T1]{fontenc}
\usepackage[inline]{enumitem}

\newcommand{\vecalpha}{\boldsymbol{\alpha}}
\newcommand{\vecbeta}{\boldsymbol{\beta}}

\newcommand{\eqdef}{\vcentcolon=}

	\renewcommand{\vec}[1]{{\mathbf{#1}}}

	\makeatletter
	\newcommand{\va}{{\vec{a}}\@ifnextchar{^}{\!\:}{}}
	\newcommand{\vb}{{\vec{b}}\@ifnextchar{^}{\!\:}{}}
	\newcommand{\vc}{{\vec{c}}\@ifnextchar{^}{\!\:}{}}
	\newcommand{\vd}{{\vec{d}}\@ifnextchar{^}{\!\:}{}}
	\newcommand{\ve}{{\vec{e}}\@ifnextchar{^}{\!\:}{}}
	\newcommand{\vy}{{\vec{y}}\@ifnextchar{^}{\!\:}{}}
	\newcommand{\vs}{{\vec{s}}\@ifnextchar{^}{\!\:}{}}
	\newcommand{\vt}{{\vec{t}}\@ifnextchar{^}{\!\:}{}}
	\newcommand{\vx}{{\vec{x}}\@ifnextchar{^}{}{}}		
	\newcommand{\vz}{{\vec{z}}\@ifnextchar{^}{\!\:}{}}
	\newcommand{\vv}{{\vec{v}}\@ifnextchar{^}{\!\:}{}}
	\newcommand{\vu}{{\vec{u}}\@ifnextchar{^}{\!\:}{}}
	\newcommand{\vf}{{\vec{f}}\@ifnextchar{^}{\!\:}{}}
	\newcommand{\vg}{{\vec{g}}\@ifnextchar{^}{\!\:}{}}
	\newcommand{\vr}{{\vec{r}}\@ifnextchar{^}{\!\:}{}}
	\newcommand{\vw}{{\vec{w}}\@ifnextchar{^}{\!\:}{}}

	\newcommand{\vY}{{\vec{Y}}\@ifnextchar{^}{\!\:}{}}
	\newcommand{\vX}{{\vec{X}}\@ifnextchar{^}{}{}}		
	\newcommand{\vZ}{{\vec{Z}}\@ifnextchar{^}{\!\:}{}}
	\newcommand{\vG}{{\vec{G}}\@ifnextchar{^}{\!\:}{}}
	
	\makeatother

	\newcommand{\vaa}{{\vecalpha}}
	\newcommand{\vbb}{{\vecbeta}}

\renewcommand{\C}{\mathbb{C}}
\renewcommand{\N}{\mathbb{N}}

\newcommand{\cA}{{\mathcal{A}}}
\newcommand{\cB}{{\mathcal{B}}}

\newcommand{\cD}{{\mathcal{D}}}

\newcommand{\cG}{{\mathcal{G}}}
\newcommand{\cF}{{\mathcal{F}}}
\newcommand{\cI}{{\mathcal{I}}}
\newcommand{\cJ}{{\mathcal{J}}}
\newcommand{\cK}{{\mathcal{K}}}
\newcommand{\calL}{{\mathcal{L}}}
\newcommand{\cQ}{{\mathcal{Q}}}
\newcommand{\calP}{{\mathcal{P}}}

\newcommand{\cT}{{\mathcal{T}}}

\newcommand{\cS}{{\mathcal{S}}}
\newcommand{\calS}{{\mathcal{S}}}
\newcommand{\calW}{{\mathcal{W}}}

\newcommand{\ideal}[1]{\left \langle{#1}\right \rangle}
\newcommand{\MVar}[2]{{#1}_1,\ldots, {#1}_{#2}}

\newcommand{\PRing}[3]{\mathbb{#1}[ \MVar{#2}{#3}]}

\newcommand{\CRing}[2]{\PRing{C}{#1}{#2}}

\newcommand{\spn}[1]{\operatorname{span}\{{#1}\}}
\renewcommand{\line}[1]{\operatorname{Line}({#1})}
\newcommand{\plane}[1]{\operatorname{Plane}({#1})}

\def\epsilon{\varepsilon} 

\newcommand{\rkq}{1000}
\newcommand{\dimV}{100}
\newcommand{\distV}{20}
\newcommand{\MS}{\text{Lin}}

\newenvironment{case}{%
	\let\olditem\item%
	\renewcommand\item[1][]{\olditem {\textbf{##1}} }%
	\begin{enumerate}[label={\textbf{Case (\roman*):}},itemindent=*,leftmargin=0em]}{\end{enumerate}%
}

\newcommand{\spsp}{\Sigma^{[3]}\Pi\Sigma\Pi^{[2]}}

\newif\ifEK
\EKtrue

\date{}

\title{Polynomial time deterministic identity testing algorithm for $\spsp$ circuits via Edelstein-Kelly type theorem for quadratic polynomials}
\author{Shir Peleg\thanks{Department of Computer Science, Tel Aviv University, Tel Aviv, Israel, E-mail: \texttt{shirpele@tauex.tau.ac.il, shpilka@tauex.tau.ac.il}. The research leading to these results has received funding from the  Israel Science Foundation (grant number 552/16) and from the Len Blavatnik and the Blavatnik Family foundation. 
}   \and  Amir Shpilka\footnotemark[1]}
\begin{document}
\maketitle
\begin{abstract}

In this work we resolve conjectures of Beecken, Mitmann and Saxena \cite{DBLP:journals/iandc/BeeckenMS13} and   Gupta \cite{Gupta14},
by proving an analog of a theorem of Edelstein and Kelly for quadratic polynomials. As immediate corollary we obtain the first deterministic polynomial time black-box algorithm for testing zeroness of $\Sigma^{[3]}\Pi\Sigma\Pi^{[2]}$ circuits. 
\end{abstract}

\thispagestyle{empty}
\newpage

\tableofcontents

\thispagestyle{empty}
\newpage
\pagenumbering{arabic}

\section{Introduction}\label{sec:intro}

The polynomial identity testing problem (PIT) asks to determine, given an arithmetic circuit as input, whether the circuit computes the identically zero polynomial. The problem can be studied both in the black-box model where the algorithm can access the circuit only via querying its value at different inputs or in the white-box model where the algorithm also has access to the graph of computation and labeling of nodes.
While there is a well known and simple randomized black-box algorithm for the problem -- simply evaluate the circuit at a random input -- no efficient\footnote{Whenever we mention running-time we always express it as a function of the size of the input circuit and the number of variables. For simplicity we shall assume here that the size is polynomial in the number of variables, $n$.} deterministic algorithm for it is known, even in the white-box model, except for some special cases.

Devising an efficient deterministic algorithm for PIT is one of the main challenges of theoretical computer science due to the fundamental nature of the problem and its relation to other basic questions such as proving lower bounds for arithmetic circuits \cite{DBLP:conf/stoc/HeintzS80,DBLP:conf/fsttcs/Agrawal05,DBLP:journals/cc/KabanetsI04,DBLP:journals/siamcomp/DvirSY09,DBLP:journals/toc/ForbesSV18,DBLP:conf/coco/ChouKS18} and other derandomization problems \cite{DBLP:journals/cc/KoppartySS15,Mulmuley-GCT-V,DBLP:conf/approx/ForbesS13,DBLP:journals/cacm/FennerGT19,DBLP:conf/stoc/GurjarT17,DBLP:conf/focs/SvenssonT17}. For more on the PIT problem see  \cite{DBLP:journals/fttcs/ShpilkaY10,Saxena09,Saxena14, ForbesThesis}. 

Remarkable results by Agrawal and Vinay \cite{DBLP:conf/focs/AgrawalV08} and by Gupta et al. \cite{DBLP:journals/siamcomp/0001KKS16} show that in order to solve the PIT problem for general arithmetic circuits, it is sufficient to solve it for  low depth circuits -- unrestricted depth-$3$ circuits or homogeneous depth-$4$ circuits. Specifically, a polynomial time PIT algorithm for depth-$4$ circuits, denoted $\Sigma\Pi\Sigma\Pi$ circuits -- polynomials are represented as sums of products of sparse polynomials in the model --  implies a quasi-polynomial time PIT algorithm for general arithmetic circuits. Thus, from the point of view of PIT (and circuit lower bounds), small depth arithmetic circuits are as interesting a model as general arithmetic circuits. Because of that, those classes attracted a lot of attention in recent years and many lower bounds and PIT algorithms were devised for restricted models of low depth circuits.  

In this paper we give the first deterministic, polynomial time, black-box PIT algorithm for  $\Sigma^{[3]}\Pi\Sigma\Pi^{[2]}$ circuits. We achieve this by proving a generalization of a theorem due to Edelstein and Kelly, which is itself an extension of the Sylvester-Gallai theorem, to quadratic polynomials, thus resolving conjectures that were raised in the works of Beecken, Mitmann and Saxena \cite{DBLP:journals/iandc/BeeckenMS13} and   Gupta \cite{Gupta14}. We next survey known results for depth-$4$ circuits and explain the connection of PIT for small depth circuits and Sylvester-Gallai type theorems. 

\paragraph{Known results:}

We shall give a brief overview of known results for depth-$3$ and depth-$4$ circuits, as PIT for these models is tightly connected to Sylvester-Gallai type questions. Furthermore, by the results of Agrawal and Vinay \cite{DBLP:conf/focs/AgrawalV08} and Gupta et al. \cite{DBLP:journals/siamcomp/0001KKS16}, we know that resolving PIT in these models will resolve the question in the general setting.

Dvir and Shpilka \cite{DBLP:journals/siamcomp/DvirS07} gave the first quasi-polynomial time deterministic white-box algorithm for $\Sigma^{[k]}\Pi\Sigma$ circuits, for $k=O(1)$. Their main idea was bounding the rank of the linear forms appearing in \emph{simple and minimal} identities. The rank-based approach of \cite{DBLP:journals/siamcomp/DvirS07} led Karnin and Shpilka to devise a black-box algorithm for the problem of the same complexity \cite{DBLP:journals/combinatorica/KarninS11}. The work \cite{DBLP:journals/siamcomp/DvirS07} also highlighted the relation between PIT for depth-$3$ circuits and colored version of the Sylvester-Gallai problem and suggested that studying the relation between the two problems could lead to improved algorithms. This was carried out by Kayal and Saraf \cite{DBLP:conf/focs/KayalS09} who greatly improved Dvir and Shpilka's original result by applying high dimensional versions of the Sylvester-Gallai Theorem to the problem.
Currently, the best known PIT algorithm is due to Saxena and Seshadhri \cite{DBLP:journals/siamcomp/SaxenaS12} who gave a deterministic black-box algorithm running in time $n^{O(k)}$ for $\Sigma^{[k]}\Pi\Sigma$ circuits. 

For depth-$4$ circuits even less is known. Recall that $\Sigma^{[k]}\Pi\Sigma\Pi^{[r]}$ circuits compute polynomials that can be expressed in the form 
\[ P(x_1,\ldots,x_n)= \sum_{i=1}^{k}\prod_jQ_{i,j} \;,
\]
where $\deg(Q_{i,j})\leq r$. When we drop the superscript $r$ and write $\Sigma^{[k]}\Pi\Sigma\Pi$ then we mean that the degree of the $Q_{i,j}$s is unrestricted. The size of a depth-$4$ circuit is the number of wires in the circuit. 

Karnin et al. \cite{DBLP:journals/siamcomp/KarninMSV13} gave a quasi-polynomial time black-box PIT algorithm for \emph{multilinear}\footnote{A circuit model is called \emph{multilinear} if every subcomputation computes a multilinear polynomial.} $\Sigma^{[k]}\Pi\Sigma\Pi$ circuits. This was later improved by Saraf and Volkovich to an $n^{O(k^2)}$ algorithm \cite{DBLP:journals/combinatorica/SarafV18}. Beecken et al. \cite{DBLP:journals/iandc/BeeckenMS13} and Kumar and Saraf \cite{DBLP:journals/toc/0001S17}  considered circuits in which the \emph{algebraic rank} of the irreducible factors in each multiplication gate is bounded, and gave a quasi-polynomial time deterministic PIT algorithm for such  $\Sigma^{[k]}\Pi\Sigma\Pi$ circuits, when the bottom fan-in is also bounded by polylog$(n)$.
Thus, prior to this work no subexponential PIT algorithm was known for $\Sigma^{[k]}\Pi\Sigma\Pi$ circuits without multilinearity restriction or without a bound on the local algebraic rank.

\autoref{cor:main:PIT} gives the first polynomial time deterministic black-box PIT algorithm for $\spsp$ circuits. We obtain it by resolving conjectures of Beecken et al. \cite{DBLP:journals/iandc/BeeckenMS13}  and of Gupta \cite{Gupta14} regarding the algebraic rank of the quadratic polynomials appearing at the bottom of such identically zero circuits. 
We next explain the conjectures of \cite{DBLP:journals/iandc/BeeckenMS13,Gupta14} and their relation to Sylvester-Gallai type theorems.

\paragraph{Sylvester-Gallai type theorems and PIT:}

Many of the algorithms mentioned above \cite{DBLP:journals/siamcomp/DvirS07,DBLP:conf/focs/KayalS09,DBLP:journals/siamcomp/KarninMSV13,DBLP:journals/combinatorica/SarafV18,DBLP:journals/iandc/BeeckenMS13,DBLP:journals/toc/0001S17} work by first bounding some algebraic quantity related to the model and then using variable-reduction, to reduce the number of variables in the circuit to the bound of the relevant algebraic quantity.
For example, for depth-$3$ circuits, Kayal and Saraf \cite{DBLP:conf/focs/KayalS09} obtained improved bounds on the \emph{linear rank} the linear functions appearing at the bottom of identically zero circuits, via a colored version of the Sylvester-Gallai theorem due to Edelstein and Kelly \cite{EdelsteinKelly66}. 

Recall that the Sylvester-Gallai theorem asserts that if a finite set of points in $\R^n$ has the property that every line passing through any two points in the set also contains a third point in the set, then all the points in the set are colinear. Kelly extended the theorem to points in $\C^n$ and proved that if a finite set of points satisfy the Sylvester-Gallai condition then  the points in the set are coplanar. Edelstein and Kelly proved that if we have $k>2$ disjoint sets of point such that every line that intersects any two of the sets must also intersect a third set, then there is a $3$-dimensional affine space containing all the points in all the sets (see \autoref{thm:ek-not-disjoint} for an extension). These theorems can also be stated algebraically as results concerning linear forms rather than points (where the condition that a line contains three points is replaced with the condition that three forms are linearly dependent). 

\sloppy To understand the connection to PIT consider the PIT problem for homogeneous\footnote{When studying the PIT problem we may assume without loss of generality that the circuit is homogeneous. See e.g. \cite{DBLP:journals/fttcs/ShpilkaY10}.} $\Sigma^{[3]}\Pi^{[d]}\Sigma$ circuits in $n$  variables. Such circuits compute polynomials of the  form
\begin{equation}\label{eq:sps3}
\Phi(x_1,\ldots,x_n) = \prod_{j=1}^{d}\ell_{1,j}(x_1,\ldots,x_n)+\prod_{j=1}^{d}\ell_{2,j}(x_1,\ldots,x_n)+\prod_{j=1}^{d}\ell_{3,j}(x_1,\ldots,x_n)\;.
\end{equation}
If $\Phi$ computes the zero polynomial then for every $j,j'\in[d]$.
$$\prod_{i=1}^{d}\ell_{1,i} \equiv 0 \mod \ideal{\ell_{2,j},\ell_{3,j'}}\;.\footnote{By $\ideal{\ell_{2,j},\ell_{3,j'}}$ we mean the ideal generated by $\ell_{2,j}$ and $\ell_{3,j'}$ (see \autoref{sec:prelim}).}$$
This means that the sets $\cT_i = \{\ell_{i,1}, \ldots, \ell_{i,d}\}$ satisfy the conditions of the Edelstein-Kelly theorem for sets of linear functions. Thus, if $\Phi\equiv 0$ then, assuming that no linear form belongs to all three sets (which is a simple case to handle), we can rewrite the expression for $\Phi$ using only constantly many variables (after a suitable invertible linear transformation). This easily leads to an efficient PIT algorithms for such $\Sigma^{[3]}\Pi^{[d]}\Sigma$ identities. The case of more than three multiplication gates is more complicated but it also satisfies a similar higher dimensional condition.

For depth-$4$ circuits the situation is different. 
As before, homogeneous  $\Sigma^{[3]}\Pi^{[d]}\Sigma\Pi^{[2]}$ circuits compute polynomials of the form
\begin{equation}\label{eq:spsp}
\Phi(x_1,\ldots,x_n) = \prod_{j=1}^{d}Q_{1,j}(x_1,\ldots,x_n)+\prod_{j=1}^{d}Q_{2,j}(x_1,\ldots,x_n)+\prod_{j=1}^{d}Q_{3,j}(x_1,\ldots,x_n)\;,
\end{equation}
where each $Q_{i,j}$ is a homogeneous quadratic polynomial. If we wish to check whether $\Phi \equiv 0$ and try to reason as before then we get 
\begin{eqnarray}\label{eq:quad-pit}
\prod_{j=1}^{d}Q_{1,j}(x_1,\ldots,x_n) = 0 \mod Q_{2,j},Q_{3,j'}.
\end{eqnarray}
However, unlike the linear case it is not clear what  can be concluded now. Indeed, if a product of linear functions vanishes modulo two linear functions, then we know that one function in the product must be in the linear span of those two linear functions. For quadratic polynomials this is not necessarily the case. For example, note that if for a quadratic $Q$ we have that $Q=0$ and $Q + x^2=0$ then also $Q+xy=0$, and, clearly, we can find $Q$ such that $Q+xy$ is not spanned by $Q$ and $Q+x^2$. An even more problematic difference is that it may be the case that \autoref{eq:quad-pit} holds but that no $Q_{1,j}$ always vanishes when, say, $Q_{2,1},Q_{3,1}$ vanish. For example, let 
$$Q_1 = xy+zw \quad,\quad Q_2 = xy-zw \quad,\quad Q_3 = xw \quad,\quad Q_4 = yz.$$
Then, it is not hard to verify that 
$$Q_3 \cdot Q_4 \equiv 0 \mod Q_1,Q_2.$$
but neither $Q_3$ nor $Q_4$ vanish identically modulo $Q_1,Q_2$.

In spite of the above, Beecken et al.  \cite{DBLP:journals/iandc/BeeckenMS13,Gupta14} and Gupta \cite{Gupta14} conjectured that perhaps the difference between the quadratic case and the linear case is not so dramatic. In fact, they suggested that this may be the case for any constant degree and not just for degree $2$. 
Specifically, Beecken et al. conjectured in \cite{DBLP:journals/iandc/BeeckenMS13} that whenever a $\Sigma^{[k]}\Pi\Sigma\Pi^{[r]}$ circuit is identically zero and also \emph{simple} (no polynomial appears in all multiplication gates) and \emph{minimal} (no subset of the multiplication gates sums to zero), then the algebraic rank of the $Q_{i,j}$'s (the polynomials computed by the bottom two layers, as in Equation~\eqref{eq:spsp}) is bounded by  $\poly(r,k)$. 

In \cite{Gupta14} Gupta took a more general approach and stated vast algebraic generalization of Sylvester-Gallai and Edelstein-Kelly type theorems. 
Specifically, Gupta observed that, whenever \autoref{eq:quad-pit} holds, it must be the case that there are four
polynomials in $\{Q_{1,j}\}$ whose product vanishes identically. That is, for every $(j,j')\in[d]^2$ there are $i_{1,j,j'}, i_{2,j,j'},i_{3,j,j'}, i_{4,j,j'}\in [d]$ so that 
$$Q_{1,i_{1,j,j'}}\cdot Q_{1,i_{2,j,j'}} \cdot Q_{1,i_{3,j,j'}}\cdot Q_{1,i_{4,j,j'}} \equiv 0 \mod Q_{2,j},Q_{3,j'}.$$
Gupta then raised the conjecture that whenever this holds, for every $j,j'$ and for every two of the multiplication gates, then it must be the case that the \emph{algebraic} rank of the set $\{Q_{i,j}\}$ is $O(1)$. More generally, Gupta conjectured that this is the case for any fixed number of sets.

\begin{conjecture}[Conjecture 1 in \cite{Gupta14}]\label{con:gupta-general}
Let $\cF_1,\ldots, \cF_k$ be finite sets of irreducible homogeneous polynomials in $\C[x_1,\ldots, x_n]$ of degree $\leq r$ such that $\cap_i \cF_i = \emptyset$ and for every $k-1$ polynomials $Q_1,\ldots,Q_{k-1}$, each from a distinct set, there are $P_1,\ldots,P_c$ in the remaining set such that  whenever $Q_1,\ldots,Q_{k-1}$ vanish then also the product $\prod_{i=1}^{c}P_i$ vanishes. Then, $\text{trdeg}_\C(\cup_i \cF_i) \leq \lambda(k, r, c)$ for some function $\lambda$, where trdeg stands for the transcendental degree (which is the same as algebraic rank).
\end{conjecture}


Note that when $r=1$ we can assume that $c=1$ and therefore, from the Edelstein-Kelly theorem, we have $\lambda(k,1,c) \leq 3$ in this case (and we can replace algebraic rank with linear rank).

We remark that Gupta's conjecture is stronger than the one made by Beecken et al. as every zero $\Sigma^{[k]}\Pi\Sigma\Pi^{[r]}$ circuit gives rise to a structure satisfying the conditions of Gupta's conjecture, but the other direction is not necessarily true.

In \cite{DBLP:conf/stoc/Shpilka19} the second author proved a special case of \autoref{con:gupta-general} by showing that $\lambda(3,2,1)=O(1)$, regardless of the number of variables or polynomials involved.

In an earlier paper \cite{DBLP:journals/corr/abs-2003-05152} we proved a non-colored version of the conjecture for the case $r=2$ and unbounded $c$.   

\begin{theorem}[Theorem 1.7 of \cite{DBLP:journals/corr/abs-2003-05152}]\label{thm:Peleg-Shpilka-SG}
There exists a universal constant $\Lambda$ such that the following holds. 
Let ${\cQ} = \{Q_i\}_{i\in \{1,\ldots,m\}}\subset\C[x_1,\ldots,x_n]$ be a finite set of pairwise linearly independent irreducible polynomials of degree at most $2$. Assume  that, for every $i\neq j$, whenever $Q_i$ and $Q_j$ vanish then so does $\prod_{k\in  \{1,\ldots,m\} \setminus\{i,j\}} Q_k$. Then,  $\dim(\spn{\cQ})\leq \Lambda$.
\end{theorem}

\subsection{Our results}

In this paper we prove a special case of \autoref{con:gupta-general}. Specifically, we prove that $\lambda(3,2,c)=O(1)$, for any $c$ (and by the discussion above it is sufficient to prove that $\lambda(3,2,4)=O(1)$). In fact, we prove a more general statement showing that the \emph{linear rank} can be bounded from above by a constant rather than the algebraic rank. As the algebraic rank is at most the linear rank our result is indeed stronger.  


\begin{theorem}\label{thm:main}
	There exists a universal constant $\Lambda$ such that the following holds. 
Let $\cT_1,\cT_2, \cT_3\subset\C[x_1,\ldots,x_n]$ be finite sets of pairwise linearly independent homogeneous polynomials satisfying the following properties:
\begin{itemize}
\item Each $Q\in\cup_{j\in[3]}\cT_j$ is either irreducible quadratic  or a square of a linear function.
\item Every two polynomials $Q_1$ and $Q_2$ from distinct sets satisfy that whenever they vanish then the product of all the polynomials in the third set  vanishes as well. Equivalently, for every two polynomials $Q_1$ and $Q_2$ from distinct sets the product of all the polynomials in the third set is in the radical of the ideal generated by $Q_1$ and $Q_2$. 
\end{itemize}
 Then,  $\dim(\spn{\cup_{j\in[3]}\cT_j})\leq \Lambda$.
\end{theorem}

This result can be seen as an extension of \autoref{thm:Peleg-Shpilka-SG} to the case of three sets, as in \autoref{con:gupta-general}.

\begin{remark}
The requirement that the polynomials are homogeneous is not essential as homogenization does not affect the property stated in the theorem.
\end{remark}

\begin{remark}
As mentioned before, Claim 11 in \cite{Gupta14} implies that for every two polynomials $Q_1$ and $Q_2$ from distinct sets there is a subset of the third set, $\cK$, such that  $|\cK| \leq 4$, and whenever $Q_1$ and $Q_2$ vanish then so does $\prod_{k\in \cK} Q_k$. Equivalently, $\prod_{k\in \cK} Q_k \in \sqrt{\ideal{Q_1,Q_2}}$, the radical of the ideal generated by $Q_1$ and $Q_2$.
\end{remark}

As an immediate corollary from earlier works (see e.g. \cite{DBLP:journals/iandc/BeeckenMS13,Gupta14}) we obtain the first black-box polynomial time PIT algorithm for $\spsp$ circuits. As mentioned above, prior to our work no subexponential time algorithm was known even in the white-box model.
Recall that a hitting set for a class of circuits is a set of inputs that intersects the set of nonzeros of any nonzero circuit in the class. Thus a hitting set provides certificates for zeroness/nonzeroness of circuits in the class.

\begin{corollary}[PIT for $\spsp$ circuits]\label{cor:main:PIT}
There is an explicit hitting set ${\mathcal H} \subset \C^n$ of size $(nd)^{O(1)}$ for the class of $n$-variate $\Sigma^{[3]}\Pi^{[d]}\Sigma\Pi^{[2]}$ circuits.
\end{corollary}

\subsection{Proof outline}\label{sec:proof-idea}

Our proof, as well as the proof of \cite{DBLP:journals/corr/abs-2003-05152}, follow the blueprint of the proof in \cite{DBLP:conf/stoc/Shpilka19}. The starting point is a theorem classifying the possible cases in which a product of quadratic polynomials belong to the radical ideal generated by two other quadratics. We state here the more general theorem of \cite{DBLP:journals/corr/abs-2003-05152}, that we will use in our proof.

\begin{theorem}[Theorem 1.8 in \cite{DBLP:journals/corr/abs-2003-05152}]\label{thm:structure}
	Let $\{Q_k\}_{k\in \cK},A,B$ be homogeneous polynomials of degree $2$ such that $\prod_{k\in \cK}Q_k \in \sqrt{\ideal{A,B}}$. Then one of the following cases hold:
	\begin{enumerate}[label={(\roman*)}]
		\item There is $k\in \cK$ such that $Q_k$ is in the linear span of $A,B$  \label{case:span}
		
		\item \label{case:rk1}
		There exists a non trivial linear combination of the form $\alpha A+\beta B = c\cdot d$ where  $c$ and $d$ are linear forms.
		
		\item There exist two linear forms $c$ and $d$ such that when setting $c=d=0$ we get that $A,B$ and one of $\{Q_k\}_{k\in \cK}$ vanish. \label{case:2}
	\end{enumerate}
\end{theorem}

%
%
%
%

The theorem guarantees that, unless the linear span of $A$ and $B$ contains one of the polynomials $\{Q_k\}$, $A$ and $B$ are far from being generic, namely, they must span a reducible quadratic or they have a very low rank (as quadratic polynomials). Thus, for sets of polynomials satisfying the requirements of \autoref{thm:main} it must hold that any two polynomials, coming from different sets, have the structure described in \autoref{thm:structure}. 

\begin{remark}\label{rem:terminology}
Following \autoref{thm:structure}, whenever we say that two quadratics $Q_1$ and $Q_2$, from distinct sets, satisfy \autoref{thm:structure}\ref{case:span} we mean that there is a polynomial $Q_3$ from the third set in their linear span. Similarly, when we say that they satisfy \autoref{thm:structure}\ref{case:rk1}  (\autoref{thm:structure}\ref{case:2}) we mean that there is a reducible quadratic in their linear span (they belong to $\ideal{a_1,a_2}$ for two linear forms $a_1,a_2$).
\end{remark}

Given this classification, the analysis in \cite{DBLP:conf/stoc/Shpilka19,DBLP:journals/corr/abs-2003-05152}  is based on which case of the theorem each pair of polynomials satisfy. The main difference is that the analysis in our work is considerably more difficult than the analysis in \cite{DBLP:conf/stoc/Shpilka19,DBLP:journals/corr/abs-2003-05152} as, unlike \cite{DBLP:conf/stoc/Shpilka19}, there is no unique polynomial $P_3$ in the radical of two other polynomials $\sqrt{\ideal{P_1,P_2}}$ but rather a product of polynomials is in the radical.  This leads to much more technical work. Similarly, the result of \cite{DBLP:journals/corr/abs-2003-05152} considered the case of only one set, and handling $3$ sets require new ideas and more work.

Similarly to \cite{DBLP:journals/corr/abs-2003-05152} we first prove that if every polynomial satisfies Case~\ref{case:span} or Case~\ref{case:2} of \autoref{thm:structure} with at least, say, $1/100$ fraction of the  polynomials in the other two sets,\footnote{More accurately, we require that it satisfies this with $1/100$ fraction of the  polynomials in larger set among the other two sets.} then we can bound the dimension of the linear span of the polynomials in $\cup_i \cT_i$. The proof of this case is given in \autoref{sec:easy-case}. At a high level, the proof has three main steps: We first find a subspace of linear forms, $V$, such that all the polynomials that satisfy Case~\ref{case:2} with many other polynomials belong to $\ideal{V}$. Once we achieve this we show that there is a small set $\cI$ so that all of our polynomials are in $\spn{\cI}+\ideal{V}$. Having this structure at hand we can then prove, using techniques similar to  \cite{DBLP:conf/stoc/Shpilka19,DBLP:journals/corr/abs-2003-05152}, that this implies that our set is contained in a small dimensional space. 

As in \cite{DBLP:journals/corr/abs-2003-05152} the most difficult case is when some of the polynomials do not satisfy this (interestingly, this was the easy case in \cite{DBLP:conf/stoc/Shpilka19}). Namely, they satisfy case~\ref{case:rk1} with more than a fraction of $98/100$ of the polynomials in the other two sets.\footnote{Here too we require this only for the larger set among the other two.} Let $Q_0$ be such a polynomial. This implies that many other polynomials are ``close'' to $Q_0$ in the sense that, after rescaling, they can be written as $Q_0 + \ell_1 \cdot \ell_2$ for some linear functions $\ell_1$ and $\ell_2$. This suggests that perhaps we could prove that any two polynomials are ``close'' to each other in the sense that their difference is of rank 
at most $1$, and then maybe this can be used to bound the overall dimension. We don't quite achieve this but we do show that there are at most two polynomials $Q_0$, $P_0$ and a vector space $V$ of linear forms of $\dim(V)\leq \dimV$, such that any polynomial in the three sets can be written as a linear combination of $Q_0$, $P_0$, a polynomial $F$, which is defined over the linear forms in $V$, and a quadratic of rank $1$. This is proved this in \autoref{sec:hard}. The proof of this statement is very technical and, as all of our proofs, is based on case analysis. 

Once we obtain this structure, we use it to prove that our polynomials live in a low dimensional space. The main idea is that if the polynomials in the $j$th set are of the form $\alpha_{i,j} Q_0 + \beta_{i,j} P_0 + F_{i,j}(V) + a_{i,j}\cdot b_{i,j}$, where $a_{i,j}$
and $b_{i,j}$ are linear forms,  then, assuming that no nonzero linear combination of $Q_0$ and $P_0$ has ``low '' rank, we can show that (modulo $V$) the linear functions in the sets $\cS_j=\{a_{i,j},b_{i,j}\}_i$, satisfy the condition of the Edelstein-Kelly theorem, and hence $\dim(\spn{\cup\cS_j})=O(1)$, which implies that there is a constant dimensional space of quadratics containins all our polynomials. We prove this in \autoref{sec:special}. 


As mentioned above, the steps in our proof are similar to the steps in the proofs of Theorem 1.8 of  \cite{DBLP:conf/stoc/Shpilka19} and Theorem 1.7 of \cite{DBLP:journals/corr/abs-2003-05152} and the arguments and ideas that we use have similar flavor to those used there. This is not very surprising as all these works rely on case analysis based on \autoref{thm:structure}.
The main difference between our proof and these earlier proofs (and also between \cite{DBLP:conf/stoc/Shpilka19} and \cite{DBLP:journals/corr/abs-2003-05152}) is that we require much more technical work to obtain each step. Perhaps surprisingly, except of relying on \autoref{thm:structure} of \cite{DBLP:journals/corr/abs-2003-05152}, and on some basic properties of quadratic polynomials, we could not use any of the claims proved there (and similarly, \cite{DBLP:journals/corr/abs-2003-05152} did not rely on claims from \cite{DBLP:conf/stoc/Shpilka19}).
On the positive side, those who are well acquainted with the proof of \cite{DBLP:journals/corr/abs-2003-05152} will note that in some cases we managed to simplify some of the arguments and unify them. Specifically, the main result of \autoref{sec:special}, \autoref{prop:P,Q_0,C[V],rank1},  captures most of the what is needed in order to obtain the result of \cite{DBLP:journals/corr/abs-2003-05152}, and its proof is simpler with significantly less case analysis than in \cite{DBLP:journals/corr/abs-2003-05152} (and it would have been even simpler had there been just one set to consider instead of three).

\subsection{Conclusions and future research}\label{sec:discussion}

In this work we solved \autoref{con:gupta-general} for the case $k=3$ and $r=2$. As a consequence we obtained the first polynomial time black-box  deterministic algorithms for testing identities of  $\Sigma^{[3]}\Pi^{[d]}\Sigma\Pi^{[2]}$ circuits. However, many questions are still left open. For example, extending our result for larger values of $k$ or $r$ is an intriguing open problem. We suspect that increasing $k$ may be harder than increasing $r$, but right now both questions are open.


Another interesting direction is proving robust versions of the results in this work and in \cite{DBLP:conf/stoc/Shpilka19,DBLP:journals/corr/abs-2003-05152}. For example, the following problem is still open.
\begin{problem}
Let $\delta \in (0,1]$. Can we bound the linear dimension (as a function of $\delta$) of a set of polynomials $Q_1,\ldots ,Q_m\in \CRing{x}{n}$ that satisfy the following property: For every $i\in [m]$ there exist at least $\delta m$ values of $j \in [m]$ such that for each such $j$ there is $\cK_j \subset [m]\setminus\{i,j\}$, satisfying $\prod_{k\in \cK_j} Q_k \in \sqrt{\ideal{Q_i,Q_j}}$.
\end{problem}



We only considered polynomials over the complex numbers in this work. However, we believe (though we did not check the details) that a similar approach should work over positive characteristic as well. Observe that over positive characteristic we expect the dimension of the set to scale like $O(\log |\cQ|)$, as for such fields a weaker version of Sylvester-Gallai theorem holds (see Corollary $1.3$ in \cite{DBLP:journals/combinatorica/BhattacharyyaDS16}).

\subsection{Organization}
The paper is organized as follows. \autoref{sec:prelim} contains our notation, some basic facts regarding quadratic polynomials and the tool of projection and its affect on quadratics.  In \autoref{sec:robust-EK} discuss the Sylvester-Gallai theorem and a theorem of Edelstein and Kelly, and state some variants and extensions of them that we will use in our proof (we give the proofs in \autoref{sec:EK-proofs}). 
The proof of \autoref{thm:main} is given in \autoref{sec:proof}. As explained above the proof has three main cases each is handled in a different subsection (\ref{sec:easy-case}, \ref{sec:special} and \ref{sec:hard} ). In \autoref{sec:proof} we give a more detailed exposition of the structure of the proof.

\section{Preliminaries}\label{sec:prelim}

In this section we explain our notation and present some basic algebraic preliminaries.

\subsection{Notation}

We will use the following notation. Greek letters $\alpha, \beta,\ldots$ denote scalars from $\C$.  Non-capitalized letters $a,b,c,\ldots$  denote linear forms and $x,y,z$ denote variables (which are also linear forms). Bold faced letters denote vectors, e.g. $\vx=(x_1,\ldots,x_n)$ denotes a vector of variables, $\vaa=(\alpha_1,\ldots,\alpha_n)$ is a vector of scalars, and $\vec{0} = (0,\ldots,0)$ the zero vector. We sometimes do not use a boldface notation for a point in a vector space if we do not use its structure as a vector. Capital letters such as $A,Q,P$  denote quadratic polynomials whereas $V,U,W$ denote linear spaces. Calligraphic letters $\cal I,J,F,Q,T$  denote sets. For a positive integer $n$ we denote $[n]=\{1,2,\ldots,n\}$. 

\subsection{Facts from algebra}

We denote with $\CRing{x}{n}$ the ring of $n$-variate polynomials over $\C$. An \emph{Ideal} $I\subseteq \CRing{x}{n}$ is an abelian subgroup that is closed under multiplication by ring elements. For $\calS \subset \CRing{x}{n}$, we  denote with $\ideal{\calS}$, the ideal generated by $\calS$, that is, the smallest ideal that contains $\calS$. For example, for two polynomials $Q_1$ and $Q_2$, the ideal $\ideal{Q_1,Q_2}$ is the set $\CRing{x}{n}Q_1 + \CRing{x}{n}Q_2$. For a linear subspace $V$, we have that $\ideal{V}$ is the ideal generated by any basis of $V$. The \emph{radical} of an ideal $I$, denoted by $\sqrt{I}$, is the set of all ring elements, $r$, satisfying that for some natural number $m$ (that may depend on $r$), $r^m \in I$. Hilbert's Nullstellensatz  implies that, in $\C[x_1,\ldots,x_n]$, if a polynomial $Q$ vanishes whenever $Q_1$ and $Q_2$ vanish, then $Q\in \sqrt{\ideal{Q_1,Q_2}}$ (see e.g. \cite{CLO}). We shall often use the notation $Q\in \sqrt{\ideal{Q_1,Q_2}}$ to denote this vanishing condition. For an ideal $I\subseteq \CRing{x}{n}$ we denote by $\CRing{x}{n} /I$ the \emph{quotient ring}, that is, the ring whose elements are the cosets of $I$ in $\CRing{x}{n}$ with the proper multiplication and addition operations. For an ideal $I\subseteq \CRing{x}{n}$ we denote the set of all common zeros of elements of $I$ by $Z(I)$.
An ideal $I$ is called \emph{prime} if for every $f$ and $g$ such that $fg\in I$ it holds that either $f\in I$ or $g\in I$.  We next present basic facts about prime ideals that are used throughout the proof.
\begin{fact}
\begin{enumerate}{}
	\item If $F$ is an irreducible polynomial then $\ideal{F}$ is a prime ideal.
	\item For linear forms $\MVar{a}{k}$ the ideal $\ideal{\MVar{a}{k}} = \ideal{\spn{\MVar{a}{k}}}$ is prime.
	\item If $I$ is a prime ideal then $\sqrt{I}=I$.
\end{enumerate}
\end{fact}

For $V_1,\ldots,V_k$ linear spaces, we use $\sum_{i=1}^k V_i$ to denote the linear space $V_1 + \ldots + V_k$. For two nonzero polynomials $A$ and $B$ we denote $A\sim B$ if $B \in \spn{A}$. For a space of linear forms $V = \spn{\MVar{v}{\Delta}}$, we say that a polynomial $P \in \CRing{x}{n}$ depends only on $V$  if the value of $P$ is determined by the values of the linear forms $v_1,\ldots,v_\Delta$. More formally, we say that $P$ depends only on $V$ if there is a $\Delta$-variate polynomial $\tilde{P}$ such that $P \equiv \tilde{P}(v_1,\ldots,v_\Delta)$. We denote by  $\C[V]\subseteq \CRing{x}{n}$ the subring of polynomials that depend only on $V$. Similarly we denote by $\C[V]_2\subseteq \CRing{x}{n}$, the linear subspace of all homogeneous quadratic polynomials that depend only on $V$. 

Another notation that we will use throughout the proof is congruence modulo linear forms.
\begin{definition}\label{def:mod-form}
Let $V\subset \CRing{x}{n}$ be a space of linear forms, and $P,Q\in \CRing{x}{n}$. We say that $P\equiv_V Q$ if $P-Q \in \ideal{V}$. 
\end{definition}

We end with a simple observation that follows immediately from the fact the quotient ring $\CRing{x}{n}/{\langle V\rangle}$is a unique factorization domain.
\begin{observation}\label{fact:ufd}
Let $V\subset \CRing{x}{n}$ be a space of linear forms and $P,Q\in \CRing{x}{n}$. If $P = \prod_{k=1}^t P_k$, and $Q = \prod_{k=1}^t Q_k$ satisfy that for all $k$, $P_k$ and $Q_k$ are irreducible in $\CRing{x}{n}/{\langle V\rangle}$, and $P \equiv_V Q \not\equiv_V 0$ then, up to a permutation of the indices, $P_k\equiv_V Q_k$ for all $k\in [t]$.
\end{observation}

When we factorize polynomials modulo a linear space of linear forms, we use this observation but do not refer it.

\subsubsection{Rank of quadratic polynomials}\label{sec:rank}
 
We next give some facts regarding quadratic polynomials. Many of these facts already appeared in  \cite{DBLP:journals/corr/abs-2003-05152},

\begin{definition}\label{def:rank-s}
For a homogeneous quadratic polynomial  $Q$ we denote with $\rank_s(Q)$\footnote{In some recent works  this was defined as $\text{algebraic-rank}(Q)$ or $\text{tensor-rank}(Q)$, but as those notions might have different meanings we decided to continue with the notation of \cite{DBLP:journals/corr/abs-2003-05152}.} the minimal $r$ such that there are $2r$  linear forms $\{a_k\}_{k=1}^{2r}$ satisfying $Q=\sum_{k=1}^r a_{2k}\cdot a_{2k-1}$. We call such representation a \emph{minimal representation} of $Q$.
\end{definition}

This is a slightly different definition than the usual one for the rank of a quadratic form,\footnote{The usual definition says that $\rank(Q)$ is the minimal $t$ such that there are $t$ linear forms $\{a_k\}_{k=1}^{t}$, satisfying $Q=\sum_{k=1}^t a_k^2$.}
but it is more suitable for our needs. We note that a quadratic $Q$ is irreducible if and only if $\rank_s(Q)>1$. The next claim shows that a minimal representation is unique in the sense that the
space spanned by the linear forms in it is unique.

\begin{claim}[Claim 2.13 in \cite{DBLP:journals/corr/abs-2003-05152}]\label{cla:irr-quad-span}
Let $Q$ be a homogeneous quadratic polynomial.  Let $Q=\sum_{i=1}^{r}a_{2i-1}\cdot a_{2i}$ and $Q = \sum_{i=1}^{r}b_{2i-1}\cdot b_{2i}$ be two different minimal representations of $Q$. Then $\spn{\MVar{a}{2r}} =\spn{\MVar{b}{2r}}$.
\end{claim}

This claim allows us to define the notion of \textit{minimal space} of a quadratic polynomial $Q$, which we shall denote $\MS(Q)$.
\begin{definition}\label{def:MS}
Let Q be a quadratic polynomial. Assume that $\rank_s(Q) = r$, and let $Q = \sum \limits_{i=1}^r a_{2i-1}a_{2i}$ be some minimal representation of $Q$.
We denote $\MS(Q)\eqdef \spn{\MVar{a}{2r}}$. 

For a set $\{Q_i\}_{i=1}^{k}$ of quadratic polynomials we denote $\MS(\MVar{Q}{k}) = \sum \limits_{i=1}^k \MS(Q_i)$.
\end{definition}
\autoref{cla:irr-quad-span} shows that the minimal space is well defined. The following fact is easy to verify.

\begin{fact}\label{cor:containMS}
Let $Q=\sum_{i=1}^{m}a_{2i-1}\cdot a_{2i}$ be a homogeneous quadratic polynomial, then $\MS(Q)\subseteq \spn{\MVar{a}{2m}}$.
\end{fact}



\begin{claim}[Claim 2.16 in \cite{DBLP:journals/corr/abs-2003-05152}]\label{cla:rank-mod-space}
Let $Q$ be a homogeneous quadratic polynomial with $\rank_s(Q)=r$, and let $V \subset \CRing{x}{n}$ be a linear space of linear forms such that $\dim(V)=\Delta$. Then $\rank_s(Q|_{V=0})\geq r-\Delta$.
\end{claim}

\begin{claim}[Claim 2.17 in \cite{DBLP:journals/corr/abs-2003-05152}]\label{cla:ind-rank}
Let $P_1\in \CRing{x}{k}$, and $P_2 = y_1y_2\in \CRing{y}{2}$. Then $ \rank_s (P_1+P_2) = \rank_s(P_1) + 1$. Moreover, $y_1,y_2 \in \MS(P_1+P_2).$
\end{claim}


\begin{corollary}[Corollary 2.18 in \cite{DBLP:journals/corr/abs-2003-05152}]\label{cla:intersection}
Let $a$ and $b$ be linearly independent linear forms. Then, if $c,d,e$ and $f$ are linear forms such that $ab+cd=ef$ then $\dim(\spn{a,b}\cap \spn{c,d})\geq 1$.
\end{corollary}

\begin{claim}[Claim 2.19 in \cite{DBLP:journals/corr/abs-2003-05152}]\label{cla:rank-2-in-V}
Let $a,b,c$ and $d$ be linear forms, and $V$ be a linear space of linear forms. Assume $\{\vec{0}\} \neq \MS(ab-cd) \subseteq V$ then $\spn{a,b} \cap V\neq \{\vec{0}\}$. 
\end{claim}

\begin{definition}\label{def:proj}
	Let $a$ be a linear form and $V\subseteq \CRing{x}{n}_1$ a linear subspace of linear forms. We denote by ${V^\perp}(a)$ the projection of $a$ to $V^\perp$ (e.g., by identifying each linear form with its vector of coefficients). We also extend this definition to linear spaces: $V^{\perp}(\spn{\MVar{a}{k}}) = \spn{{V^\perp}(a_1),\ldots, {V^\perp}(a_k)}$.
\end{definition}

\begin{claim}\label{cla:lin-rank-r-U}
	Let $Q,Q'$ be quadratic polynomials, and $U$ be a linear space of linear forms. Let $r \in \N$ be a constant. Then, there exists a linear space of linear forms, $V$, of dimension at most $8r$, such that for every $P\in \C[U]_2$ and every linear combination $\alpha Q+\beta Q' +P$ satisfying $\rank_s(\alpha Q+\beta Q' +P) \leq r$ it holds that $\MS(\alpha Q+\beta Q'+P)\subseteq V+U$.
\end{claim} 
\begin{proof}
	If there are $T,T' \in \C[U]_2$ such that $\rank_s(Q-T),\rank_s(Q'-T') \leq 2r$ then let $V = \MS(Q-T) + \MS(Q'-T')$ and the statement clearly holds.
	Thus, assume without loss of generality that for every $T \in \C[U]_2$, $\rank_s(Q-T) > 2r$.
	Let $A_1 = \alpha_1 Q+ \beta_1 Q' + P_1$ satisfy $\rank_s(A_1) \leq r$. Set $V = \MS(A_1)$. If $V$ does not satisfy the statement then let $A_2 = \alpha_2 Q+ \beta_2 Q' + P_2$ be such that $\rank_s(A_2) \leq r$ and $\MS(A_2) \not \subseteq V+U$. In particular, the vectors  $(\alpha_1, \beta_1)$ and $(\alpha_2, \beta_2)$ are linearly independent. Hence, $Q\in \spn{A_1,A_2, P_1,P_2}$. Consequently, there is $T\in \spn{P_1,P_2}\subseteq \C[U]_2$ such that $\rank_s(Q-T) \leq 2r$, in contradiction.
\end{proof}

\begin{claim}\label{cla:lin-rank-r}
	Let $Q,Q'$ be quadratic polynomials and let $r \in \N$ be a constant. Then, there exists a linear space of linear forms, $V$, of dimension at most $8r$, such that for every linear combination satisfying $\rank_s(\alpha Q+\beta Q') \leq r$ it holds that $\MS(\alpha Q+\beta Q')\subseteq V$
\end{claim} 

\begin{proof}
This claim follows immediately from \autoref{cla:lin-rank-r-U} with $U = \{\vec{0}\}.$
\end{proof}

\begin{claim}\label{cla:case-rk1-gen}
	Let $P$ be an homogeneous irreducible quadratic polynomial and let $a$ and $b$ be linear forms. Assume that  for some finite $\cI$, $\prod_{i\in\cI} T_i \in \sqrt{\ideal{P,ab}}$. Then either $\rank_s(P) = 2$ and $a \in \MS(P)$ or there is $i\in \cI$ such that $T_i = \alpha P + ac$ for some linear form $c$ and scalar $\alpha \in \C$. 
\end{claim}
\begin{proof}
Consider the ideal $\ideal{P,ab}$. If $P$ remains irreducible after setting $a = 0$ then   
$\ideal{P|_{a=0}}$ is a prime ideal. Hence, 
$\sqrt{\ideal{P|_{a=0}}}=\ideal{P|_{a=0}}$ and thus there is $i\in \cI$ with $T_i|_{a=0}\in \ideal{P|_{a=0}}$. In particular,  $T_i =\alpha P + ac$ for some linear form $c$.

Since $P$ is irreducible we have  that $\rank_s(P) \geq 2$. On the other hand, if $P$ becomes reducible when setting $a = 0$ then  $\rank_s(P|_{a=0})=1$. Therefore it must hold that $\rank_s(P) = 2$ and  $a\in \MS(P)$.
\end{proof}

 
 \ifEK
 
In \cite{DBLP:journals/corr/abs-2003-05152} the following claim was proved.
 
\begin{claim}[Claim 2.20 in \cite{DBLP:journals/corr/abs-2003-05152}]\label{cla:linear-spaces-intersaction }
Let $V = \sum_{i=1}^m V_i$ where $V_i$ are linear subspaces, and for every $i$, $\dim(V_i) = 2$. If for every $i\neq j \in [m]$, $\dim(V_i\cap V_j) = 1$, then either $\dim(\bigcap_{i=1}^m V_i) = 1$ or $\dim(V)=3$. 
\end{claim}

We shall need a colorful version of \autoref{cla:linear-spaces-intersaction }:

\begin{claim}\label{cla:colored-linear-spaces-intersaction}
	Let $m \geq 2$ be an integer. For $i \in [m]$  let $V_i = \sum_{j=1}^{m_i} V^j_i$  where $V^j_i$ are distinct linear subspaces that satisfy that for every $i,j$, $\dim(V^j_i) = 2$. Assume that for every $i\neq i' \in [m]$, $j \in [m_i], j'\in [m_{i'}]$, it holds that $\dim(V^j_i\cap V^{j'}_{i'}) = 1$. Then, there exists $w\neq \vec{0}$ and a linear space $U$, such that $\dim(U)\leq 4$ and for every $i\in [m]$, $j\in [m_i]$ either $w \in V^j_i$ or $V^j_i \subseteq U$.  
\end{claim}
\begin{proof}

We split the proof into two cases:
\begin{itemize}
	\item There exists $i\in [m]$ such that $\cap_{j=1}^{m_i} V^j_i \neq \{\vec{0}\}$.
	
	To ease notation we assume, without loss of generality, that $i=1$. Let $\vec{0} \neq w \in \cap_j V^j_1 $. In addition, denote $V_1^1 = \spn{w,x_1}$ and $V_1^2 = \spn{w,x_2}$.
	 If for every $i \in [m]$ and $j \in [m_i]$, we have that $w \in V_i^j$ then the  statement clearly holds. On the other hand, if $w\notin V_i^j$ then let $0\neq z_1 \in V_1^1 \cap V_i^j$. Thus,  $z_1 = \alpha_1 w + \beta_1 x_1$, for $\beta_1 \neq 0$. Similarly, let $z_2 \in V_1^2 \cap V_i^j$, and so $z_2 = \alpha_2 w + \beta_2 x_2$, where $\beta_2 \neq 0$. As $V_1^1\neq V_1^2$ it follows that $z_1 \notin \spn{z_2}$ and therefore $V_i^j =\spn{z_1,z_2}\subseteq \spn{w,x_1,x_2}$. Thus, the statement holds with $w$ and $U = \spn{w,x_1,x_2}$.
	 
	 \item For every $i\in [m]$, $\cap_{j=1}^{m_i} V^j_i = \{\vec{0}\}$. 
	 
	 Consider $0 \neq w \in V_1^1 \cap V^1_2$ and let $V^1_1 = \spn{w, x_1}$ and $V^1_2= \spn{w,y_1}$.
	 Set $U = \spn{w,x_1,y_1}$. If for every  $V_i^j$ it holds that $w \in V_i^j$ or $V_i^j \subseteq U$, then we are done. Assume then that there is $V_i^j$ such that $w \notin V_i^j$ and $V_i^j \not \subseteq U$. If $i \neq 1,2$ then consider the intersection of $V_i^j$ with $V_1^1$ and with $V_2^1$. Similarly to the previous case, we obtain that $V_i^j \subseteq U$. Thus, we only have to consider the case $i \in \{1,2\}$. Assume without loss of generality that $i=1$ and $j=2$.
	 Let $z_1 \in V_1^2 \cap V_2^1\subseteq U$. Hence, $z_1 = \alpha_1 w + \beta_1 y_1$, where $\beta_1 \neq 0$. It follows that $V_1^2 = \spn{z_1, x_2}$ where $x_2\not\in U$ (as $V_1^2 \not \subseteq U$ ). We now show that $U' = \spn{w,x_1,y_1,x_2}$ satisfies the requirements of the theorem (with $w$ being the special vector).
	 
	Since $\cap_j V_2^j = \{\vec{0}\}$ we can assume without loss of generality that $z_1 \notin V_2^2$. 
	 Let  $z_2 \in V_2^2 \cap V_1^1$ and  $z_3 \in V_2^2 \cap V_1^2$. We have that $z_2 = \alpha_2 w + \beta_2 x_1$, and   $z_3 = \alpha_3 z_1 + \beta_3 x_2$ where $\beta_3 \neq 0$ (since $z_1 \notin V_2^2$). Note that  $z_3 \notin \spn{z_2}$ as otherwise we would have that $x_2\in\spn{w,x_1,z_1}=\spn{w,x_1,y_1}=U$ in contradiction. Hence, $V_2^2=\spn{z_2,z_3} \subseteq  U'$. A similar argument shows that for every $j$, $V_2^j \subseteq \spn{w,x_1,y_1,x_2}=U'$. 
	 
	 We now show a similar result for the spaces in $V_1$.
	 Let $V_1^j$ be such that $w\notin V_1^j$, and let $z_4 \in V_2^1 \cap V_1^j$. Then $z_4 = \alpha_4 w + \beta_4 y_1$ where, $\beta_4 \neq 0$. Let $z_5 \in V_2^2 \cap V_1^j$. Denote $z_5 = \alpha_5 z_2 + \beta_5 z_3$. As $x_2 \notin U= \spn{w,x_1,y_1}$, it follows that $z_5 \notin \spn{z_4}$ and thus $V_1^j=\spn{z_4,z_5} \subseteq \spn{w,x_1,x_2,y_1}$ and the claim holds for $V_1^j$ as well.\qedhere
\end{itemize}
\end{proof}
\fi

\subsubsection{Projection mappings}\label{sec:z-map}

This section collects some facts from \cite{DBLP:journals/corr/abs-2003-05152} concerning projections of linear spaces and the effect on relevant quadratic polynomials.


\begin{definition}[Definition 2.21 of \cite{DBLP:journals/corr/abs-2003-05152}]\label{def:z-mapping}
Let $V = \spn{\MVar{v}{\Delta}}\subseteq \spn{x_1,\ldots,x_n}$ be a $\Delta$-dimensional linear space of linear forms, and let $\{\MVar{u}{{n-\Delta}}\}$ be a basis for $V^\perp$. For $\vaa = (\MVar{\alpha}{\Delta})\in \C^{\Delta}$ we define $T_{\vaa, V} : \CRing{x}{n} \mapsto \C[\MVar{x}{n},z]$, where $z$ is a new variable, to be the linear map given by the following action on the basis vectors: $T_{\vaa, V}(v_i) = \alpha_i z$ and $T_{\vaa, V}(u_i)=u_i$.
\end{definition}

Thus, $T_{\alpha,V}$ projects $V$ to $\spn{z}$ in a random way while keeping the perpendicular space intact. 
Clearly $T_{\vaa, V}$ is a linear transformation, and it defines a ring homomorphism from $\CRing{x}{n}$ to $\C[\MVar{x}{n},z]$ in the natural way. 

\begin{claim}[Claim 2.23 of \cite{DBLP:journals/corr/abs-2003-05152}]\label{cla:res-z-ampping}
	Let $V\subseteq \spn{x_1,\ldots,x_n}$ be a $\Delta$-dimensional linear space of linear forms. Let $F$ and  $G$ be two polynomials that share no common irreducible factor. Then, with probability $1$ over the choice of $\vaa \in [0,1]^{\Delta}$ (say according to the uniform distribution), $T_{\vaa, V}(F)$ and $T_{\vaa, V}(G)$ do not share a common factor that is not a polynomial in $z$.
\end{claim}


\begin{corollary}[Corollary 2.24 of \cite{DBLP:journals/corr/abs-2003-05152}]\label{cla:still-indep}
	Let $V$ be a $\Delta$-dimensional linear space of linear forms. Let $F$ and  $G$ be two linearly independent, irreducible quadratics, such that $\MS(F),\MS(G)\not\subseteq V$. Then, with probability $1$ over the choice of $\vaa \in [0,1]^{\Delta}$ (say according to the uniform distribution), $T_{\vaa, V}(F)$ and $T_{\vaa, V}(G)$ are linearly independent.
\end{corollary}


\begin{claim}[Claim 2.25 of \cite{DBLP:journals/corr/abs-2003-05152}]\label{cla:z-map-rank}
Let $Q$ be an irreducible quadratic polynomial, and $V$  a $\Delta$-dimensional linear space.
Then for every $\vaa \in \C^{\Delta}$,  $\rank_s(T_{\vaa, V}(Q)) \geq \rank_s(Q)-\Delta$.
\end{claim}

\begin{claim}[Claim 2.26 of \cite{DBLP:journals/corr/abs-2003-05152}]\label{cla:z-map-dimension}
Let $\cQ$ be a set of quadratics, and $V$ be a $\Delta$-dimensional linear space. Then, if there are linearly independent vectors, $\{\vaa^1,\dots, \vaa^{\Delta}  \}\subset \C^{\Delta}$ such that for every $i$,\footnote{Recall that $\MS(T_{\vaa^i,V}(\cQ))$ is the space spanned by $\cup_{Q\in\cQ}\MS(T_{\vaa^i,V}(\cQ))$.} $\dim(\MS(T_{\vaa^i,V}(\cQ)))\leq \sigma$ then $\dim(\MS(\cQ))\leq (\sigma+1) \Delta$.
\end{claim}
\ifEK

\section{Sylvester-Gallai theorem and some of its variants}
\label{sec:robust-EK}

In this section we give the formal statements of the Sylvester-Gallai and Edelstien-Kelly theorems, and present some of their extensions that we use in this work.

\begin{definition}
	Let $\{v_1,\ldots ,v_m\}$ be a set of distinct points in $\R^n$ or $\C^n$. We call a line that intersects the set at exactly two points an \emph{ordinary line}. 
\end{definition}

\begin{theorem}[Sylvester-Gallai theorem]\label{thm:SG}
	If $m$ distinct points $v_1,\ldots ,v_m$ in $\R^n$ are not collinear, then they define at least one ordinary line.
\end{theorem}

\begin{theorem}[Kelly's theorem]\label{thm:Kelly}
	If $m$ distinct points $v_1,\ldots ,v_m$ in $\C^n$ are not coplanar, then they define at least one ordinary line.
\end{theorem}

The robust version of the Sylvester-Gallai theorem was stated and proved in \cite{barak2013fractional,DBLP:journals/corr/abs-1211-0330}.
\begin{definition}\label{def:delta-SGConf}
	We say that a set of points $v_1,\ldots ,v_m\in \C^n$ is a \textit{$\delta$-SG configuration} if for every $i\in [m]$ there exists at least $\delta m$ values of $j \in [m]$ such that the line through $v_i,v_j$ contains a third point in the set.
\end{definition}

\begin{theorem}[Robust Sylvester-Gallai theorem, Theorem $1.9$ of \cite{DBLP:journals/corr/abs-1211-0330}]\label{thm:robustSG}
	Let $V = \lbrace v_1,\ldots ,v_m\rbrace \subset \C^n$ be a $\delta$-SG configuration. Then, $\dim(\spn{v_1,\ldots ,v_m}) \leq \frac{12}{\delta}+1$.
\end{theorem}

The following is the colored version of the Sylvester-Gallai theorem that was stated and proved by Edelstein and Kelly  \cite{EdelsteinKelly66}.

\begin{theorem}[Theorem $3$ of \cite{EdelsteinKelly66}]\label{thm:EK}
	Let $\cT_i$, for $i\in [3]$, be disjoint finite subsets of $\C^n$ such that for every $i\neq j$ and any two points $p_1\in \cT_i$ and $p_2 \in \cT_j$ there exists a point $p_3$ in the third set that lies on the line passing through $p_1$ and $p_2$. Then, it must be the case that $\dim(\spn{\cup_i \cT_i})\leq 3$.
\end{theorem}

Next, we state extensions of \autoref{thm:EK} and of a result that was proved in \cite{DBLP:conf/stoc/Shpilka19}. As the proofs are modification of the original proofs in \cite{EdelsteinKelly66,DBLP:conf/stoc/Shpilka19} we only state the theorems here and postpone their proofs to \autoref{sec:EK-proofs}.


\begin{theorem}\label{thm:ek-not-disjoint}
	Let $k\geq 3$ and $\{\MVar{\cS}{k}\}$ be finite sets of points in $\C^n$. Assume that for every $p_i\in \cS_i$ and $p_j\in \cS_j$ such that $p_i\neq p_j$ there exists $t \in [k]\setminus \{i,j\}$ and $p_t\in \cS_t\setminus \{p_i,p_j\}$, such that  $p_t,p_i$ and $p_j$ are colinear. Then $\dim(\cup_{i\in [k]}\cS_i) \leq 3$.
\end{theorem}

Observe that the main difference from the $k$-set version of \autoref{thm:EK} is that we do not require that the set are disjoint, rather that the third point on the line differs from the first two.


\begin{definition}\label{def:partial-EK}
We say that the sets $\cT_1,\cT_2,\cT_3\subset \C^n$ form a partial-$\delta$-EK configuration if for every $i \in [3]$ and $p\in \cT_i$, if $\cT_j$ is the larger set among the other two sets, then
at least $\delta$ fraction of the vectors $p_j\in\cT_j$ satisfy that  $p$ and $p_j$ span some vector in the third set.
\end{definition}

\begin{theorem}[Extension of Theorem 1.9 of \cite{DBLP:conf/stoc/Shpilka19}]\label{thm:partial-EK-robust}
	Let $0<\delta \leq 1$ be any constant. Let $\cT_1,\cT_2,\cT_3\subset\C^n$ be disjoint finite subsets that form a partial-$\delta$-EK configuration.  Then $\dim(\spn{\cup_i \cT_i}) \leq O(1/\delta^3)$.
\end{theorem}

Finally, we state equivalent algebraic versions of \autoref{thm:EK}. The proofs 
follow immediately from \autoref{thm:ek-not-disjoint}. For the simple translation from points to vectors and to linear forms see Remark 2.7 in \cite{DBLP:conf/stoc/Shpilka19}.  We shall refer to each of Theorems~\ref{thm:ek-not-disjoint}, \ref{thm:ek-not-disjoint-vec} and \ref{thm:ek-not-disjoint-linforms} as the Edelstien-Kelly theorem. We shall also refer to sets of points/vectors/linear forms that satisfy the conditions of the relevant theorem as satisfying the condition of  the Edelstien-Kelly theorem.

\begin{theorem}\label{thm:ek-not-disjoint-vec}
	Let $k\geq 3$ and $\{\MVar{\calL}{k}\}$ be finite sets of vectors in $\C^n$. Assume that for every $v_i\in \calL_i$ and $v_j\in \calL_j$ such that $v_i\not \sim v_j$ there exists $t \in [k]\setminus \{i,j\}$ and $v_t\in \calL_t\setminus (\spn{v_i}\cup \spn{v_j})$, such that  $v_t \in \spn{v_i,v_j}$. Then $\dim(\cup_{i\in [k]}\calL_i) \leq 4$.
\end{theorem}

\begin{theorem}\label{thm:ek-not-disjoint-linforms}
	Let $k\geq 3$ and $\{\MVar{\calP}{k}\}$ be finite sets of linear forms in $\CRing{x}{n}$. Assume that for every $\ell_i\in \calP_i$ and $\ell_j\in \calP_j$ such that $\ell_i\not \sim \ell_j$ there exists $t \in [k]\setminus \{i,j\}$ and $\ell_t\in \calP_t\setminus (\spn{\ell_i}\cup \spn{\ell_j})$, such that  $\ell_t \in \spn{\ell_i,\ell_j}$. Then $\dim(\cup_{i\in [k]}\calP_i) \leq 4$.
\end{theorem}

\fi

\ifEK
\section{Proof of Theorem~\ref{thm:main}}\label{sec:proof}



Let $\calL_i$ be the set of all squares in $\cT_i$ and let $\cQ_i$ be the remaining irreducible quadratics. Thus, $\cT_i = \cQ_i \cup \calL_i$. Denote $|\cQ_i| = m_i$ and $|\calL_i| = r_i$. We also denote
\begin{equation}\label{eq:Lj}
\calL_j = \{ a_{j,i}^2 \mid m_j+1\leq i\leq m_j+r_j\}\;.
\end{equation}
Let $\delta =\frac{1}{100}$. 
The following sets will be the basis for the case analysis: \begin{equation}\label{eq:P^{(i)}}
	\calP^{\ref{case:span}}_1 =\left\lbrace P \in \cQ_1 \;\middle|\;
	\begin{tabular}{@{}l@{}}
		At least $\delta$ fraction of the polynomials in the {\bf larger} set \\ among $\cQ_2$  and $\cQ_3$  satisfy \autoref{thm:structure}\ref{case:span}, but not\\ \autoref{thm:structure}\ref{case:rk1}, with P
	\end{tabular}
	\right\rbrace
\end{equation}
	and
\begin{equation}\label{eq:P^{(iii)}}
	\calP^{\ref{case:2}}_1 =\left\lbrace P \in \cQ_1 \;\middle|\;
	\begin{tabular}{@{}l@{}}
	P satisfies \autoref{thm:structure}\ref{case:2} with at least a $\delta$ fraction \\ of the polynomials in {\bf one} of the sets
	 $\cQ_2$ or $\cQ_3$
	\end{tabular}
	\right\rbrace\;.
\end{equation}
	
We define the sets $\calP^{\ref{case:span}}_2,\calP^{\ref{case:span}}_3,\calP^{\ref{case:2}}_2,\calP^{\ref{case:2}}_3$ analogously.

\begin{remark}\label{rem:lin-only-2}
Our proof heavily relies on \autoref{thm:structure}. This theorem speaks about a pair of polynomials $P$ and $Q$. Whenever one of them is a square of a linear function we shall always assume/say that they satisfy case~\ref{case:rk1} of \autoref{thm:structure}. Note that even if they satisfy cases~\ref{case:span} or~\ref{case:2} of the theorem it is still true that they satisfy case~\ref{case:rk1} as well. Namely, saying that  $P,Q$ do not satisfy \autoref{thm:structure}\ref{case:rk1}, in particular implies that neither polynomials is a square of a linear function.
\end{remark}

The proof of \autoref{thm:main} is organized as follows. In \autoref{sec:easy-case} we deal with the case where for every $j\in[3]$, $\cQ_j = \calP^{\ref{case:span}}_j\cup \calP^{\ref{case:2}}_j$. 
In \autoref{sec:special} we concentrate on a special case that will play an important role in the proof of the theorem for the case not covered in  \autoref{sec:easy-case}.
Finally, in \autoref{sec:hard} we  handle the case that was not covered by our previous arguments.  

\subsection{For every $j\in [3]$, $\cQ_j = \calP^{\ref{case:span}}_j\cup \calP^{\ref{case:2}}_j$}\label{sec:easy-case}

Assume that for every $j\in[3]$, $\cQ_j = \calP^{\ref{case:span}}_j\cup \calP^{\ref{case:2}}_j$. For our purposes, we may further assume that $\calP^{\ref{case:span}}_j\cap \calP^{\ref{case:2}}_j = \emptyset$ by setting $\calP^{\ref{case:span}}_j = \calP^{\ref{case:span}}_j \setminus \calP^{\ref{case:2}}_j$.
The proof of \autoref{thm:main} for this case consists of the following steps:
\begin{enumerate}
	\item \label{step:1} We first prove the existence of a constant dimensional vector space of linear forms, $V$, such that each $\calP^{\ref{case:2}}_j$ is contained in $\ideal{V}$. This is proved in \autoref{lem:VforC}.
	\item \label{step:2} Next, we find a small set of polynomials $\cI'$ such that  $\cup_{j\in[3]}\cQ_j \subset \spn{(\cup_{j\in[3]}\cQ_j\cap \ideal{V})\cup \cI'}$. This is proved in \autoref{cor:I'}.
	\item \label{step:3} The last step is bounding the dimension of $\cup_{j\in[3]}\cT_j$ given that  $\cup_{j\in[3]}\cQ_j \subset \spn{(\cup_{j\in[3]}\cQ_j\cap \ideal{V})\cup \cI'}$.
\end{enumerate}

\paragraph{Step~\ref{step:1}:}

\begin{lemma}\label{lem:VforC}
There exists a linear space of linear forms, $V$, such that $\dim(V)=O(1)$ and $\cup_j\calP^{\ref{case:2}}_j \subset \ideal{V}$.
\end{lemma}

To get the intuition behind the lemma we make  the following observation.
\begin{observation}\label{rem:4-2-dim}
If $Q_1\in\cQ_1$ and $Q_2\in\cQ_2$ satisfy \autoref{thm:structure}\ref{case:2} then $\dim(\MS(Q_1)), \dim(\MS(Q_2)) \leq 4$ and $\dim(\MS(Q_1)\cap\MS(Q_2)) \geq 2$. 
\end{observation}
This shows that we have many small dimensional spaces that have large pairwise intersections. It is thus conceivable that such $V$ may exist.

\begin{proof}[Proof of \autoref{lem:VforC}]
We shall prove the existence of a vector space $V_j$ for $\calP^{\ref{case:2}}_j$ and at the end take $V=V_1+V_2+V_3$. 

We construct $V_j$ via an iterative process. To simplify notation we describe the process for $j=1$. The other cases are completely analogous.
Denote 
	\[
\calP^{\ref{case:2}\rightarrow(2)}_1 =\left\lbrace P \in \calP^{\ref{case:2}}_1 \;\middle|\;
\begin{tabular}{@{}l@{}}
$P$ satisfies \autoref{thm:structure}\ref{case:2}  with \\ at least $\delta$ fraction of the polynomials in $\cQ_2$
\end{tabular}
\right\rbrace \;,
\]
and
\[
\calP^{\ref{case:2}\rightarrow(3)}_1 =\left\lbrace P \in \calP^{\ref{case:2}}_1 \;\middle|\;
\begin{tabular}{@{}l@{}}
$P$ satisfies \autoref{thm:structure}\ref{case:2}  with \\ at least $\delta$ fraction of the polynomials in $\cQ_3$
\end{tabular}
\right\rbrace \;.
\]
It clearly holds that $\calP^{\ref{case:2}}_1 = \calP^{\ref{case:2}\rightarrow(2)}_1 \cup \calP^{\ref{case:2}\rightarrow(3)}_1$. 
 
Consider the following process. Set $W_2 = \{\vec{0}\}$ and $\calP^{\ref{case:2}\rightarrow(2)'}_1=\emptyset$. At each step consider any $Q \in \calP^{\ref{case:2}\rightarrow(2)}$ such that $Q\notin \ideal{W_2}$ and update $W_2 \leftarrow \MS(Q) + W_2$, and $\calP^{\ref{case:2}\rightarrow(2)'}_1 \leftarrow\calP^{\ref{case:2}\rightarrow(2)'}_1 \cup \{Q\}$. We repeat this process as long as possible, i.e, as long as $\calP^{\ref{case:2}\rightarrow(2)}\not \subseteq \ideal{W_2}$. 

We next  show that this process terminates after at most $\frac{3}{\delta}$ steps. In particular, $|\calP^{\ref{case:2}\rightarrow(2)'}_1| \leq \frac{3}{\delta}$. It is clear that at the end of the process it holds that $\calP^{\ref{case:2}\rightarrow(2)}_1 \subset \ideal{W_2}$.
 
 \begin{claim}\label{cla:3-case3}
 	Let $Q\in \cQ_2$ and let $\cB\subseteq \calP^{\ref{case:2}\rightarrow(2)'}_1$ be the subset of all polynomials in $\calP^{\ref{case:2}\rightarrow(2)'}_1$ that satisfy \autoref{thm:structure}\ref{case:2} with $Q$.  Then, $|\cB| \leq 3$.
 \end{claim}
 
 \begin{proof}
 	Assume towards a contradiction that $|\cB| \geq 4$. Let $Q_1,Q_2,Q_3$ and $Q_4$ be the first four elements of $\cB$ that where added to $\calP^{\ref{case:2}\rightarrow(2)'}_1$ (in that order). Denote $U=\MS(Q)$ and, for $1\leq i \leq 4$, let  $U_i = U\cap \MS(Q_i)$.
 	
 	As $Q$ satisfies \autoref{thm:structure}\ref{case:2} we have that  $\dim(U) \leq 4$. Furthermore, for every $i$, $\dim(U_i)\geq 2$ (by \autoref{rem:4-2-dim}). As the $Q_i$s were picked by the iterative process, we have that $U_2 \not \subseteq U_1$. Indeed, since $Q_2 \in \ideal{U_2}$, if $U_2 \subseteq U_1$ then after adding $Q_1$ to 
$\calP^{\ref{case:2}\rightarrow(2)'}_1$ we would get that $U_2 \subseteq U_1\subseteq  \MS(Q_1)\subseteq W_2$, in contradiction to $Q_2\in \calP^{\ref{case:2}\rightarrow(2)'}_1$. Similarly we get that
 	$U_3 \not \subseteq U_1 + U_2$ and  $U_4 \not \subseteq U_1+U_3 +U_3$. However, as the next simple claim shows, this is not possible.
 	\begin{claim} \label{cla:3-are-V}
 		Let $V$ be a linear space of dimension $\leq 4$, and let $V_1,V_2,V_3 \subset V$, each of dimension $\geq 2$, such that $V_1\not \subseteq V_2$ and  $V_3\not \subseteq V_2 + V_1$. Then, $V = V_1+V_2+V_3$.
 	\end{claim}
 	\begin{proof}
 		As $V_1\not \subseteq V_2$ we have that $\dim(V_1+V_2)\geq 3$. Similarly we get 
 		$4\leq \dim(V_1+V_2+V_3)\leq \dim(V)=4$.
 	\end{proof}
 	Thus, \autoref{cla:3-are-V} implies that $U=U_1+U_2+U_3$ and in particular, $U_4 \subseteq U_1+U_2+U_3$ in contradiction. This completes the proof of \autoref{cla:3-case3}.
 \end{proof}
 
We continue with the proof of \autoref{lem:VforC}. For $Q_i \in \calP^{\ref{case:2}\rightarrow(2)'}_1$, define 
 \[\cG_i =\left\{Q\in \cQ_2 \mid Q \text{ and }Q_i \text{ satisfiy \autoref{thm:structure}\ref{case:2}}\right\}\;.
 \] 
Since $|\cG_{i}|\geq \delta m_2$, and as by \autoref{cla:3-case3} each $Q\in \cQ_2$ belongs to at most $3$ different sets, it follows by double counting that  $|\calP^{\ref{case:2}\rightarrow(2)'}_1|\leq 3/\delta$.
As in each step of the process we add at most $4$ linearly independent linear forms to $W_2$, we obtain $\dim(W_2)\leq \frac{12}{\delta}$.
 
We can now repeat a similar process to obtain $W_3$ such that $\calP^{\ref{case:2}\rightarrow(3)}\subset \ideal{W_3}$. We now have that $V_1 \eqdef W_2+W_3$ is such that $\calP^{\ref{case:2}}_1= \calP^{\ref{case:2}\rightarrow(2)}\cup\calP^{\ref{case:2}\rightarrow(3)}\subset \ideal{V_1}$ as we wanted.

This completes the proof of \autoref{lem:VforC}.
\end{proof}


%
%
%
%
%
%
%
%
%
%
%
%

\paragraph{Step~\ref{step:2}:}
We would now like to find a small set of polynomials $\cI$ such that $\cup_{j\in[3]}\cQ_j \subset \ideal{V}+\spn{\cI}$.  
This will follow if we could prove that $\calP^{\ref{case:span}}_1,\calP^{\ref{case:span}}_2,\calP^{\ref{case:span}}_3$ form a partial-$\frac{\delta}{3}$-EK configuration as in \autoref{def:partial-EK}. Unfortunately, we do not know how to prove this directly. Instead, we shall describe an iterative process for constructing $\cI$, and prove that when the process terminates we have that either  $\cup_{j\in[3]}\cQ_j \subset \spn{\cup_{j\in[3]}\cQ_j\cap \ideal{V}, \cI}$ or the polynomials that remain in the sets $\calP^{\ref{case:span}}_1,\calP^{\ref{case:span}}_2,\calP^{\ref{case:span}}_3$ form a partial-$\frac{\delta}{3}$-EK configuration.

The intuition behind the next process is as follows: Assume without loss of generality that $m_1\geq m_2\geq m_3$. Consider a polynomial $P\in\calP^{\ref{case:span}}_1$. We know that there are at least $\delta m_2$ polynomials in $\cQ_2$ such that $P$ satisfies \autoref{thm:structure}\ref{case:span} but not \autoref{thm:structure}\ref{case:rk1} with each of them. In particular, for every such $Q_i\in \cQ_2$ there is a polynomial $T_i\in \cQ_3$  such that $T_i \in \spn{Q_i,P}$. Indeed, $T_i \notin \calL_3$ as $P,Q_i$ do not satisfy \autoref{thm:structure}\ref{case:rk1}. 
Now, if $Q_i\in \calP^{\ref{case:2}}_2$ and $T_i\in \calP^{\ref{case:2}}_3$ then it holds that $P \in \spn{\cup_{j\in[3]}\cQ_j\cap \ideal{V}, \cI}$ as well. On the other hand, if for at least $\delta / 3$ of those $Q_i$s it holds that either $Q_i  \in \calP^{\ref{case:span}}_2$ or $T_i  \in \calP^{\ref{case:span}}_3$, then by adding $P$ to $\cI$ we get that a constant fraction of the polynomials from $\calP^{\ref{case:span}}_2 \cup \calP^{\ref{case:span}}_3$ now belongs to $\spn{\cup_{j\in[3]}\cQ_j\cap \ideal{V}, \cI}$. Thus, whenever we find a polynomial not in $\spn{\cup_{j\in[3]}\cQ_j\cap \ideal{V}, \cI}$ we can move it to $\cI$ and get that a constant fraction of remaining polynomials were added to $\spn{\cup_{j\in[3]}\cQ_j\cap \ideal{V}, \cI}$. In particular, we expect the process to terminate after a constant number of steps. We next give a formal description of the process explained above.

\paragraph{The process for constructing $\cI$:}
Set $I = \emptyset$. Let 
\[\cD_i = \left\{ Q\in \cQ_i\mid Q\in \spn{\cup_{j\in[3]}\cQ_j\cap \ideal{V}, \cI}\right\} \quad \text{and}\quad \calP^{\ref{case:span}'}_i = \cQ_i \setminus \cD_i \;.
\]
As long as $\cup_i \calP^{\ref{case:span}'}_i\neq \emptyset$ we do the following until we cannot proceed further: Consider  
$P\in \calP^{\ref{case:span}'}_i$ and let $\cQ_j$ be the larger among the two sets not containing $P$ (e.g., if $P\in \calP^{\ref{case:span}'}_2$ then $j=1$).  
By definition, there are at least $\delta m_j$ polynomials $Q\in \cQ_j$ that satisfy \autoref{thm:structure}\ref{case:span} but not \autoref{thm:structure}\ref{case:rk1} with $P$. Each of these polynomials defines a polynomial $T\in \cQ_k$ such that $T \in \spn{P,Q}$ (if there is more than one such $T$ then pick any of them). We call any such pair $(Q,T)$ a $P$-pair. If more than  $\frac{2}{3}\delta m_j$ of the $P$-pairs thus defined belong to $\calP^{\ref{case:span}'}_j\times \calP^{\ref{case:span}'}_k$ then we move to the next polynomial in $\cup_i \calP^{\ref{case:span}'}_i$. Otherwise, we add $P$ to $\cI$, and update $\cD_1,\cD_2,\cD_3,\calP^{\ref{case:span}'}_1,\calP^{\ref{case:span}'}_2,\calP^{\ref{case:span}'}_3$ accordingly. The process continues until we cannot add any new polynomial to $\cI$.




\paragraph{Analysis:}
First, we claim that this process terminates after at most $18 /\delta$ steps. This will follow from showing that  from each $\cQ_j$ we added at most $6 /\delta$ polynomials to $\cI$.

\begin{claim}\label{cla:2nd-proc}
	At every step at which we added a polynomial $P\in \calP^{\ref{case:span}'}_i$ to $\cI$, at least $\frac{\delta}{3}m_j$  polynomials were moved to $\cD_j\cup \cD_k$ from $\calP^{\ref{case:span}'}_j\cup \calP^{\ref{case:span}'}_k$. In particular, $|\cI \cap \cQ_i|\leq 6 /\delta$.
\end{claim}

\begin{proof}
	We use the notation from the description of the process and let $\cQ_j$ is the larger among the sets not containing $P$.	
	
	By the description of the process, $P$ was added to $\cI$ if at most  $\frac{2}{3}\delta m_j$ of the $P$-pairs $(Q,T)$  belong to $\calP^{\ref{case:span}'}_j\times \calP^{\ref{case:span}'}_k$. As $P\not\in \cD_i$ it follows that at least $\frac{1}{3}\delta m_j$ of those $P$-pairs belong to $(\calP^{\ref{case:span}'}_j \times \cD_k)\cup(\cD_j \times \calP^{\ref{case:span}'}_k)$. Indeed, no $P$-pair $(Q,T)$ belongs to $\cD_j \times \cD_k$ as in that case  $P\in \cD_i$ in contradiction.
		
	After adding $P$ to $\cI$, every $P$-pair $(Q,T)$ such that $(Q,T)\in  (\calP^{\ref{case:span}'}_j \times \cD_k)\cup(\cD_j \times \calP^{\ref{case:span}'}_k)$ will now satisfy $Q,T\in \spn{\cup_{j\in[3]}\cQ_j\cap \ideal{V}, \cI}$. Indeed, $P,Q$ and $T$ satisfy a non trivial linear dependence and since two of the polynomials in the linear combination are in $\spn{\cup_{j\in[3]}\cQ_j\cap \ideal{V}, \cI}$ then so is the third. 
	
To conclude the proof we just have to show that we added many polynomials to $\spn{\cup_{j\in[3]}\cQ_j\cap \ideal{V}, \cI}$.

By definition, if $(Q,T)$ is a $P$-pair then there is no other $T'$ such that  $(Q,T')$ is a $P$-pair. Thus, if $(Q,T)\in \calP^{\ref{case:span}'}_j \times \cD_k$ then, after adding $P$ to $\cI$, $Q$ was added to $D_j$. Consider now a $P$-pair $(Q,T)\in \cD_j \times \calP^{\ref{case:span}'}_k$. 
We claim that there is no other $Q\neq Q'\in \cD_j$ such that $(Q',T)\in \cD_j \times \calP^{\ref{case:span}'}_k$ is a $P$-pair. Indeed, if there was such a $Q'$ then by pairwise independence we would conclude that $P\in \spn{Q,Q'}\subseteq \spn{\cup_{j\in[3]}\cQ_j\cap \ideal{V}, \cI}$, in contradiction. 
It follows that when adding $P$ to $\cI$ at least $\frac{1}{3}\delta m_j$ polynomials were moved from $\calP^{\ref{case:span}'}_j\cup\calP^{\ref{case:span}'}_k$ to $\spn{\cup_{j\in[3]}\cQ_j\cap \ideal{V}, \cI}$. 
To see this, assume that there are $\eta m_j$ $P$-pairs $(Q,T)\in \calP^{\ref{case:span}'}_j \times \cD_k$. Then, since the projection on the first coordinate is unique, at least $\eta m_j$ polynomials from $\calP^{\ref{case:span}'}_j$ were added to $\cD_j$. Similarly, at least $\left(\frac{\delta}{3}-\eta\right)m_j$ polynomials $T\in\calP^{\ref{case:span}'}_k$ were added to $\cD_k$. 
As $|\calP^{\ref{case:span}'}_j\cup \calP^{\ref{case:span}'}_k|\leq 2|\cQ_j|=2m_j$ it follows that we can repeat this for at most $6/\delta$ many polynomials $P\in\cQ_i$. Consequently,  $|\cI\cap \cQ_i |\leq 6/\delta$.
\end{proof}

We next show that if at least one of $\calP^{\ref{case:span}'}_1$ and $\calP^{\ref{case:span}'}_2$ is empty then we have that $\cup_j \cQ_j \subset \spn{\cup_{j\in[3]}\cQ_j\cap \ideal{V}, \cI}$.


\begin{claim}\label{cla:2nd-proc2}
	When the process terminates if one of $\calP^{\ref{case:span}'}_1,\calP^{\ref{case:span}'}_2$ or $\calP^{\ref{case:span}'}_3$ is empty then for every $j\in [3]$, $\cQ_j = \cD_j$ and in particular $\cup_j \cQ_j \subset \spn{\cup_{j\in[3]}\cQ_j\cap \ideal{V}, \cI}$.
\end{claim}

\begin{proof}
	First we note that if for two different indices $\calP^{\ref{case:span}'}_i=\calP^{\ref{case:span}'}_j=\emptyset$ then, as every polynomial in $\calP^{\ref{case:span}'}_k$ is in the span of two polynomials from $\cQ_i= \cD_i$ and $\cQ_j= \cD_j$, it must hold that $\calP^{\ref{case:span}'}_k=\emptyset$ as well. In particular, the claim holds in this case.

	Assume for a contradiction that $\calP^{\ref{case:span}'}_j=\emptyset$ and that for some $i\neq j$, $\calP^{\ref{case:span}'}_i\neq \emptyset$. Let $P\in \calP^{\ref{case:span}'}_i$. Observe that every $P$-pair $(Q,T)\in \cQ_j\times \cQ_k$ that spans $P$, is in $\cD_j\times \calP^{\ref{case:span}'}_k$ (it does not matter which among $\cQ_j$ and $\cQ_k$ is larger). In particular, by the description of the process we must add $P$ to $\cI$ in contradiction to the fact that the process already terminated.
%
%
\end{proof}

We are now ready to show that if at the end of the process the sets $\calP^{\ref{case:span}'}_1,\calP^{\ref{case:span}'}_2$ and $\calP^{\ref{case:span}'}_3$ are not empty then they form a partial EK configuration.

\begin{claim}\label{cla:2nd-proc3}
If none of the sets $\calP^{\ref{case:span}'}_i$ is empty when the process terminates then 
it must hold that $|\calP^{\ref{case:span}'}_1| \geq \frac{2}{3} \delta m_1$ and  $|\calP^{\ref{case:span}'}_2| \geq \frac{2}{3} \delta m_2$.
\end{claim}

\begin{proof}
	Assume towards a contradiction that $0 < |\calP^{\ref{case:span}'}_1| < \frac{2}{3} \delta m_1$. From \autoref{cla:2nd-proc2} we know that $\calP^{\ref{case:span}'}_2 \neq \emptyset$. Let $P\in \calP^{\ref{case:span}'}_2$. As $|\calP^{\ref{case:span}'}_1| < \frac{2}{3} \delta m_1$, the description of the process implies that we must have added $P$ to $\cI$ in contradiction.
	A similar argument shows that it cannot be the case that $0 < |\calP^{\ref{case:span}'}_2| < \frac{2}{3} \delta m_2$.
\end{proof}

To summarize, after the process terminates we have either  $\calP^{\ref{case:span}'}_1 = \calP^{\ref{case:span}'}_2 = \calP^{\ref{case:span}'}_3 = \emptyset$ and then $\cup_{j \in [3]}\cQ_j \subset \spn{\cup_{j\in[3]}\cQ_j\cap \ideal{V}, \cI}$ or $\calP^{\ref{case:span}'}_1 , \calP^{\ref{case:span}'}_2 , \calP^{\ref{case:span}'}_3 \neq \emptyset$ and $ \frac{2}{3} \delta m_1\leq |\calP^{\ref{case:span}'}_1| $ and $ \frac{2}{3} \delta m_2 \leq |\calP^{\ref{case:span}'}_2|$.
In addition, when the process terminates it must be the case that for every polynomial $P\in \calP^{\ref{case:span}'}_1$ at least  $\frac{2}{3}\delta m_2$ of the $P$-pairs are in $\calP^{\ref{case:span}'}_2\times \calP^{\ref{case:span}'}_3$. In other words, $P$ satisfies  \autoref{thm:structure}\ref{case:span} but not \autoref{thm:structure}\ref{case:rk1} with at least $\frac{2}{3}\delta m_2$  polynomials in $\calP^{\ref{case:span}'}_2$. Similarly, every polynomial $P\in \calP^{\ref{case:span}'}_2$ satisfies  \autoref{thm:structure}\ref{case:span} but not \autoref{thm:structure}\ref{case:rk1} with at least $\frac{2}{3}\delta m_1$  polynomials in $\calP^{\ref{case:span}'}_1$ and every polynomial $P\in \calP^{\ref{case:span}'}_3$ satisfies  \autoref{thm:structure}\ref{case:span} but not \autoref{thm:structure}\ref{case:rk1} with at least $\frac{2}{3}\delta m_1$  polynomials in $\calP^{\ref{case:span}'}_1$ .

  We are almost done. To show that $\calP^{\ref{case:span}'}_1,\calP^{\ref{case:span}'}_2$ and $\calP^{\ref{case:span}'}_3$ form a partial-$\frac{2\delta}{3}$-EK configuration we just have to show that 
  $|\calP^{\ref{case:span}'}_1|\geq |\calP^{\ref{case:span}'}_2|\geq |\calP^{\ref{case:span}'}_3|$. While this does not have to be the case, since the sizes of $\calP^{\ref{case:span}'}_1$ and $\calP^{\ref{case:span}'}_2$ remain large it is not hard to show that no matter what is the order of the sizes of the sets, the configuration that we have is a partial-$\delta'$-EK configuration for some $\delta'=\Theta(\delta)$. The only possible exception  is when 
  $|\calP^{\ref{case:span}'}_1|\geq |\calP^{\ref{case:span}'}_3|\geq |\calP^{\ref{case:span}'}_2|$. However, in this case a close inspection of the proof of \autoref{thm:partial-EK-robust}  reveals that the conclusion holds in this case as well. We thus have the following corollary.
  
\begin{corollary}\label{cor:J}
  If $\cup_i \calP^{\ref{case:span}'}_i\neq \emptyset$ then there is a set $\cJ$ of size $|\cJ|= O\left(1/\delta^3 \right)$ such that $\cup_i \calP^{\ref{case:span}'}_i\subset \spn{\cJ}$.		
\end{corollary}

\begin{proof}
	The claim follows immediately from the discussion above and from \autoref{thm:partial-EK-robust}.
\end{proof}

\begin{corollary}\label{cor:I'}
	There exists a set $\cI'$ such that $|\cI'|= O\left(1/\delta^3 \right)$ and $\cup_{j\in[3]}\cQ_j \subset \spn{\cup_{j\in[3]}\cQ_j\cap \ideal{V}, \cI'}$.
\end{corollary}

\begin{proof}
	Let $\cI'=\cI\cup\cJ$ where $\cI$ is the set found in \autoref{cla:2nd-proc} and $\cJ$ is the set guaranteed in \autoref{cor:J}.
\end{proof}

\paragraph{Step~\ref{step:3}:}
We next show how to use \autoref{thm:EK} to bound the dimension of $\cup_{j\in[3]}\cQ_j$ given that  $\cup_{j\in[3]}\cQ_j \subset \spn{(\cup_{j\in[3]}\cQ_j\cap \ideal{V})\cup \cI'}$. To achieve this we introduce yet another iterative process:  Go over all $P\in \cup_{j\in[3]}\cQ_j\setminus \ideal{V}$. For each such $P$, if there is a quadratic polynomial $L$, with $\rank_s(L) \leq 2$, such that  $P + L \in \ideal{V}$, then update $V$ to  $V \leftarrow V+\MS(L)$. Observe that this increases the dimension of $V$ by at most $4$. Also note that as this step increases $\dim\left( \ideal{V}\cap \cup_{j\in[3]}\cQ_j\right)$, we can remove one polynomial from $\cI'$ while still maintaining the property $\cup_{j\in[3]}\cQ_j \subset \spn{\cup_{j\in[3]}\cQ_j\cap \ideal{V}, \cI'}$. We repeat this process until either $\cI'$ is empty or until none of the polynomials in $\cup_{j\in[3]}\cQ_j\setminus \ideal{V}$ satisfies the condition of the process. By the upper bound on $|\cI'|$ the dimension of $V$ grew by at most  $4|\cI'|= O\left( 1/{\delta^3}\right)$ and hence when the process terminates we still have $\dim(V)=O\left( 1/{\delta^3}\right)=O(1)$. 

It is also clear that at each step, $\cup_{j\in[3]}\cQ_j \subset \spn{\cup_{j\in[3]}\cQ_j\cap \ideal{V}, \cI'}$. Finally, when the process terminates, every polynomial  $P \in \cup_{j\in[3]}\cQ_j\setminus \ideal{V}$ satisfies $\rank_s(P) > 2$, even if we set all linear forms in $V$ to zero.

Consider the map $T_{\vaa,V}$ as given in \autoref{def:z-mapping}, for a uniformly random $\vaa\in[0,1]^{\dim(V)}$. Each polynomial in $\cup_{j\in[3]}\cQ_j\cap \ideal{V} $ is mapped to a polynomial of the form  $zb$, for some linear form $b$.  \autoref{cla:rank-mod-space} guarantees that every polynomial in $\cup_{j\in[3]}\cQ_j\setminus \ideal{V}$ still has rank larger than $2$ after the mapping.
Let $$\cB_j= \{b \mid \text{ some polynomial in } \cQ_j\cap \ideal{V}  \text{ was mapped to } zb\} \cup T_{\vaa,V}(\calL_j) \;.$$
By definition $\cB_j$ contains all linear forms that divide some polynomial in $T_{\vaa,V}(\cT_j)$.

We now show that, modulo $z$, the sets $\cB_1,\cB_2$ and $\cB_3$ satisfy the conditions of \autoref{thm:EK}.
Let $b_1\in \cB_i$ and  $b_2\in\cB_j$ be linear forms taken from two different sets such that $b_1 \not\in\spn{z}$ and $b_2 \not\in\spn{z,b_1}$. If no such forms exist then clearly $\dim\left(\cup_k B_k\right)\leq 2$. To simplify notation let us assume without loss of generality that $i=1$ and $j=2$.

As $\cup_{j\in[3]}\cT_j$ satisfies the conditions of \autoref{thm:main},  there are  polynomials $Q_1,\ldots,Q_4 \in \cT_3$  such that $\prod_{i=1}^{4}T_{\vaa,V}(Q_i) \in \sqrt{\ideal{b_1,b_2}}=\ideal{b_1,b_2}$, where the last equality holds as $\ideal{b_1,b_2}$ is a prime ideal.
It follows that, without loss of generality,  $T_{\vaa,V}(Q_4)\in \ideal{b_1,b_2}$. Thus, $T_{\vaa,V}(Q_4)$ has rank at most $2$ and therefore $Q_4\in\calL_k\cup(\cQ_k\cap \ideal{V}) $. Hence, $T_{\vaa,V}(Q_4)$ was mapped to $zb_4$ or to $b_4^2$, for some linear form $b_4$. In particular, $b_4\in\cB_3$. \autoref{cla:res-z-ampping} and \autoref{cla:still-indep} imply that $b_4$ is neither a multiple of $b_1$ nor a multiple of $b_2$. Consequently, $b_4$ depends non-trivially on both $b_1$ and $b_2$. Thus, $\cB_1,\cB_2$ and $\cB_3$ satisfy the conditions of \autoref{thm:EK} modulo $z$. It follows that $\dim(\cup_{j}\cB_j)=O(1)$.

The argument above shows that 
\[ \dim\left( T_{\vaa,V}\left(\cup_{j}\calL_j\cup\left(\cup_{j}(\cQ_j\cap \ideal{V})\right) \right)\right) =O(1) \;.
\] 
\autoref{cla:z-map-dimension} implies that if we denote $U = \spn{\cup_{j}\calL_j\cup(\cup_{j}(\cQ_j\cap \ideal{V})) }$ then  $\dim(U )$ is $O(1)$. As $\cup_{j}\cQ_j \subseteq \spn{(\cup_{j}\cQ_j\cap \ideal{V}) \cup \cI'}$, we obtain that $\dim(\cup_{j}\cT_j)=\dim(\cup_{j}(\cQ_j\cup \calL_j)) = O(1)$, as we wanted to prove.

This concludes the proof of \autoref{thm:main} in the case where for every $j\in[3]$ it holds that $\cQ_j = \calP^{\ref{case:span}}_j\cup \calP^{\ref{case:2}}_j$.

\subsection{Special case of Theorem~\ref{thm:main}}\label{sec:special}

In this section we handle a special case of \autoref{thm:main}. At this point it may not be clear why this case is so important, but when we handle the cases that were not captured by the arguments in \autoref{sec:easy-case} we shall often reduce to it.


The case we deal with here roughly says that \autoref{thm:main} holds when all the polynomials are ``close'' to a special vector space of polynomials. 

%
%

Throughout this subsection we make the following assumption.

\begin{assumption}\label{assume:special-case}
	$V$ is a linear space of linear forms and  $P_0$, $Q_0$ are quadratic polynomials (which can also be identically zero) such that for every nonzero linear combination $0 \neq \alpha P_0 + \beta Q_0$ it holds that $\rank_s(\alpha P_0 + \beta Q_0) > 2(\dim(V)) + \distV$.  
\end{assumption}

\begin{observation}\label{obs:P-Q-ideal}
	Let $V$, $P_0$ and $Q_0$ satisfy \autoref{assume:special-case}.
	If there exists a linear combination of $P_0,Q_0$ and a polynomial $T \in \ideal{V}$ such that $\alpha P_0 + \beta Q_0 + T =R_V + L$ where $\rank_s(L) < \distV$ and $R_V\in \ideal{V}$ then, as every polynomial in $\ideal{V}$ is of $\rank_s \leq \dim(V)$, it holds that $\rank_s(R_V -T + L)< \dim(V)+\distV$ and this implies that $\alpha P_0 + \beta Q_0 =0$. 
\end{observation}

\begin{observation}\label{obs:dint-inc}
	Let $V$, $P_0$ and $Q_0$ satisfy \autoref{assume:special-case}. Let $U$ be a linear space of linear forms such that $\dim(U) \leq 5$. Then every nonzero linear combination $0 \neq \alpha P_0 + \beta Q_0$ satisfies $\rank_s(\alpha P_0 + \beta Q_0) > 2\dim(V+U) + 10$.  
\end{observation}

The following is the main proposition of this section.

\begin{proposition}\label{prop:P,Q_0,C[V],rank1}
	Let $\cup_{j \in [3]}\cT_j = \cup_{j \in [3]}(\cQ_j \cup\calL_j)$ be a set of quadratic polynomials that satisfy \autoref{thm:main}. and $V,P_0$ and $Q_0$ satisfy \autoref{assume:special-case}.
	
	Assume that for every $j\in [3]$ and polynomial $Q_i \in \cQ_j$ one of the following options hold:   $Q_i \in \ideal{V}$ or there are two linear forms $a_{j,i}$ and $b_{j,i}$ and a polynomial $0 \neq Q'_i \in \spn{P_0,Q_0,\C[V]_2}$  such that $Q_i = Q'_i + a_{j,i}b_{j,i}$. 
	Then there exists a linear space of linear forms $V'$, such that $\dim(V')\leq 5 \dim(V)+25$ and $\cup_{j \in [3]}\cT_j\subset \spn{P_0,Q_0,\C[V+V']_2}$.
\end{proposition}

\begin{proof}
	Denote $\dim(V) = \Delta$.
The proof of \autoref{prop:P,Q_0,C[V],rank1} relies on the following lemma.

\begin{lemma}	\label{lem:U-exist}
	Assuming the setting of \autoref{prop:P,Q_0,C[V],rank1}, there exists a linear space of linear forms  $U$, such that $\dim(U)\leq 5$ and:
\begin{itemize}
\item For $j\in [3]$, every polynomial $Q_i \in \cQ_j$ satisfies either $Q_i \in \ideal{V+U}$ or there are linear forms $a_{j,i}$, $v_{j,i}$ where $v_{j,i} \in V+U$, and a polynomial $Q'_i \in \spn{P_0,Q_0,\C[V+U]_2}$  such that $Q_i = Q'_i + a_{j,i}(\epsilon_{j,i} a_{j,i} + v_{j,i})$, for some $\epsilon_{j,i}\in\C$.
\item Every nonzero linear combination $0 \neq \alpha P_0 + \beta Q_0$ satisfies $\rank_s(\alpha P_0 + \beta Q_0 ) >2 \dim(V+U) + 10$.
\end{itemize}  
\end{lemma}


	We postpone the proof of \autoref{lem:U-exist} to \autoref{sec:U-exist} and continue.
	Set $V'= V+U$, where $U$ is the space guaranteed in \autoref{lem:U-exist}. Thus, we now have
	\[ \dim(V')\leq \Delta + 5 \;.\]
	Recall Equation~\eqref{eq:Lj} and denote
	\[\calW_j=\text{span}\left\{ \bigcup \left\{ \MS(Q_i) \mid Q_i \in \cQ_j\cap \ideal{V} \right\}\cup \{a_{j,i}\}_{i\in [m_j+r_j]}\right\}.\]
	 In other words, $\calW_j$ is the space spanned by all the linear forms appearing in polynomials in $\cT_j$ after we remove from them the component coming from $\spn{P_0,Q_0,\C[V']_2}$. The proof of \autoref{prop:P,Q_0,C[V],rank1} will follow if we show that $\dim(\calW_j)=O(1)$ for $j\in [3]$.
	
	
	The idea is to apply a projection mapping $T_{\vaa, V'}$ (as in \autoref{def:z-mapping}) and show that the sets $T_{\vaa, V'}(\calW_j)$ satisfy the conditions of the Edelstein-Kelly theorem. Since we choose $\vaa$ at random we can assume that $T_{\vaa, V'}$ is such that if a polynomial $A$ satisfies $A \in \cup \cQ_i \setminus \spn{P_0,Q_0,\C[V']_2}$ then  $T_{\vaa, V'}(A)\notin \spn{T_{\vaa, V'}(P_0),T_{\vaa, V'}(Q_0),z^2}$, and that the conclusion of \autoref{cla:res-z-ampping} holds.

For $j\in [3]$ and $i\in [m_j+r_j]$ denote $\tilde{a}_{j,i} = T_{\vaa,V'}({a}_{j,i})$. By our choice of $T_{\vaa,V'}$ we get that if $a_{j,i}\not\in V'$ then $\tilde{a}_{j,i}\not\in\spn{z}$.

Consider what happens to a polynomial  $Q_i \in \cQ_j$  after we apply $T_{\vaa,V'}$.	If $Q_i \in \ideal{V'}$ then it was mapped to a polynomial of the form $T_{\vaa,V'}(Q_i)= z\cdot a$. In this case we abuse notation and denote $\tilde{a}_{j,i}=a$. If $Q_i = Q'_i + a_{j,i}(\epsilon_{j,i} a_{j,i} + v_{j,i})$, where  $Q'_i \in \spn{P_0,Q_0,\C[V']_2}$, then $Q_i$ was mapped to $T_{\vaa,V'}(Q_i)=T_{\vaa,V'}(Q'_i)+\tilde{a}_{j,i}(\epsilon_{j,i}\tilde{a}_{j,i} +\beta_{j,i}z)$. Similarly, every $a_{j,i}^2\in\calL_j$ was mapped to $\tilde{a}_{j,i}^2$.
	We next show that the linear forms $\{\tilde{a}_{j,i}\}$ live in a low dimensional space. 
	For this we define the following sets 
\[\cS_j = \left\{ \tilde{a}_{j,i}\notin \spn{z}\right\}\;.\]
	By our choice of $T_{\vaa,V'}$ it holds that for every $j \in [3]$,  \[\cS_j = \emptyset \iff  \cT_j \subseteq \spn{P_0,Q_0, \C[V']_2}.\]
	
	\autoref{lem:P-Q-and-C[V]} that we state next proves \autoref{prop:P,Q_0,C[V],rank1} in the case where for some $i\neq j$, $\cS_i=\cS_j = \emptyset$. 

	\begin{lemma}\label{lem:P-Q-and-C[V]}
	Let $Q_0$, $P_0$ and $V$ be as in \autoref{prop:P,Q_0,C[V],rank1} and let $V'$ satisfy $\dim(V')\leq \dim(V)+5$. 
	 Assume that there are two sets, $\cT_i,\cT_j$ such that $\cT_i, \cT_j \subset \spn{P_0,Q_0,\C[V']_2}$. Then the third set, $\cT_k$ satisfies $\cT_k \subset \spn{P_0,Q_0,\C[V']_2}$ as well.  
\end{lemma}

	We prove \autoref{lem:P-Q-and-C[V]} in \autoref{sec:P-Q-and-C[V]} and continue with the proof of \autoref{prop:P,Q_0,C[V],rank1}.
	The only case left is when there are $i\neq j$ such that $\cS_i,\cS_j \neq \emptyset$. In \autoref{lem:ek-lines} we prove that in this case the sets $\cS_1, \cS_2$ and $\cS_3$  satisfy the conditions of \autoref{thm:EK} (modulo $z$). The proof of the lemma is given in \autoref{sec:ek-lines}.
	
\begin{lemma}\label{lem:ek-lines}
	Let $i\neq j$, $\tilde{a}_{i,t}\in \calS_i$ and  $\tilde{a}_{j,r}\in \calS_j$ be such that $\tilde{a}_{j,r}\notin \spn{\tilde{a}_{i,t},z}$. Then there exists $\tilde{a}_{k,s}\in \calS_k$, where $k\neq i,j$, such that $\tilde{a}_{k,s}\in \spn{\tilde{a}_{i,t}, \tilde{a}_{j,r},z} \setminus (\spn{\tilde{a}_{i,t},z} \cup \spn{\tilde{a}_{j,r},z})$.
\end{lemma}
	
	From \autoref{thm:EK} we get that the dimension of $\cup_{j \in [3]} \calS_j$ is at most $3$. Combining with \autoref{cla:z-map-dimension} we obtain that
	the dimension of the set $\cup_{j \in [3]} \calW_j$ is at most $5\dim(V')\leq 5\Delta + 25$. By letting $V''=V'+ \Sigma_{j \in [3]} \calW_j$, we get that $\dim(V'')\leq \dim(V')+ 5\Delta + 25\leq 6\Delta + 30$, and  that  $\cup_{j \in [3]}\cT_j\subset \spn{P_0,Q_0,\C[V'']_2}$ as claimed. 	
\end{proof}
	
\subsection{Missing proofs}
We now prove all the lemmas that were used in the proof of \autoref{prop:P,Q_0,C[V],rank1}.

\subsubsection{Proof of Lemma~\ref{lem:U-exist}}\label{sec:U-exist}

	Let $j \in [3]$ and $ t \in [m_i]$ we denote (using the notation of \autoref{prop:P,Q_0,C[V],rank1})
	\[V_{i,t} = V^\perp(\spn{a_{i,t}, b_{i,t}})\;.
	\] 
	If $Q_{i,t}\in\ideal{V}$ then we set $V_{i,t} = \{\vec{0}\}$. We consider several cases:

	\paragraph{For every $i\in [3]$ and $t\in[m_i]$, $\dim(V_{i,t})\leq 1$:} In this case the lemma holds with $U = \{\vec{0}\}$.
	
	\paragraph{There exist $t_i \in [m_i]$ and $t_j\in [m_j]$, for $i\neq j$, such that $\dim(V_{i,t_i}) = \dim(V_{j,t_j}) = 2$:} The next claim shows that $V_{i,t_i}$ and $V_{j,t_j}$ have a non trivial intersection.
	
	\begin{claim}\label{cla:a_ib_i-intersection}
	Let $A_{i,t}\in \cQ_i$ and $B_{j,r} \in \cQ_j$ satisfy $A_{i,t} = A'_{i,t} + a_{i,t}b_{i,t}$ and $B_{j,r} = B'_{j,r} + a_{j,r}b_{j,r}$, where $A'_{i,t},B'_{i,t} \in \spn{P_0,Q_0,\C[V]_2}$. Assume that $\dim(V_{i,t}) =\dim(V_{j,r}) =  2$. Then $ V_{i,t}\cap V_{j,r}\neq \{\vec{0}\}$. 
\end{claim}

\begin{proof}
	We first note that proving $ V_{i,t}\cap V_{j,r}\neq \{\vec{0}\}$, is equivalent to proving that $\spn{a_{i,t}, b_{i,t}, V} \cap \spn{a_{j,r}, b_{j,r}}\neq \{\vec{0}\}$.
	The proof depends on  which case of \autoref{thm:structure} $A_{i,t}$ and $B_{j,r}$ satisfy:
	\begin{case}
		\item  In this case  there are $\alpha,\beta \in \C\setminus \{0\}$ and\footnote{Whenever it is clear from the context that $k$ is the index of the third set we shall not mention it explicitly.} $C_{k,s} \in \cT_k$ such that $\alpha A_{i,t} + \beta B_{j,r} = C_{k,s}$. 		
		We continue our analysis based on the structure of $C_{k,s}$. If $C_{k,s} = C'_{k,s} + a_{k,s}b_{k,s}$ for some $C'_{k,s}\in \spn{P_0,Q_0,\C[V]_2}$ (note that this includes the case $C_{k,s}\in\calL_k$) then we get that
		\[\alpha A'_{i,t} + \beta B'_{j,r} - C'_{k,s} = a_{k,s}b_{k,s} - \alpha a_{i,t}b_{i,t} - \beta a_{j,r}b_{j,r} \;.\]
		\autoref{obs:P-Q-ideal} implies that $\alpha A'_{i,t} + \beta B'_{j,r} - C'_{k,s} \in  \C[V]_2$.
		If  $\spn{a_{i,t}, b_{i,t}, V} \cap \spn{a_{j,r}, b_{j,r}} = \{\vec{0}\}$, then by \autoref{cla:ind-rank} it holds that 
		\begin{align*}
				1 &= \rank_s(a_{k,s}b_{k,s}) = \rank_s(\alpha A'_{i,t} + \beta B'_{j,r} - C'_{k,s} + \alpha a_{i,t}b_{i,t} + \beta a_{j,r}b_{j,r} )\\ &\geq \rank_s(\alpha A'_{i,t} + \beta B'_{j,r} - C'_{k,s}) + 2 \geq 2
		\end{align*}		 
		 in contradiction. Thus $\spn{a_{i,t}, b_{i,t},V} \cap \spn{a_{j,r}, b_{j,r}} \neq \{\vec{0}\}$, which is what we wanted to prove.
		Similarly, if $C_{k,s} \in \ideal{V}$ then \[\alpha A'_{i,t} + \beta B'_{j,r} = C_{k,s}   - \alpha a_{i,t}b_{i,t} - \beta a_{j,r}b_{j,r} \;. \]
		As $\rank_s(C_{k,s}   - \alpha a_{i,t}b_{i,t} - \beta a_{j,r}b_{j,r})\leq \dim(V)+2$ it holds that $\alpha A'_{i,t} + \beta B'_{j,r} \in \C[V]_2$ and thus $ \alpha a_{i,t}b_{i,t} + \beta a_{j,r}b_{j,r}\in \ideal{V}$. Consequently, $\spn{a_{i,t}, b_{i,t},V} \cap \spn{a_{j,r}, b_{j,r}} \neq \{\vec{0}\}$.
		
		\item 	
		In this case there are $\alpha,\beta \in \C\setminus \{0\}$ and linear forms $e$ and $f$ such that $\alpha A_{i,t} + \beta B_{j,r} = ef$. As in the previous case it follows that $\alpha A'_{i,t} + \beta B'_{j,r} \in \C[V]_2$, and as before we conclude that $\spn{a_{i,t}, b_{i,t}, V} \cap \spn{a_{j,r}, b_{j,r}} \neq \{\vec{0}\}$.
		
		\item 		
		Here there are linear forms  $e$ and $f$ such that $A_{i,t},B_{j,r} \in \ideal{e,f}$. As $\rank_s(A_{i,t})=\rank_s(B_{j,r}) = 2$, \autoref{obs:P-Q-ideal}, implies that $A'_{i,t},B'_{j,r}\in \C[V]_2$.  Consequently, $\MS(A_{i,t})\subseteq \spn{a_{i,t}, b_{i,t}, V}$ and similarly $\MS(B_{j,r})\subseteq \spn{a_{j,r}, b_{j,r}, V}$.
		If $e,f \in V$ then $A_{i,t},B_{j,r} \in \ideal{e,f}\subseteq \ideal{V}$ an then $a_{i,t}b_{i,t}\in \ideal{V}$ in contradiction. Thus,  without loss of generality, assume $e \notin V$.  Since $e,f \in \MS(A_{i,t}) \cap \MS(B_{j,r})$  we get that  $\spn{a_{i,t}, b_{i,t}, V} \cap \spn{a_{j,r}, b_{j,r}} \neq \{\vec{0}\}$ as we wanted.\qedhere
	\end{case}
\end{proof}
We continue with the proof of \autoref{lem:U-exist}. 
The conclusion of \autoref{cla:a_ib_i-intersection} allows us to use  \autoref{cla:colored-linear-spaces-intersaction} and conclude the existence of $w$ and $ \tilde{U}$ such that $\dim(\tilde{U}) \leq 4$ and for every $i \in [3], t\in [m_i]$ with $\dim(V_{i,t})=2$, either $w \in V_{i,t}$ or $V_{i,t}\subseteq  \tilde{U}$. Set $U = \tilde{U}+\spn{w}$.  It follows that for every $i \in [3]$ and $ t\in [m_i]$, without loss of generality, $b_{i,t} = \epsilon_{i,t} a_{i,t} + v_{i,t}$ for some $v_{i,t} \in V+U$, and, clearly, $\dim(U) \leq 5$. This concludes the proof of \autoref{lem:U-exist} in this case.
	
	\paragraph{There is at most one set with a polynomial whose associated space has dimension $2$:} 
	Assume that $\cQ_k$ is that set. If it contains two polynomials $Q_s,Q_{s'}$ such that $\dim(V_{k,s}) = \dim(V_{k,s'}) = 2$ and $V_{k,s} \cap V_{k,s'} =\{\vec{0}\}$   then \autoref{lem:U-exist} (and in fact also \autoref{prop:P,Q_0,C[V],rank1}) follows from  \autoref{cla:P-Q_0-C[V]-a_i}.

\begin{claim}\label{cla:P-Q_0-C[V]-a_i}
		Assume that there are two sets, $\cT_i,\cT_j$ such that every polynomial $A_t \in \cT_i\cup \cT_j$ is either of the form  $A_t \in \ideal{V}$ or $A_t = A'_t + a_t(\epsilon_t a_t+v_t)$ for $ \epsilon_t \in \C, A'_t \in \spn{P_0,Q_0,\C[V]_2}$ and linear forms $a_t, v_t$, where $v_t \in V$. Assume further that there are polynomials  $C, \tilde{C} \in \cQ_k$ and linear forms $c,d,e,f$ such that $C'=C -cd, \tilde{C}'=\tilde{C}-ef \in  \spn{P_0,Q_0,\C[V]_2}$, $\dim(\spn{c,d}) = 2$, $\spn{c,d}\cap V = \emptyset$, and $\spn{e,f} \cap \spn{V, c,d} = \emptyset$. Denote $U = V + \spn{c,d,e,f}$ then $\cT_i,\cT_j \subset   \spn{P_0,Q_0,\C[U]_2}$. 
\end{claim}

We postpone the proof of \autoref{cla:P-Q_0-C[V]-a_i} and continue with the proof of \autoref{lem:U-exist}.
Observe that \autoref{cla:P-Q_0-C[V]-a_i} gives a subspace $U$ with $\dim(U)= 4$ that satisfies the requirements of \autoref{lem:U-exist}. 

The last case to consider is when there is only one such $V_{k,s}$. In this case we set $U= V_{k,s}$, and clearly,  $\dim(U)\leq 4$. Note that for every $i\in[3]$ and $t\in[m_i]$ we have that $\dim((V+U)^{\perp}(V_{i,t}))\leq 1$. This is exactly what we wanted to prove.
	
	To conclude, in all the possible cases, it holds that $\dim(U) \leq 5$ and from \autoref{obs:dint-inc} it holds that every nonzero linear combination $0 \neq \alpha P_0 + \beta Q_0$, satisfies $\rank_s(\alpha P_0 + \beta Q_0 ) > 2\dim(V+U) + 10$.
This concludes the proof of \autoref{lem:U-exist}. \qed\\
	
We now give the proof of  \autoref{cla:P-Q_0-C[V]-a_i}.	
	
	\begin{proof}[Proof of \autoref{cla:P-Q_0-C[V]-a_i}]
	Let $A_{i,t} \in \cT_i$ (the case $A_{i,t}\in\cT_j$ is analogous).  	
	We first handle polynomials of the form
	$A_{i,t} = A'_{i,t} + a_{i,t}(\epsilon_{i,t} a_{i,t} + v_{i,t})$ or $A_{i,t} = a_{i,t}^2 \in \calL_i$. We  show that in this case $a_{i,t} \in U$, from which the claim follows. If $a_{i,t} \in V$ then the statement holds. Thus, assume from now on that $a_{i,t}\notin V$. Consider the possible cases of \autoref{thm:structure} that $C$ and $A_{i,t}$ can satisfy.
	
	\begin{case}
		\item  Let $B_{j,r}\in \cQ_j$ be such that $\alpha B_{j,r} - \beta A_{i,t} = C$. 
		
		If $B_{j,r} = B'_{j,r} + b_{j,r}(\epsilon_{j,r} b_{j,r} + v_{j,r})$ then 
		\[(\alpha B'_{j,r} + \beta A'_{i,t} - C') = cd - \alpha b_{j,r}(\epsilon_{j,r} b_{j,r}+ v_{j,r})  -\beta a_{i,t}(\epsilon_{i,t} a_{i,t} + v_{i,t}) \;.
		\] 
		By assumption, $\alpha B'_{j,r}+ \beta A'_{i,t} - C' \in \spn{P_0,Q_0,\C[V]_2}$, and by the equality above $\rank_s(\alpha B'_{j,r} + \beta A'_{i,t} - C') \leq 3$.  \autoref{obs:P-Q-ideal} now implies that
		$D_V:= \alpha B'_{j,r} + \beta A'_{i,t} - C' \in \C[V]_2$. 
Observe that
		\begin{equation}\label{eq:D_V}
		D_V + \alpha b_{j,r}(\epsilon_{j,r} b_{j,r} + v_{j,r})  +\beta a_{i,t}(\epsilon_{i,t} a_{i,t} + v_{i,t}) = cd \;.		
		\end{equation}
		As $c,d \notin V$ it holds that both sides are not zero modulo $V$, and thus $\alpha\epsilon_{j,r} b_{j,r}^2 + \beta\epsilon_{i,t} a_{i,t}t^2  \equiv_{V} cd$. Observe that the left hand side is a reducible polynomial with both factors being linear combinations of $b_{j,r}$ and $a_{i,t}$. Thus, it is not hard to see that if $\epsilon_{i,t} \neq 0$ then  $a_{i,t} \in U$. If $\epsilon_{i,t} =0$, then $\epsilon_{j,r}\neq 0$ and then $b_{j,r} \in U$. As $\{c,d,b_{j,r}\} \subseteq U$ we get from Equation~\eqref{eq:D_V} that so does $a_{i,t}$.
		
		If on the other hand $B_{j,r} \in \ideal{V}$ then similarly we obtain that $cd - \beta a_{i,t}(\epsilon_{i,t} a_{i,t} + v_{i,t}) \in \ideal{V}$ and thus $a_{i,t} \in U$.

		\item  There are linear forms $g$ and $h$ such that:
		$A_{i,t} = \alpha C +gh$. 
		
		If $\alpha \neq 0$ then
		\[\frac{1}{\alpha} A'_{i,t} - C' = cd- \frac{1}{\alpha} a_{i,t}(\epsilon_{i,t} a_{i,t} + v_{i,t}) + \frac{1}{\alpha} gh.\]
		As before, it follows that $D_V := \frac{1}{\alpha} A'_{i,t} - C'\in \C[V]_2$, and
		\[D_V+ \frac{1}{\alpha} a_{i,t}(\epsilon_{i,t} a_{i,t} + v_{i,t}) - \frac{1}{\alpha}gh = cd.\]
		Since $c,d \notin V$, both sides are not zero modulo $V$, and thus $\frac{1}{\alpha} \epsilon_{i,t} a_{i,t}^2 -cd \equiv_V \frac{1}{\alpha}gh$. If $\epsilon_{i,t} \neq 0 $ we can use \autoref{cla:ind-rank} to deduce that $a_{i,t} \in U$. If $\epsilon_{i,t} = 0$ then $g,h \in U$ and, as before, so does $a_{i,t}$.
		
		If $\alpha = 0$ then we know that $A_{i,t}$ is reducible. The following observation tells us what cases we should consider. 
		
\begin{observation}\label{obs:gen-case-reducible}
Let $C\in \cT_k$ and $A=a^2\in \calL_i$. Then, either there is $T\in \cQ_j$ such that $T = \alpha C + a b$ for some linear form $b$, or $a\in \MS(C)$ or $C=b^2 \in \calL_k$ and $T\in {\ideal{a,b}}$.
\end{observation}

\begin{proof}[\textmd{Proof}]
		If $C$ is irreducible even after setting $a=0$ then the claim follows immediately from \autoref{cla:case-rk1-gen}.
	If  $C=b^2 \in \calL_k$ then $\sqrt{\ideal{b^2,a^2}} = {\ideal{a,b}}$ is a prime ideal and therefore some $T\in\cQ_j$ must satisfy $T\in \ideal{a,b}$. 
\end{proof}

		First, consider the case where there is a polynomial $B_{j,r}\in\cT_j$  satisfying $B_{j,r} = \beta C + a_{i,t} \ell$, for some linear form $\ell$. If $B_{j,r} \in \ideal{V}$, 
then  
\[B_{j,r}=\beta C +a_{i,t} \ell= \beta C' + \beta cd + a_{i,t}\ell \in \ideal{V} \;.\] 
By rank arguments we get that $C' \in \C[V]_2$, 
and therefore $a_{i,t} \in U$. If $B_{j,r}\not\in\ideal{V}$ then $B_{j,r} = B'_{j,r} + b_{j,r}(\epsilon_{j,r} b_{j,r} + v_{j,r})$ and then, using similar arguments to before, we deduce $a_{i,t}\in U$. 

		The other case is when $\rank_s(C)=2$ and $a_{i,t}\in \MS(C)$. From \autoref{obs:P-Q-ideal}, we get that $C'\in \C[V]_2$, and the claim follows in this case as well. 
		
		\item  In this case, there are linear forms $g$ and $h$ such that $C,A_{i,t} \in \ideal{g,h}$. In particular it holds that $\rank_s(C' + cd) =2$ and  \autoref{obs:P-Q-ideal} gives $C'\in \C[V]_2$. Similarly, we get that $A'_{i,t} \in \C[V]_2$.
		Note that we cannot have $g,h \in V$ as this would imply $cd \in \ideal{V}$, in contradiction to our assumption. Thus, as $g,h \in \MS(A_{i,t}) \subseteq \spn{V,a_{i,t}}$ it holds that, without loss of generality, $g = a_{i,t} + v'$ for $v'\in V$. As $g \in \MS(C)\subseteq U$ it follows that  $a_{i,t}\in U$ as well, which is what we wanted to prove.
	\end{case}
	
	We are not done yet as we have to  handle the case $A_{i,t} \in \cQ_i\cap \ideal{V}$. In this case we show that $\MS(A_{i,t}) \subseteq U$ and thus $A_{i,t} \in \C[U]_2$. Again we break the proof to three cases according to \autoref{thm:structure}.
	\begin{case}
		\item
		In this case  there is  a polynomial $B_{j,r}\in\cT_j$ such that $\alpha B_{j,r} + \beta A_{i,t} = C$. 		
		If $B_{j,r} \in \ideal{V}$ then we get that $C\in \ideal V$ which together with \autoref{obs:P-Q-ideal} imply that $C' \in \C[V]_2\subseteq\ideal{V}$. Consequently, we get that $cd\in \ideal{V}$ in contradiction to the choice of $C$. 
		
		If $B_{j,r} = B'_{j,r} + b_{j,r}(\epsilon_{j,r} b_{j,r} + v_{j,r})$ then as before we get that   $D_V:= \alpha B'_{j,r}  - C' \in \C[V]_2$, and that $D_V+\beta A_{i,t} +\alpha b_{j,r}(\epsilon_{j,r}b_{j,r} + v_{j,r}) = cd$. As before, by looking at this equality modulo $V$ we deduce that  $b_{j,r} \in U$ and $\MS(A_{i,t}) \subseteq \MS( cd -\alpha b_{j,r}(\epsilon_{j,r} b_{j,r} + v_{j,r}) - D_V) \subseteq U$ as we wanted to prove.

		\item  There are linear forms $g$ and $h$ such that
		$A_{i,t} = \alpha C +gh$ where $\alpha \neq 0$ (otherwise $A_{i,t}\in \calL_i$ which we handled before).
		Again, \autoref{obs:P-Q-ideal} implies $C' \in \C[V]_2$ and $\alpha cd + gh \in \ideal{V}$. Therefore, $\alpha cd \equiv_V -gh$ and hence $g,h \in U$. Consequently, $\MS(A_{i,t}) \subseteq U$.
		
		\item	 In this case, there are linear forms $g$ and $h$ such that $C,A_{i,t} \in \ideal{g,h}$. In particular it holds that $\rank_s(C' + cd) =2$ and therefore $C'\in \C[V]_2$. As before we cannot have $g,h \in V$. Therefore, without loss of generality, $g \in \spn{V,c,d}\setminus V$. Denote $A_{i,t} = g\ell + h\ell'$. As $A_{i,t}\in\ideal{V}$ we have that either $\ell\in V$ and $h\ell'\in \ideal{V}$, or $h,\ell' \in \spn{V,g, \ell}$. In either case we get that there is one linear form $\ell''\in \{\ell,\ell',h\}$ such that $\MS(A_{i,t})\subseteq \spn{V,c,d,\ell''}$. 
				
		We now repeat the same argument for $A_{i,t}$ and $\tilde{C}$. If $A_{i,t}$ and $\tilde{C}$ satisfy \autoref{thm:structure}\ref{case:span} or \autoref{thm:structure}\ref{case:rk1} then, as we already proved, the claim holds. Thus  $A_{i,t}$ and $\tilde{C}$ satisfy \autoref{thm:structure}\ref{case:2} as well. By the same argument it holds that there is a liner form $g' \in \spn{V,e,f}\setminus V$ such that $g' \in \MS(A_{i,t})\subseteq \spn{V,c,d,\ell''}$. From our assumption it follows that $g' \notin \spn{V,c,d}$. Hence, it must hold that $\ell'' \in \spn{V,c,d,g'}\subseteq U$, which implies $\MS(A_{i,t})\subseteq U$ and the claim follows.
	\end{case}
	This completes the proof of \autoref{cla:P-Q_0-C[V]-a_i} and with it the proof of \autoref{lem:U-exist}.
\end{proof}

\subsubsection{Proof of Lemma~\ref{lem:P-Q-and-C[V]}}
\label{sec:P-Q-and-C[V]}

	We first handle the case where $\cT_i,\cT_j \subset \C[V']_2$.  \autoref{cla:C[V]} implies that in this case $\cT_k \subseteq \C[V']_2$ and we are done.

\begin{claim}\label{cla:C[V]}
	Let $V$ be a linear space of linear forms. Assume $\cT_i, \cT_j \subset {\C[V]_2}$, for $i\neq j$. Then, the third set $\cT_k$ satisfies $\cT_k \subset \C[V]_2$ as well.  
\end{claim}

\begin{proof}[Proof of \autoref{cla:C[V]}]
	Without loss of generality assume $\cT_1,\cT_2 \subset \C[V]_2$.
	Assume towards a contradiction that there is $C \in \cT_3$ such that $C\notin \C[V]_2$. 

	First we handle the case that $\cQ_1\cup \cQ_2 = \emptyset$. Let $B = v^2 \in \calL_2$. We have that 
	\begin{equation}\label{eq:vaa}
	\prod_{T_r \in \calL_1} T_r \in \sqrt{\ideal{C, v^2}} \;.
	\end{equation}
	Note that $C|_{v=0}$ must be a reducible polynomial. Indeed, if it was irreducible then Equation~\eqref{eq:vaa} would imply that there is some $T\in \calL_1$ such that $T\in \sqrt{\ideal{C,v^2}}$. But as $T$ is reducible this implies that so is $C|{v=0}$. Denote $C = \tilde{u}\tilde{u}' + vb$. Equation~\eqref{eq:vaa} implies that there are polynomials ${u}^2,{u}'^2\in\calL_1$ such that $\tilde{u}\equiv_{v} u$ and $\tilde{u}'\equiv_{v} u'$. Thus, we can write $C =u u' + v b'$.
Similarly, we have that 
	\begin{equation}\label{eq:vaa2}
	\prod_{T_r \in \calL_2} T_r \in \sqrt{\ideal{C, u'^2}} \;.
	\end{equation}
	We know that $C|_{u'=0} = v b'$, and thus there is a polynomial of the form $(b'+\alpha u')^2 \in \calL_2\subseteq \C[V]_2$. This implies that $b'\in V$ and hence, $\MS(C) \subseteq V$ in contradiction.

	We now handle the case that $\cQ_1\cup\cQ_2 \neq \emptyset$.  Let $B \in \cQ_1\cup \cQ_2$. Assume without loss of generality that $B\in \cQ_2$. If $C\in \calL_3$ then considering the different cases of \autoref{obs:gen-case-reducible} we get  that in each of them $C \in \C[V]_2$, in contradiction.

	Hence from now on we assume $C \in \cQ_3$.	 It is not hard to verify, since $B \in \cQ_2$, that our assumption implies that $C$ and $B$ cannot satisfy \autoref{thm:structure}\ref{case:rk1} or \autoref{thm:structure}\ref{case:span}. Thus, $C$ and $B$  satisfy \autoref{thm:structure}\ref{case:2}. This implies that $C\in \ideal{V}$ and $\rank_s(C) =2$. Let $v_1$, $v_2$, $c$ and $d$ be linear forms such that $v_1,v_2\in V$ and  $C=v_1c + v_2d$.  As $C\not\in \C[V]_2$. we have that $\{c,d\} \not \subset V$. 
	We are now going to get a contradiction to the equation
	\begin{equation}\label{eq:vaa3}
	\prod_{T_r \in \cT_1} T_r \in \sqrt{\ideal{C,B}}=\sqrt{\ideal{v_1c+v_2d,B}} \;.
	\end{equation}
	We shall reach a contradiction by proving the existence of a common zero of $C$ and $B$ that is not a zero of $\prod_{T_r \in \cT_1} T_r$.
	First, we note that $Z(B)\setminus Z(v_1,v_2) \neq \emptyset$, as $B$ is irreducible.
	Moreover, if for every $\vaa \in  Z(B))\setminus Z(v_1,v_2)$ it holds that $\prod_{T_r \in \cT_j} T_r(\vaa) = 0$, then, as  $Z(B)\setminus Z(v_1,v_2)$ is open in $ Z(B)$, we get that $Z(B)\subseteq Z(\prod_{T_r \in \cT_j} T_r)$. Consequently, $\prod_{T_r \in \cT_j} T_r\in \ideal{B}$, and, as $\ideal{B}$ is prime (since $B$ is irreducible), there must be some $T_{r'}\in \cT_j$ such that $B | T_{r'} $. This contradicts the condition  that our polynomial are pairwise linearly independent.	 Thus, there is  $\vaa \in  Z(B)\setminus Z(v_1,v_2)$ such that $\prod_{T_r \in \cT_j} T_r(\vaa) \neq 0$. Note that as $B$ and the polynomials in $\cT_j$ are defined only using linear forms from $V$, we get that the same property holds for every point $\vbb$ that agrees with $\vaa$ on all linear forms in $V$. We can therefore modify $\vaa$, if needed, by changing the values of $c(\vaa)$ and $d(\vaa)$ (if $\spn{c,d}\cap V \neq \{\vec{0}\}$ then we modify only one of them) so that $C(\vaa) = 0$. Note that this is possible as after plugging $\vaa$ to the linear forms in $V$ we get that $C$ becomes a linear equation in $c$ and $d$. This  contradicts Equation~\eqref{eq:vaa3}. Hence, it must be the case that $\cQ_k \subset \C[V]_2$, as claimed.
\end{proof}

	We continue with the proof of Claim~\ref{lem:P-Q-and-C[V]}.
	Without loss of generality assume $\cT_1,\cT_2 \subset \spn{P_0,Q_0, \C[V']}$ and that $\cT_2 \not\subset \C[V']_2$.	Let $C \in \cT_3$  and $B\in \cT_2\setminus \C[V']_2$. 	By \autoref{obs:P-Q-ideal} we get that $\rank_s(B) >2$ and therefore $C, B$ do not satisfy \autoref{thm:structure}\ref{case:2}. If $C,B$ satisfy \autoref{thm:structure}\ref{case:span} then there is $A \in \cT_1$ such that $C \in \spn{B,A} \subseteq \spn{P_0,Q_0,\C[V']_2}$ as we wanted to prove. 
	So assume that $C$ and $B$ satisfy \autoref{thm:structure}\ref{case:rk1}.  Hence, there are linear forms $c$ and $d$ such that $C = \beta B+cd$.  \autoref{cla:case-rk1-gen} implies that there is $A \in \cT_1$ such that $A = \alpha B +ce$ for some linear form $e$.
	From  pairwise linear independence we know that $e \neq 0$. Thus $A -\alpha B = ec$, and in particular $\rank_s(A -\alpha B)=1$. As $A -\alpha B \in \spn{Q_0,P_0,\C[V']_2}$ it follows from \autoref{obs:P-Q-ideal} that $A -\alpha B \in \C[V']_2$, and therefore $\MS(A -\alpha B) \subseteq V'$ and thus $c,e\in V'$. The same argument shows that $d\in V'$, and  we get that $C \in \spn{Q_0, P_0,, \C[V']_2}$ as claimed.

This concludes the proof of Claim~\ref{lem:P-Q-and-C[V]}. \qed

\subsubsection{Proof of Lemma~\ref{lem:ek-lines}}\label{sec:ek-lines}

	



		Let $A_{i,t} \in \cT_i$ and $B_{j,r} \in \cT_j$ be such that  $\tilde{a}_{i,t}, \tilde{a}_{j,r}$ satisfy the condition in the lemma.  Observe that no matter whether $A_{i,t}\in\ideal{V'}$ or $A_{i,t}=A'_{i,t}+a_{i,t}(\epsilon_{i,t} a_{i,t}+v_{i,t})$, for $A'_{i,t}\in \spn{P_0,Q_0,\C[V']_2}$, we can express $T_{\vaa, V}(A_{i,t})$ as
	\[ T_{\vaa, V}(A_{i,t}) =  {A''}_{i,t} + \tilde{a}'_{i,t}(\epsilon'_{i,t} \tilde{a}'_{i,t} +\delta'_{i,t}z) \;,
	\]  
	with $A''_{i,t} \in \spn{T_{\vaa, V'}(P_0),T_{\vaa, V'}(Q_0)}$ and $\tilde{a}'_{i,t}\in\spn{\tilde{a}_{i,t},z}$. Similarly we write $T_{\vaa, V'}(B_{j,r}) =  {B''}_{j,r} + \tilde{a}'_{j,r}(\epsilon'_{j,r} \tilde{a}'_{j,r} +\delta'_{j,r}z)$. Observe that if $A''_{i,t}\neq 0$ then by \autoref{assume:special-case} and the choice of $T_{\vaa,V'}$, it follows that $\rank_s(T_{\vaa, V'}(A_{i,t})) > 2$. On the other hand, if $A''_{i,t}=0$ then $\rank_s(T_{\vaa, V'}(A_{i,t})) =1$.
	We split the proof according to whether $A''_{i,t}=B''_{j,r}= 0$ or not.

	
	\begin{enumerate}
		\item $\max\left(\rank_s\left(T_{\vaa, V'}(A_{i,t})\right), \rank_s\left(T_{\vaa, V}(B_{j,r})\right)\right) >2$:
		
		Assume without loss of generality that $\rank_s(T_{\vaa, V'}(A_{i,t}))>2$.
		As before we split the proof according to which cases of \autoref{thm:structure}  $T_{\vaa, V'}(A_{i,t})$ and $T_{\vaa, V'}(B_{j,r})$  satisfy. Note that we do not need to consider $\autoref{thm:structure}\ref{case:2}$ as we assume that at least one of the polynomials has $\rank_s$ higher than $2$.
		\begin{case}
			
			\item 	There are $\lambda,\beta \in \C\setminus \{0\}$ and $C_{k,s} \in \cQ_k$ such that  $\lambda T_{\vaa, V}(A_{i,t}) + \beta T_{\vaa, V}(B_{j,r}) = T_{\vaa, V}(C_{k,s})$: 	We have that
			\[\lambda A''_{i,t} + \beta B''_{j,r} - C''_{k,s} = \tilde{a}'_{k,s}(\epsilon'_{k,s} \tilde{a}'_{k,s} +\delta'_{k,s}z) - \lambda \tilde{a}'_{i,t}(\epsilon'_{i,t} \tilde{a}'_{i,t} +\delta'_{i,t}z) - \beta \tilde{a}'_{j,r}(\epsilon'_{j,r} \tilde{a}'_{j,r} +\delta'_{j,r}z) \;.\]
			If $\lambda A''_{i,t} + \beta B''_{j,r} - C''_{k,s} \neq 0$ then there is a non trivial linear combination of $ T_{\vaa, V'}(P_0), T_{\vaa, V'}(Q_0)$ of $\rank_s \leq 3$, in contradiction to \autoref{obs:P-Q-ideal}. Thus, $\lambda A''_{i,t} + \beta B''_{j,r} - C''_{k,s} =0$. This implies that 
			\[\tilde{a}'_{k,s}(\epsilon'_{k,s} \tilde{a}'_{k,s} +\delta'_{k,s}z) - \lambda \tilde{a}'_{i,t}(\epsilon'_{i,t} \tilde{a}'_{i,t} +\delta'_{i,t}z) - \beta \tilde{a}'_{j,r}(\epsilon'_{j,r} \tilde{a}'_{j,r} +\delta'_{j,r}z) =0\;.
			\]
			\autoref{cla:c,d,V}, which we state and prove next, shows that $\tilde{a}'_{k,s} \in \spn{\tilde{a}'_{i,t}, \tilde{a}'_{j,r},z} \setminus \left(\spn{\tilde{a}'_{i,t},z}\cup \spn{ \tilde{a}'_{j,r},z}\right)$. As $\tilde{a}'_{k,s} \in\spn{\tilde{a}_{k,s} ,z}$ this implies that $\tilde{a}_{k,s} \in \spn{\tilde{a}_{i,t}, \tilde{a}_{j,r},z} \setminus \left(\spn{\tilde{a}_{i,t},z}\cup \spn{ \tilde{a}_{j,r},z}\right)$ as we wanted to prove.
			
			\begin{claim}\label{cla:c,d,V}
				If there are linear forms $c$ and $d$, and $\lambda,\beta \in \C\setminus \{0\}$ such that $\lambda \tilde{a}_{i,t}(\epsilon_{i,t} \tilde{a}_{i,t} + \delta_{i,t}z) + \beta \tilde{a}_{j,r}(\epsilon_{j,r} \tilde{a}_{j,r} + \delta_{j,r}z) + cd \in \spn{z^2}$ then there are $\mu, \eta \in \C\setminus\{0\}$ and $\epsilon\in \C$, such that, without loss of generality, $c = \mu \tilde{a}_{i,t}+\eta \tilde{a}_{j,r} + \epsilon z$.
			\end{claim}
			
			\begin{proof}
				Since 
				\begin{equation}\label{eq:c,d,mu}
				\lambda \tilde{a}_{i,t}(\epsilon_{i,t} \tilde{a}_{i,t} + \delta_{i,t}z) + \beta \tilde{a}_{j,r}(\epsilon_{j,r} \tilde{a}_{j,r} + \delta_{j,r} z) + cd =\delta z^2 \;,
				\end{equation} 
				for some $\delta \in \C$, we get that $cd \equiv_{\tilde{a}_{j,r}} \delta z^2 - \lambda \tilde{a}_{i,t}(\epsilon_{i,t} \tilde{a}_{i,t} + \delta_{i,t}z)$. As the left hand side cannot be zero by our assumption,  it follows that $c,d\in \spn{\tilde{a}_{i,t}, \tilde{a}_{j,r}, z}$. Finally, we note that we cannot have $c,d \in  \spn{\tilde{a}_{i,t},z} \cup \spn{\tilde{a}_{j,r},z}$. Indeed, if both belong to, say, $\spn{\tilde{a}_{i,t},z}$ then we get from Equation~\eqref{eq:c,d,mu} that $\tilde{a}_{j,r}\in \spn{\tilde{a}_{i,t}, z}$ in contradiction. If $c\in 
				\spn{\tilde{a}_{i,t},z}$ and $d\in \spn{\tilde{a}_{j,r},z}$ then we get that in Equation~\eqref{eq:c,d,mu} the term  $\tilde{a}_{i,t}\tilde{a}_{j,r}$ cannot be cancelled, in contradiction.
			\end{proof}

			\item There are nonzero linear forms $c,d$ such that  $\lambda T_{\vaa, V'}(A_{i,t})  +T_{\vaa, V'}(B_{j,r}) = cd$, for some scalar $\lambda\in \C$: We split the proof of this case to two subcases
			\begin{itemize}
				\item $\rank_s( T_{\vaa, V'}(B_{j,r})) = 1$: In this case we can assume $\lambda=0$ and  $cd= T_{\vaa, V'}(B_{j,r})= \tilde{a}'_{j,r}(\epsilon'_{j,r} \tilde{a}'_{j,r} +\delta'_{j,r}z) $. 
				The assumption that $\rank_s(T_{\vaa, V'}(A_{i,t}))\geq 3$ implies that $T_{\vaa, V'}(A_{i,t})$ is irreducible even after setting $\tilde{a}'_{j,r}=0$. It follows that if a product of irreducible polynomials satisfies $\prod_{s} T_{\vaa, V'}(C_{k,s}) \in  \sqrt{\ideal{T_{\vaa, V'}(A_{i,t}),  T_{\vaa, V'}(B_{j,r})}}$ then, after setting $\tilde{a}'_{j,r}=0$, some $T_{\vaa, V'}(C_{k,s})$ is divisible by ${T_{\vaa, V'}(A_{i,t})}|_{\tilde{a}'_{j,r}=0}$. Thus, there is a multiplicand that is equal to $ T_{\vaa, V'}(C_{k,s}) =\gamma T_{\vaa, V'}(A_{i,t}) + \tilde{a}'_{j,r}e$ for some linear form $e$. 
				
				If $\gamma = 0$ then $\tilde{a}'_{j,r}e=T_{\vaa, V'}(C_{k,s}) = C''_{k,s}+\tilde{a}'_{k,s}(\epsilon'_{k,s} \tilde{a}'_{k,s} + \delta'_{k,s} z)$, for $C''_{k,s}\in \spn{T_{\vaa, V'}(P_0),T_{\vaa, V'}(Q_0)}$. From \autoref{obs:P-Q-ideal} we get  $C''_{k,s}=0$. It follows that $\tilde{a}'_{j,r}$ divides both $T_{\vaa, V'}(C_{k,s})$ and $T_{\vaa, V'}(B_{j,r})$, in contradiction to the choice of $T_{\vaa,V'}$ (recall \autoref{cla:res-z-ampping}).
				
				The last case to consider is $\gamma \neq 0$. As before, \autoref{obs:P-Q-ideal} implies that $C''_{k,s}-\gamma A''_{i,t} = 0$. Therefore, 
				\begin{equation}\label{eq:gammaneq0}
				\tilde{a}'_{k,s}(\epsilon'_{k,s} \tilde{a}'_{k,s} + \delta'_{k,s} z)  =  \gamma \tilde{a}'_{i,t}(\epsilon'_{i,t} \tilde{a}'_{i,t} + \delta'_{i,t} z) +\tilde{a}'_{j,r}e \;.
				\end{equation}	 
				Observe that if $\tilde{a}'_{i,t}$ divides $\tilde{a}'_{k,s}(\epsilon'_{k,s} \tilde{a}'_{k,s} + \delta'_{k,s} z)$ then it must also divide $\tilde{a}'_{j,r} e$ and hence $\tilde{a}'_{i,t}\sim e$. But then if we divide both sides of Equation~\eqref{eq:gammaneq0} by $\tilde{a}'_{i,t}$ then we get that $\tilde{a}'_{j,r}\in\spn{\tilde{a}'_{i,t},z}$ in contradiction. Therefore, if we consider Equation~\eqref{eq:gammaneq0}  modulo $\tilde{a}'_{i,t}$ then we get that both sides are nonzero and that   $\tilde{a}'_{k,s}(\epsilon'_{k,s} \tilde{a}'_{k,s} + \delta'_{k,s} z) \equiv_{\tilde{a}'_{i,t}} \tilde{a}'_{j,r}e$. It follows  that either $\tilde{a}'_{k,s}$ or $\epsilon_{k,s}' \tilde{a}'_{k,s} +\delta'_{k,s}z$ is equivalent to $\tilde{a}'_{j,r}$, modulo $\tilde{a}'_{i,t}$.  Observe that we cannot have $\tilde{a}'_{k,s} = \lambda \tilde{a}'_{j,r}+\epsilon' z$ as in this case Equation~\eqref{eq:gammaneq0} implies that  $\tilde{a}'_{i,t}\in\spn{ \tilde{a}'_{j,r},z}$ in contradiction to the choice of $ \tilde{a}'_{i,t}$ and $\tilde{a}'_{j,r}$.
				This proves that $a'_{k,s}$ satisfies the requirements of the claim.
				
				\item 	$\rank_s(T_{\vaa, V'}(B_{j,r})) \geq3$: In this case we must have $\lambda\neq 0$. By repeating the argument from {\bf Case~(i)} we get from \autoref{cla:c,d,V}  that, without loss of generality, $c=\mu \tilde{a}'_{i,t}+\eta \tilde{a}'_{j,r}+\epsilon z$ with $\mu\eta\neq 0$.  \autoref{cla:case-rk1-gen} implies the existence of $C_{k,s}\in \cT_k$ with 
				\[T_{\vaa,V'}(C_{k,s})=\gamma T_{\vaa, V'}(B_{j,r}) + ce = -\gamma\lambda T_{\vaa, V'}(A_{i,t}) + c(e+\gamma d)\;.\] 
				As before we conclude that 
				\begin{align*}
				\tilde{a}'_{k,s}(\epsilon'_{k,s} \tilde{a}'_{k,s} +\delta'_{k,s}z) &= \gamma  \tilde{a}'_{j,r}(\epsilon'_{j,r} \tilde{a}'_{j,r} + \delta'_{j,r} z)+(\mu \tilde{a}'_{i,t}+\eta \tilde{a}'_{j,r}+\epsilon z)e \\ &= -\gamma\lambda  \tilde{a}'_{i,t}(\epsilon'_{i,t} \tilde{a}'_{i,t} + \delta'_{i,t} z)+(\mu \tilde{a}'_{i,t}+\eta \tilde{a}'_{j,r}+\epsilon z)(e+\gamma d)\;.
				\end{align*}
				Rank arguments imply that $e,(e+\gamma d)\in\spn{\tilde{a}'_{i,t}, \tilde{a}'_{j,r}, z}$ and therefore $\tilde{a}'_{k,s}\in\spn{\tilde{a}'_{i,t}, \tilde{a}'_{j,r}, z}$ as well. We cannot have  $\tilde{a}'_{k,s}\in\spn{\tilde{a}'_{j,r}, z}$ or $\tilde{a}'_{k,s}\in\spn{\tilde{a}'_{i,t}, z}$ since this would imply $\mu=0$ or $\eta=0$, respectively,  in contradiction.

			\end{itemize}

		\end{case}
		\item $\rank_s(T_{\vaa, V'}(A_{i,t}))= \rank_s(T_{\vaa, V'}(B_{j,r})) = 1$:
		
		In this case, since $\ideal{\tilde{a}'_{i,t}, \tilde{a}'_{j,r}}$ is a prime ideal, we have that $\sqrt{\ideal{T_{\vaa,V'}(A_{i,t}), T_{\vaa,V}(B_{j,r})}} \subseteq \ideal{\tilde{a}'_{i,t}, \tilde{a}'_{j,r}}$.
		We also know that there are $C_{k_1},C_{k_2},C_{k_3},C_{k_4}\in \cT_k$ such that 
		\[
		T_{\vaa,V'}(C_{k_1})\cdot T_{\vaa,V'}(C_{k_2})\cdot T_{\vaa,V'}(C_{k_3})\cdot T_{\vaa,V'}(C_{k_4}) \in \sqrt{\ideal{T_{\vaa,V'}(A_{i,t}), T_{\vaa,V'}(B_{j,r})}}\subseteq \ideal{a'_{i,t}, a'_{j,r}}.
		\] 
		Since $\ideal{a'_{i,t}, a'_{j,r}}$ is prime it follows that, without loss of generality, $T_{\vaa,V'}(C_{k_1})\in\ideal{a'_{i,t}, a'_{j,r}}$. As $\rank_s(T_{\vaa,V'}(C_{k_1}))$ is either $1$ or greater than $2$, we conclude that $\rank_s(T_{\vaa,V'}(C_{k_1})) = 1$.
		Note that it cannot be the case that $T_{\vaa,V'}(C_{k_1})\in\ideal{a'_{i,t}} \cup \ideal{a'_{j,r}}$, as in this case we get that $T_{\vaa,V'}(C_{k_1})$ and one of $T_{\vaa,V'}(A_{i,t})$ and $T_{\vaa,V'}(B_{j,r})$ share a common factor, which is not a polynomial in $z$, in contradiction to our choice of $T_{\vaa, V'}$ (recall  \autoref{cla:res-z-ampping}). This means that there is a factor of $T_{\vaa,V'}(C_{k_1})$ lying in $\spn{\tilde{a}'_{i,t}, \tilde{a}'_{j,r},z} \setminus \left(\spn{\tilde{a}'_{i,t},z}\cup \spn{ \tilde{a}'_{j,r},z}\right)$.
		Consequently, $\tilde{a}'_{k,k_1} \in \spn{\tilde{a}'_{i,t}, \tilde{a}'_{j,r},z} \setminus \left(\spn{\tilde{a}'_{i,t},z}\cup \spn{ \tilde{a}'_{j,r},z}\right)$ as we wanted to prove.
	\end{enumerate}
	
	
	This concludes the proof of \autoref{lem:ek-lines} and with it the proof of \autoref{prop:P,Q_0,C[V],rank1}.

\subsection{For some $j\in [3]$, $\cQ_j \neq \calP^{\ref{case:span}}_j \cup  \calP^{\ref{case:2}}_j $}
\label{sec:hard}

We now handle the case where for some $j\in [3]$, there is some polynomial $Q_0 \in \cQ_j \setminus \calP^{\ref{case:2}}_j \cup \calP^{\ref{case:span}}_j$. This is the last case to consider in the proof of \autoref{thm:main}.

Assume without loss of generality that there is a polynomial $Q_0\in\cQ_2 \setminus \calP^{\ref{case:2}}_2 \cup \calP^{\ref{case:span}}_2$. We also assume,  without loss of generality, that 
\[|\cQ_1|=m_1\leq m_3=|\cQ_3|\;.\footnote{Note that we cannot assume anything regarding the value of $|\cQ_2|=m_2$.} \]



\begin{remark}\label{rem:gamma=1}
To slightly simplify the notation we assume that whenever a polynomial $Q\in\cup \cT_j$ is equal to $Q=\gamma Q_0 + ab$ for some $0\neq \gamma \in \C$, then it holds that $\gamma=1$. As multiplying polynomials in $\cup \cT_j$ by nonzero constants does not affect the conditions nor conclusion of \autoref{thm:main} this is indeed without loss of generality.
\end{remark}
 
Our proof will be based on yet another case analysis that relies on the following notation. For $i\in [3]$, we denote, similarly to \cite{DBLP:journals/corr/abs-2003-05152}, 
\[\cQ_i = \cQ^1_i \cup \cQ^2_i
\] 
where 
\begin{equation}\label{eq:defQ^1}
\cQ^1_i = \lbrace Q_j \in \cQ_i \mid Q_j = Q_0+ a_jb_j \rbrace
\end{equation} 
and 
\[\tilde{\cQ}^2_i = \lbrace Q_j \in \cQ_i \mid Q_j = \gamma Q_0+ L, \rank_s(L) \geq 2\rbrace \;.
\] 
Finally, we denote 
\begin{equation}\label{eq:defQ^2}
\cQ^2_i = \tilde{\cQ}^2_i \setminus \cQ^1_i\;.
\end{equation} 
Thus, in particular  $\cQ^1_i \cap \cQ^2_i = \emptyset$.\footnote{Note that if $\cQ^2_i \neq \tilde{\cQ}^2_i$ then $\rank_s(Q_0)\leq 3$.}





The outline of the proof of \autoref{thm:main}  in this case is as follows:
\begin{enumerate}
	\item In \autoref{sec:q22=0} we study the case  $\cQ_2^2 = \emptyset$. The analysis again splits to several subcases:
	\begin{enumerate}
		\item In \autoref{lem:Q_i^2=0} we handle the case where $\forall i$, $\cQ_i^2 = \emptyset$.
		\item In \autoref{lem:Q_i^2=Q_i^1} we handle the case $\cQ_1^2 =\cQ_2^2 = \emptyset$ (its analysis uses the previous case at some point).
		\item Finally we handle the case $\cQ_1^2 \neq \emptyset$ and $\cQ_2^2 = \emptyset$ in \autoref{lem:Q22=0}. Here we prove the existence of a constant dimensional subspace of linear forms $V$ such that most of the polynomials $C\in \cQ^1_3$ satisfy $C\in \spn{Q_0,\C[V]_2}$. \label{case:Vexists}
	\end{enumerate}
\item The next step is showing that the existence of a subspace $V$ as in Case~\ref{case:Vexists} above
implies \autoref{thm:main} (regardless of whether $\cQ_2^2 = \emptyset$ or not). We prove this by showing that the sets $\cQ_2$, $\cQ_1$ and $\cQ_3$ satisfy the conditions of \autoref{prop:P,Q_0,C[V],rank1}. The proof is composed of three claims -  \autoref{cla:Q_2-strac}, \autoref{cla:Q_1-strac} and \autoref{cla:Q_3-strac} - each handling a different set. 

	\item The last case to consider is when $\cQ_2^2 \neq \emptyset$ and no such special space $V$ exists. This is handled In \autoref{sec:q22<>0} where we prove \autoref{lem:noV-main}. As in the previous cases we prove that the conditions of  \autoref{prop:P,Q_0,C[V],rank1} hold and deduce \autoref{thm:main} from this. Here too we handle $\cQ_2$, $\cQ_1$ and $\cQ_3$ in separate claims (Claims~\ref{lem:Q_2-strac-noV-gen}, \ref{cla:Q_1-struc-noV} and \ref{cla:Q_3-struc-noV}, respectively).
\end{enumerate}

For the reduction to \autoref{prop:P,Q_0,C[V],rank1} we shall construct a certain subspace $V$ of linear forms that will capture much of the structure of the linear forms $a_j$ and $b_j$ (as in Equation~\eqref{eq:defQ^1}). The way that we construct $V$ depends on the $\rank_s$ of $Q_0$. This motivates the following definition.

\begin{definition}\label{def:ranksV}
	Let $V$ be a linear space of linear forms, and let $Q$ be a quadratic polynomial. We define an operator $V_{\MS}(V,Q)$ as follows:
	\[
	V_{\MS}(V,Q)=
	\begin{cases}
	V, & \rank_s(Q)\geq \rkq\\
	V + \MS(Q), & \rank_s(Q)< \rkq
	\end{cases}.
	\]
\end{definition}

Observe that this definition implies the following simple claim.

\begin{claim}\label{cla:lin-comb-V}
	Let $V=V_\MS(\MS(L),Q)$ for quadratics $L$ and $Q$ such that $\rank_s(L)<100$. Then, if $A\in\spn{L,Q}$ is such that $\rank_s(A)<100$ then $\MS(A)\subseteq V$.  
\end{claim}

\begin{proof}
	If $A$ is a multiple of $L$ then there is nothing to prove. Otherwise, it follows that $\rank_s(Q)<200$. Hence, $\MS(L),\MS(Q)\subseteq V$, which implies the claim.  
\end{proof}

Another difference in the analysis when $\rank_s(Q_0)$ is large is that we do not need to consider  \autoref{thm:structure}\ref{case:2} in our arguments.

\begin{observation}\label{obs:case3}
	Assume that $\cup_{j\in[3]} \cT_j$ satisfy the assumption in \autoref{sec:hard} with $\rank_s(Q_0)\geq \rkq$. Then, every $C_j \in \cQ_3^1$ satisfies  $\rank_s(C_j) \geq \rkq-1>2$ and in particular, $C_j$ never satisfies \autoref{thm:structure}\ref{case:2} with any other polynomial in $\cup_{j\in[3]} \cT_j$.
\end{observation}




We note that by our assumptions so far, namely that $Q_0\in\cQ_2 \setminus \calP^{\ref{case:2}}_2 \cup \calP^{\ref{case:span}}_2$ exists and that $m_1\leq m_3$, we have that 
\begin{equation}\label{eq:size-Q3}
|\cQ_3^1| \geq (1-2\delta)m_3\;.
\end{equation}

Finally, a convention that we will use throughout the proof is that $A_i$, $B_i$ and $C_i$ denote polynomials in  $\cQ_1$, $\cQ_2$ and  $\cQ_3$, respectively.

\subsubsection{The case $\cQ_2^2=\emptyset$}\label{sec:q22=0}

As mentioned above we start by handling the simple case where for every $i$, $\cQ_i^2 = \emptyset$.

\begin{lemma}\label{lem:Q_i^2=0}
		Let $\cup_{j\in[3]} \cT_j$ satisfy the assumption in \autoref{sec:hard}. Assume further that for every $j\in [3]$, $\cQ_j^2 = \emptyset$. Then, $\dim(\cup_{j\in[3]} \cT_j) = O(1)$. 
\end{lemma}
\begin{proof}
	In this case, for every $Q_i \in \cup_{j\in[3]} \cQ_j$ it holds that $Q_i = \alpha_i Q_0 + a_ib_i$. 
	
		If $\rank_s(Q_0)\geq \rkq$ then let $P_0=0$ and $V= \spn{\vec{0}}$.  \autoref{prop:P,Q_0,C[V],rank1}, when applied to $Q_0,P_0$ and $V$, implies that  $\dim(\cup_{j\in[3]} \cT_j) = O(1)$.

If $\rank_s(Q_0)< \rkq$ then  we set $V= \MS(Q_0)$. As before, we conclude from applying \autoref{prop:P,Q_0,C[V],rank1} with $Q_0=0$, $P_0=0$ and $V$, that  $\dim(\cup_{j\in[3]} \cT_j) = O(1)$.
\end{proof}

Next we prove that a similar conclusion holds when $\cQ_1^2 =\cQ_2^2 = \emptyset$.

\begin{lemma}\label{lem:Q_i^2=Q_i^1}
	Let $\cup_{j\in[3]} \cT_j$ satisfy the assumption in \autoref{sec:hard}. Assume further that for every $i \in [2]$, $\cQ_i^2 = \emptyset$. Then $\dim(\cup_{j\in[3]} \cT_j) = O(1)$.
\end{lemma}
\begin{proof}
	Let $C \in \cQ_3^2$. Consider the possible cases of \autoref{thm:structure} that $C$ and $Q_0$ can satisfy.

	\begin{case}
		\item In this case there is a polynomial, $A\in \cT_1$ such that $C \in \spn{Q_0,A}$. From the assumption that $A \in \cT_1=\cQ_1^1\cup \calL_1$ it follows that $C \in \cQ_3^1$.
		\item In this case, by definition, $C \in \cQ_3^1$.
	\end{case}
	If $\rank_s(Q_0)\geq \rkq$ then $C$ and $Q_0$ cannot satisfy \autoref{thm:structure}\ref{case:2}, and thus in this case $\cQ_3^2=\emptyset$ and the claim follows from \autoref{lem:Q_i^2=0}. 

If  $\rank_s(Q_0) < \rkq$ then we have to consider \autoref{thm:structure}\ref{case:2} as well. Denote $V = \MS(Q_0)$. Clearly, if $C$ satisfies \autoref{thm:structure}\ref{case:2} with $Q_0$  then $C\in \ideal{V}$.

Thus we obtain that $\cQ_3^2\subseteq \ideal{V}$, and since $\cQ_1 = \cQ_1^1$ and $\cQ_2= \cQ_2^1$ we can apply \autoref{prop:P,Q_0,C[V],rank1} with $Q_0=0,P_0=0$ and $V$ and conclude  that  $\dim(\cup_{j\in[3]} \cT_j) = O(1)$.
\end{proof}

 We next consider the case $\cQ_1^2 \neq \emptyset$ and $\cQ_2^2 = \emptyset$.
The case  $\cQ_2^2 \neq \emptyset$ is handled in \autoref{sec:q22<>0}. Our goal is again to prove that the conditions of \autoref{prop:P,Q_0,C[V],rank1} hold here as well.

\begin{lemma}\label{lem:Q22=0}
	Let $\cup_{j\in[3]} \cT_j$ satisfy the assumptions of \autoref{sec:hard}. Assume further that $\cQ_2^2 = \emptyset$ and $\cQ_1^2 \neq \emptyset$. Then there is a linear space of linear forms $\tilde{V}$ such that $\dim(\tilde{V})\leq 4$, and for $V= V_{\MS}(\tilde{V},Q_0)$ it holds that there are at least $(1-3\delta)m_3$ polynomials  $C \in \cQ_3^1\cap\spn{Q_0,\C[V]_2}$.
\end{lemma}
\begin{proof}
Let $A\in \cQ_1^2$. Partition the polynomials in $\cQ_3^1$ to three sets according to which case of \autoref{thm:structure} they satisfy with $A$:
\[\cI^{\ref{case:span}} = \{C_k \in \cQ_3^1 \mid A,C_k \text{ satisfy \autoref{thm:structure}\ref{case:span}}\}.\]
Similarly define $\cI^{\ref{case:rk1}},\cI^{\ref{case:2}}$. Equation~\eqref{eq:size-Q3} implies that $\cI^{\ref{case:span}}\cup\cI^{\ref{case:rk1}}\cup\cI^{\ref{case:2}}=\cQ_3^1\neq \emptyset$.

\begin{claim}
Under the assumptions of \autoref{lem:Q22=0} there is a quadratic form $L$ such that $\rank_s(L) = 2$ and $A = \alpha Q_0 + L$. When $\rank_s(A)=2$ we assume $\alpha=0$.
\end{claim} 

\begin{proof}
If $\cI^{\ref{case:2}} \neq \emptyset$ then by definition $\rank_s(A) = 2$ and $L = A$ satisfies the desired properties.
If $\cI^{\ref{case:rk1}} \neq \emptyset$ then there is $0\neq \alpha\in \C$, a polynomial $C_i\in \cQ_3^1$ and linear forms $a$ and $b$ such that $A = \alpha C_i + ab = \alpha Q_0 + \alpha a_ib_i+ab$ and thus $L = \alpha a_ib_i+ab$. Note that $\rank_s(L) \neq 1$ as we assumed $A \in \cQ_1^2$.
If $\cI^{\ref{case:span}} \neq \emptyset$ then there are polynomials $C_i\in \cQ_3^1$ and $B_k \in \cT_2$ such that $A = \alpha C_i + \beta B_k$ with $\alpha\beta\neq 0$. As 
$\cT_2= \cQ_2^1 \cup \calL_2$ we can denote $B_k = \gamma_kQ_0 +a_kb_k$. Hence, $A= (\alpha + \beta\gamma_k) Q_0 + \alpha a_ib_i+\beta a_kb_k$ and thus $L = \alpha a_ib_i+\beta a_kb_k$. As before we have that $\rank_s(L) \neq 1$.
\end{proof}

Set $\tilde{V} = \MS(L)$ and let $V=V_{\MS}(\tilde{V}, Q_0)$.

\begin{claim}\label{cla:IjinV}
Under our assumptions it holds that $\cI^{\ref{case:span}},\cI^{\ref{case:rk1}}\subseteq \spn{Q_0,\C[V]_2}$.
\end{claim}

\begin{proof}
We first prove the claim for $\cI^{\ref{case:rk1}}$. Let $ C_i \in \cI^{\ref{case:rk1}}$. There are linear forms $a$ and $b$ and a constant $\beta$ such that \[\alpha Q_0 + L= A = \beta C_i + ab = \beta Q_0 + \beta a_ib_i+ab \;.\] 
Therefore, $(\alpha - \beta)Q_0+ L = \beta a_ib_i+ab$. From the fact that $A \in \cQ_1^2$ we deduce that $\rank_s(\beta a_ib_i+ab) = 2$.
If $\rank_s(Q_0) \geq \rkq$ then  $\alpha - \beta = 0$ and \autoref{cor:containMS}  implies that $a_i,b_i,a,b\in \MS(L) = V$. If $\rank_s(Q_0)< \rkq$ then $a_i,b_i,a,b\in \MS((\alpha - \beta)Q_0+ L) \subset \MS(Q_0)+ \MS(L) = V$.

We now consider $\cI^{\ref{case:span}}$. Let $C_i\in \cI^{\ref{case:span}}$.  There is a polynomial $B_k \in \cT_2= \cQ_2^1 \cup \calL_2$ such that 
	\[\alpha Q_0 + L=A = \gamma C_i + \beta B_k = (\gamma + \beta\gamma_k) Q_0 + \gamma a_ib_i+\beta a_kb_k.\]
We again see that if $\rank_s(Q_0) \geq \rkq$, then  $a_i,b_i,a_k,b_k\in \MS(L) = V$ and if $\rank_s(Q_0)< \rkq$ then $a_i,b_i,a_k,b_k\in \MS((\alpha- \gamma- \beta\gamma_k)Q_0+ L) \subset \MS(Q_0)+ \MS(L) = V$.
\end{proof}

Note that if
$\rank_s(Q_0)\geq \rkq$  then $\cI^{\ref{case:2}} = \emptyset$ and so \autoref{lem:Q22=0} follows from \autoref{cla:IjinV}.

So assume $\rank_s(Q_0)< \rkq$. We next show that at most  $\delta m_3$ of the polynomials $C_i\in \cI^{\ref{case:2}}$ are not in $\spn{Q_0,\C[V]_2}$. Note that here we have that $\rank_s(A)=2$ and thus we can assume that $A=L$.

\begin{claim}\label{cla: Q_0,L,2}
	Let $C_i=Q_0+a_ib_i \in \cQ_3^1$. If there are linear forms $c$ and $d$ such that $C_i, L\in \ideal{c,d}$. Then, either $a_i,b_i \in V=\MS(Q_0) + \MS(L)$ or $Q_0\in \ideal{c,d}$.
\end{claim}

\begin{proof}
$Q_0$ is irreducible and thus $\rank_s(Q_0)\geq 2$. Since $\rank_s(C_i) =2$ we get from \autoref{cla:ind-rank} that, without loss of generality, $b_i = \epsilon_i a_i + v_i$ for some $v_i \in \MS(Q_0)\subseteq V$. Therefore, $\MS(C_i)\subseteq\spn{\MS(Q_0),a_i}$. As $C_i,L\in \ideal{c,d}$ it follows that $c,d\in \MS(C_i)\cap \MS(L)  \subseteq\spn{\MS(Q_0), a_i}$. 

If $a_i \in \MS(Q_0)$ then clearly the claim holds. So assume $a_i\notin \MS(Q_0)$. 

Consider first the case where $c,d \in \MS(Q_0)$. As $a_i\notin \MS(Q_0)$, setting $a_i =0$ does not effect $\MS(Q_0)$ and as  $Q_0\equiv_{a_i} C_i \in \ideal{c,d}$ we conclude that $Q_0 \in \ideal{c,d}$. 

On the other hand, if, say, $c\not\in \MS(Q_0)$ then  $a_i \in \spn{\MS(Q_0),c} \subseteq \MS(Q_0)+ \MS(L)=V$, and thus $a_i,b_i \in V$. 
\end{proof}

\autoref{cla: Q_0,L,2} implies that if $C_i\in \cI^{\ref{case:2}}$ is such that $C_i\not\in\spn{Q_0,\C[V]_2}$ then $C_i$ and $Q_0$ satisfy \autoref{thm:structure}\ref{case:2}. By choice of $Q_0$ there are at most $\delta m_3$ such polynomials in $\cQ_3$. 
Thus, with the exception of those $\delta m_3$ polynomials and possibly the (at most) $2\delta m_3$ polynomials in $\cQ_3^2$ we get that all other polynomials $C_i\in \cQ_3^1$ satisfy $C_i\in\spn{Q_0,\C[V]_2}$ as claimed. This concludes the proof  of \autoref{lem:Q22=0}.
\end{proof}

\subsubsection{A special $V$ exists}\label{sec:Vexists}

We now show that whenever a subspace that satisfies the properties described in \autoref{lem:Q22=0} exists, the conclusion of \autoref{thm:main} holds.


\begin{lemma}\label{lem:case-V-exsists}
Let $\cup_{j\in[3]} \cT_j$ satisfy the conditions of \autoref{sec:hard}. Assume  that there is a linear space of linear forms $\tilde{V}$ such that $\dim(\tilde{V})\leq \dimV$ and  $V= V_{\MS}(\tilde{V},Q_0)$
satisfies that there are more than $0.8 m_3$ polynomials  $C \in \cQ^1_3\cap \spn{Q_0,\C[V]_2}$. Then, $\dim(\cup_{j\in[3]} \cT_j) = O(1).$
\end{lemma}

As mentioned earlier, we prove \autoref{lem:case-V-exsists} by showing that each of the sets $\cQ_i$ satisfies the conditions of \autoref{prop:P,Q_0,C[V],rank1}, with a slightly larger subspace $V'$. We prove this first for $\cQ_2$, then for $\cQ_1$ and finally for $\cQ_3$.

\begin{claim}\label{cla:V-and-P}
Let $\cup_{j\in[3]} \cT_j$ and $V$ be as in \autoref{lem:case-V-exsists}.
Then every $B_i\in \cQ_2$ satisfies:
\begin{enumerate}
	\item \label{cla:V-and-P:item:V} $B_i\in \C[V]_2$, or
	\item \label{cla:V-and-P:item:V-c} There is a quadratic polynomial $B'_i\in \C[V]_2$, a linear form $v_i\in V$, and  a linear form $c_i$, such that $B_i = Q_0+B'_i + c_i(\epsilon_i c_i + v_i )$, or
	\item \label{cla:V-and-P:item:span} At least $0.8m_1$ of the polynomials in $\cQ_1$ are in $\spn{B_i, Q_0, \C[V]_2}\setminus \spn{Q_0, \C[V]_2}$, or, 
	\item \label{cla:V-and-P:item:ideal} $\rank_s(Q_0)< \rkq$ and $B_i \in \ideal{V}$.
\end{enumerate}
\end{claim}
We stress that the case $B_i \in \ideal{V}$ can happen only when $\rank_s(Q_0)< \rkq$.
\begin{proof}
Let 
\[\cI = \lbrace C_i \in \cQ^1_3 \mid C_i = Q_0+a_ib_i \text{ with } a_i,b_i\in V \rbrace\;.\] The definition of $V$ guarantees that $|\cI| > 0.8 m_3$ (recall \autoref{rem:gamma=1}).
Consider $B\in \cQ_2$. If there are $C_i,C_j\in \cI$ such that $C_i,B$ and $C_j,B$ satisfy $\autoref{thm:structure}\ref{case:rk1}$, then there are linear forms $c,d,e$ and $f$ and scalars $\alpha$ and $\beta$ such that $B = \alpha C_i + cd= \beta C_j + ef$.  Hence,  $(\alpha-\beta)Q_0 +\alpha a_ib_i- \beta a_jb_j = ef-cd$.
If $\rank_s(Q_0) \geq \rkq$ then we have that $\alpha = \beta$ and $\{\vec{0}\} \neq \MS(ef-cd) \subseteq V$. If $\rank_s(Q_0) < \rkq$ then $\MS(cd-ef) \subseteq \MS((\alpha-\beta)Q_0 +\alpha a_ib_i- \beta a_jb_j) \subseteq V$. In either case, \autoref{cla:rank-2-in-V} implies that $\spn{c,d}\cap V \neq \{\vec{0}\}$ and therefore $B$ satisfies Case~\ref{cla:V-and-P:item:V-c}.

If there is $C_i \in \cI$ such that $B$ and $C_i$ satisfy \autoref{thm:structure}\ref{case:2} then by \autoref{obs:case3} it must be the case that $\rank_s(Q_0) < \rkq$ and $B \in \ideal{V}$ and in particular $B$ satisfies Case~\ref{cla:V-and-P:item:ideal}.

Thus we are left with the case that $B$ satisfies \autoref{thm:structure}\ref{case:span} with all but at most one of the polynomials in $\cI$.
If there is $A_t \in \calL_1$ and $C_i \in \cI$ such that $A_t \in \spn{B,C_i}$ then in particular $B$ and $C_i$ satisfy \autoref{thm:structure}\ref{case:rk1} with $c=d$ and we are done by the previous case.

If there is $A_t \in \cQ_1$ and $C_i,C_j \in \cI$ such that $A_t \in \spn{B,C_i}\cap \spn{B,C_j}$ then by pairwise linear independence it follows that $B \in \spn{C_i,C_j}$ and then either Case~\ref{cla:V-and-P:item:V-c} with $c_i=0$  or Case~\ref{cla:V-and-P:item:V} hold.

The only case left is when for every $C_i \in \cI$ (except possibly the one satisfying  \autoref{thm:structure}\ref{case:2} with $B$)  there is a different $A_i \in \cQ_1$ such that $A_i \in \spn{B,C_i}$.
As $m_1 \leq m_3$, it follows that $B, Q_0, \C[V]_2$ span at least $0.8 m_1$  polynomials in $\cQ_1$. Thus, either $B\in\spn{Q_0,\C[V]}$, and in particular it satisfies Case~\ref{cla:V-and-P:item:V-c}, or it satisfies Case~\ref{cla:V-and-P:item:span}.
\end{proof}

Denote  
\begin{equation}\label{eq:J2}
\cJ_2 = \left\{ Q\in \cQ_2 \mid Q \text{ satisfies  Case~\ref{cla:V-and-P:item:span} of \autoref{cla:V-and-P} and none of the other cases}\right\}\;.
\end{equation}
Fix $B_0\in \cJ_2$. 
Then, for every other $B'\in \cJ_2$  there is $A_i\in \cQ_1$ such that $A_i \in \left(\spn{B_0,Q_0,\C[V]_2}\setminus \spn{Q_0, \C[V]_2}\right)\cap  \left(\spn{B',Q_0,\C[V]_2}\setminus \spn{Q_0, \C[V]_2}\right)$ and therefore $ B' \in \spn{B_0,Q,\C[V]_2}$. 

Thus from now on, if $\cJ_2 \neq \emptyset$ then we can assume that there is a polynomial $B_0\in \cQ_2$ such that 
$\cJ_2 \subset \spn{B_0,Q_0,\C[V]_2}$. 


\begin{corollary}\label{cla:Q_2-strac}
	Every polynomial $B_i \in \cQ_2$ satisfies one of the following cases:
	\begin{enumerate}
		\item \label{cla:Q_2-strac:item:V} $B_i\in\C[V]_2$, or
		\item \label{cla:Q_2-strac:item:V-c} There is a quadratic polynomial $B'_i\in\C[V]_2$, a linear form $v_i\in V$, and  a linear form $c_i$, such that $B_i =  Q_0+B'_i + c_i(\epsilon_i c_i + v_i )$, or
		\item \label{cla:Q_2-strac:item:span}  $B_i \in \spn{B_0,Q_0,\C[V]_2}$, for some fixed polynomial $B_0\in \cJ_2$ (this case is possible only when $\cJ_2\neq \emptyset$), or \item \label{cla:Q_2-strac:item:ideal} $\rank_s(Q_0)< \rkq$ and
$B_i \in \ideal{V}$.
\end{enumerate}
\end{corollary}

We next prove a similar statement for the polynomials in $\cQ_1$.

\begin{claim}\label{cla:Q_1-strac}
	Let $\cup_{j\in[3]} \cT_j$ and $V$ be as in \autoref{lem:case-V-exsists}. Let $\cJ_2$ and $B_0$ be as in \autoref{cla:Q_2-strac}.
Then, every polynomial $A_i \in \cQ_1$ satisfies one of the following cases:
	\begin{enumerate}
		\item \label{cla:Q_1-strac:item:V-c} There is a  polynomial $A'_i\in\C[V]_2$, a linear form $v_i \in V$, and  a linear form $c_i$, such that $A_i = \alpha_i Q_0+A'_i + c_i(\epsilon_i c_i + v_i )$, or
		\item \label{cla:Q_1-strac:item:span} $A_i \in \spn{B_0,Q_0,\C[V]_2}$, or
		\item \label{cla:Q_1-strac:item:ideal}  $\rank_s(Q_0)< \rkq$ and	$A_i \in \ideal{V}$.
\end{enumerate}
\end{claim}

\begin{proof}
	As in the proof of \autoref{cla:V-and-P}, let  
\[\cI = \lbrace C_i \in \cQ^1_3 \mid C_i = Q_0+a_ib_i \text{ with } a_i,b_i\in V \rbrace\;.\] Again our assumption implies that $|\cI| > 0.8 m_3$.	
	Let $A \in \cQ_1$. If there are $C_i,C_j\in \cI$ such that $C_i,A$ and $C_j,A$ satisfy $\autoref{thm:structure}\ref{case:rk1}$, then we can repeat the analogous part from the proof of \autoref{cla:V-and-P} and conclude that $A$ satisfies Case~\ref{cla:Q_1-strac:item:V-c}.
	
	
If there is $C_i \in \cI$ such that $A$ and $C_i$ satisfy \autoref{thm:structure}\ref{case:2} then by \autoref{obs:case3} it must be the case that $\rank_s(Q_0) < \rkq$ and $A \in \ideal{V}$ and in particular  Case~\ref{cla:Q_1-strac:item:ideal} holds.
	
	If $A,C_i$ satisfy \autoref{thm:structure}\ref{case:span} then there is a polynomial in $B_t \in \cT_2$ such that $A \in \spn{C_i,B_t}$. If $B_t\in \calL_2$ or $B_t$ satisfies either  Case~\ref{cla:Q_2-strac:item:V-c} or Case~\ref{cla:Q_2-strac:item:V} of  \autoref{cla:Q_2-strac} then $A$ satisfies Case~\ref{cla:Q_1-strac:item:V-c} of the claim.	
	If $B_t$ satisfies Case~\ref{cla:Q_2-strac:item:span} of \autoref{cla:Q_2-strac} then $A$ satisfies Case~\ref{cla:Q_1-strac:item:span}. 	
	If $B_t$ satisfies Case~\ref{cla:Q_2-strac:item:ideal} of \autoref{cla:Q_2-strac} 
	then  $\rank_s(Q_0)< \rkq$ and Case~\ref{cla:Q_1-strac:item:ideal} holds for $A$.
\end{proof}

Finally, we prove the same structure for $\cQ_3$. The proof is very similar to the previous proofs except that here we cannot have $\cQ_3$ as the second set from which we take polynomials. 

\begin{claim}\label{cla:Q_3-strac}
	Let $\cup_{j\in[3]} \cT_j$ and $V$ be as in \autoref{lem:case-V-exsists}. Let $\cJ_2$ and $B_0$ be as in \autoref{cla:Q_2-strac}. Each polynomial $C_i \in \cQ_3$ satisfies one of the following cases:
	\begin{enumerate}
		\item \label{cla:Q_3-strac:item:V-c} There is a quadratic polynomial $C'_i\in \C[V]_2$ and  linear forms $c_i, d_i$, such that $C_i = \alpha_i Q_0+C'_i+c_id_i$, or
		\item \label{cla:Q_3-strac:item:span} $C_i \in \spn{B_0,Q_0,\C[V]_2}$, or
	\item \label{cla:Q_3-strac:item:ideal} $\rank_s(Q_0)< \rkq$ and $C_i \in \ideal{V}$.
\end{enumerate}
\end{claim}
\begin{proof}
	Let $\cI$ be as in the proof of \autoref{cla:V-and-P}.  
Every polynomial in $\cI$ satisfies Case~\ref{cla:Q_3-strac:item:V-c} of \autoref{cla:Q_3-strac}.
	Let $C \in \cQ_3\setminus \cI$. If $Q_0$, $C$ satisfy \autoref{thm:structure}\ref{case:span} then there is $A_i\in \cT_1$ such that $C \in \spn{A_i,Q_0}$. It is not hard to verify that $C$ satisfies the same case as $A_i$.
	

	If $Q_0$, $C$ satisfy \autoref{thm:structure}\ref{case:rk1} then there are linear forms $c$ and $d$ such that $C= \alpha Q_0 +cd$, and $C$ satisfies Case~\ref{cla:Q_3-strac:item:V-c}. Finally, if $Q_0$ and $C$ satisfy  \autoref{thm:structure}\ref{case:2}, then, as before,  Case~\ref{cla:Q_3-strac:item:ideal} holds for $C$.
\end{proof}

We can now prove \autoref{lem:case-V-exsists}.

\begin{proof}[Proof of \autoref{lem:case-V-exsists}] 
	The combination of \autoref{cla:Q_2-strac}, \autoref{cla:Q_1-strac}, and \autoref{cla:Q_3-strac} guarantees that there are $Q_0$ and $B_0$ (if $\cJ_2\neq \emptyset$) such that every polynomial in  $\cup_{j \in [3]}\cQ_j$  is either in $\ideal{V}$ or is of the form $Q' + ab$ for linear forms $a$ and $b$, and a quadratic $Q'\in \spn{Q_0,B_0, \C[V]}_2$. 	Furthermore, if $\rank_s(Q_0) < \rkq$ then $\MS(Q_0)\subseteq V$. We next show that we can apply   \autoref{prop:P,Q_0,C[V],rank1}. For that we have to find appropriate $Q_0$, $P_0$ and $V$ that satisfy \autoref{assume:special-case} (and the rest of the conditions of \autoref{prop:P,Q_0,C[V],rank1}).
	
	Consider the case $\rank_s(Q_0) \geq \rkq$. In this case we have that (in the notation of \autoref{lem:case-V-exsists}) $\tilde{V}=V$ and in particular, $\dim(V)\leq \dimV$.  If there is  a linear combination  $\alpha Q_0+ \beta B_0$ such that $\rank_s(\alpha Q_0+ \beta B_0) \leq 2 \cdot \dim(V) + 20$ then set $V = V + \MS(\alpha Q_0+ \beta B_0)$. It holds that $\dim(V) \leq \dimV + 2 \cdot \dimV + 20 \leq 320$, and $\rank_s(Q_0)\geq \rkq > 2\cdot \dim(V) +20$. In this case if we let $P_0=0$ then $Q_0$, $P_0$ and $V$  satisfy \autoref{assume:special-case}. If no such linear combination of small rank exists then $Q_0$, $P_0=B_0$ and $V$ satisfy \autoref{assume:special-case}.
	
Consider now the case $\rank_s(Q_0)\leq \rkq$. If $\rank_s(B_0) \leq 2 \cdot \dim(V)  +20$ then add $\MS(B_0)$ to $V$. We now get that $Q_0=P_0=0$ and $V$ satisfy the conditions of  \autoref{assume:special-case}. If $\rank_s(B_0) > 2 \cdot \dim(V)  +20$ then we get that $Q_0=0,P_0=B_0$ and $V$ satisfy \autoref{assume:special-case}. 
	
	Thus, in all possible case we get polynomials $Q_0$ and $P_0$ and a subspace $V$ of dimension $O(1)$ such that  the conditions of \autoref{prop:P,Q_0,C[V],rank1} are satisfied. Consequently, $\dim(\cup_{j \in [3]}\cT_j) = O(1)$, as claimed.
\end{proof}

\subsubsection{The case $\cQ_2^2\neq\emptyset$ and no such $V$ exists}\label{sec:q22<>0}

We now handle the case where $\cQ_2^2 \neq \emptyset$ and there is no such  vector space ${V}$. This  is the last case we needed in order to conclude the proof of \autoref{thm:main}.

\begin{lemma}\label{lem:case-noV}
	Let $\cup_{j\in[3]} \cT_j$ satisfy the conditions of \autoref{sec:hard}. Assume further that $\cQ_2^2 \neq \emptyset$ and that for every linear space of linear forms $\tilde{V}$ such that $\dim(\tilde{V})\leq \dimV$, when we set $V= V_{\MS}(\tilde{V},Q_0)$
	it holds that $|\cQ_3^1\cap \spn{Q_0,\C[V]_2}| \leq 0.8 m_3$. Then, $\dim(\cup_{j\in[3]} \cT_j) = O(1).$
\end{lemma}

We prove the lemma by a reduction to \autoref{prop:P,Q_0,C[V],rank1}. To show that the conditions of the proposition hold we prove the next lemma, which is the main focus of this section. 

\begin{lemma}\label{lem:noV-main}
		Let $\cup_{j\in[3]} \cT_j$ satisfy the conditions of \autoref{sec:hard}. Assume further that $\cQ_2^2 \neq \emptyset$ and that for every linear space of linear forms $\tilde{V}$ such that $\dim(\tilde{V})\leq \dimV$, when we set $V= V_{\MS}(\tilde{V},Q_0)$
		it holds that $|\cQ_3^1\cap \spn{Q_0,\C[V]_2}| \leq 0.8 m_3$.  
		
		Then, there is a polynomial $T_0\in \cup_{j\in[3]} \cT_j \cup \{0\}$ and a linear space of linear forms, $\tilde{U}$, such that $\dim(\tilde{U}) \leq \dimV$ and for $U= V_{\MS}(\tilde{U},Q_0)$ the following holds:  every $Q_i\in\cup_{j\in[3]} \cQ_j$ satisfies one of the following statements:
		\begin{enumerate}
			\item $Q_i = a_ib_i + Q'_i$ where $Q'_i \in \spn{Q_0,T_0,\C[U]_2}$, and $a_i,b_i$ are linear forms.\label{case:general-Q}
			\item $\rank_s(Q_0) \leq \rkq$ and $Q_i \in \ideal{U}$.\label{case:low-rank-Q}
		\end{enumerate}
\end{lemma}

Before turning to the proof of \autoref{lem:noV-main} we show how to obtain \autoref{lem:case-noV} from it.

\begin{proof}[Proof of \autoref{lem:case-noV}]
We wish to show that the conditions of  \autoref{prop:P,Q_0,C[V],rank1} are satisfied for our set of polynomials. For this it is enough to prove that $Q_0$, $T_0$ and $U$, from the conclusion of \autoref{lem:noV-main}, satisfy \autoref{assume:special-case}.

We first consider the case $\rank_s(Q_0) \geq \rkq$. If $Q_0$, $P_0=T_0$ and $U$ do not satisfy \autoref{assume:special-case} then 
there is  a nonzero linear combination  $\alpha Q_0+ \beta T_0$ such that $\rank_s(\alpha Q_0+ \beta T_0) \leq 2 \cdot \dimV + 20$. In this case we let $U = U + \MS(\alpha Q_0+ \beta T_0)$. Clearly, $\dim(U) \leq \dimV + 2 \cdot \dimV + 20 = 320$, and $\rank_s(Q_0)\geq \rkq > 2\cdot \dim(U) +20$. It follows that  $Q_0$, $P_0=0$ and $U$  satisfy \autoref{assume:special-case} and the  conditions of  \autoref{prop:P,Q_0,C[V],rank1} hold. 

If $\rank_s(Q_0) < \rkq$ then $\MS(Q_0)\subseteq U$. If $\rank_s(T_0) \leq 2 \cdot \dim(U)  +20$ then add $\MS(T_0)$ to $U$ and   the  conditions of  \autoref{prop:P,Q_0,C[V],rank1} hold for $Q_0=P_0=0$ and $U$. If  $\rank_s(T_0) > 2 \cdot \dim(U)  +20$ then we take $Q_0=0$ and $P_0=T_0$.

Consequently, $\dim(\cup_{j \in [3]}\cT_j) = O(1)$, as claimed.
%
\end{proof}

We now turn our attention to proving \autoref{lem:noV-main}.
Similarly to the case where the special subspace $V$ exists, we prove the desired structure on one set $\cQ_i$ at a time. 
We start by proving \autoref{lem:noV-main} for $\cQ_2$.

\begin{lemma}\label{lem:Q_2-strac-noV-gen}
	Let $\cup_{j\in[3]} \cT_j$, $\cQ_2^2$ and $\cQ_3^1$ be as in \autoref{lem:noV-main}. Then, the statement of \autoref{lem:noV-main} holds for every $B\in \cQ_2$. Furthermore, in this case we have $\dim(\tilde{U}) < \dimV / 2$. 
\end{lemma}

\begin{proof}

We split this proof into two cases. In  the first case we assume that there is a polynomial $B_0 \in \cQ_2^2$ such that every linear combination of $B_0$ and $Q_0$ is of $\rank_s$ strictly greater than $2$. The second case is when no such $B_0$ exist.

\begin{claim}\label{cla:Q_2-strac-noV-rk3}
	Consider the setting of \autoref{lem:Q_2-strac-noV-gen}. Assume further that there is $B_0\in \cQ_2^2$ such that every nonzero linear combination of $B_0$ and $Q_0$ is of $\rank_s\geq 3$.
	
	Then, the statement of \autoref{lem:noV-main} is true for every $B\in\cQ_2$ with $T_0=B_0$. Furthermore, in this case we have $\dim(\tilde{U}) = 0$. 
%
\end{claim}

\begin{claim}\label{cla:Q_2-strac-noV-rk2}
	Consider the setting of \autoref{lem:Q_2-strac-noV-gen}. Assume further that $\cQ_2^2\neq \emptyset$ and for every $B_i\in \cQ_2^2$ there exists a linear combination of $B_i$ and $Q_0$  of $\rank_s$ exactly $2$.
	
	Then, the statement of \autoref{lem:noV-main} is true for every $B\in\cQ_2$ with $T_0=0$. Furthermore, in this case we have $\dim(\tilde{U}) \leq \dimV / 2$.
%
\end{claim}

\autoref{lem:Q_2-strac-noV-gen} clearly follows from the two claims above.
%
\end{proof}

We next prove Claims~\ref{cla:Q_2-strac-noV-rk3} and \ref{cla:Q_2-strac-noV-rk2}. The following notation will be used throughout the rest of this section. 
For a polynomial $P\in \cQ_2^2\cup\cQ_1^2$ we define the following partition of the polynomials in $\cQ_3^1$, where $k\in \{1,2\}$ is such that $P\not \in \cQ_k$: 
	\begin{align}\label{eq:defI^P}
		\cI_1^P &=\left\lbrace C_j \in \cQ_3^1 \;\middle|\;
	\begin{tabular}{@{}l@{}}
	$C_j,P$ satisfy \autoref{thm:structure}\ref{case:rk1}, or\\
	$C_j,P$ span a polynomial in $\calL_k$, or\\
	$C_j,P$ satisfy \autoref{thm:structure}\ref{case:span} and  there is $C_j\neq C_t \in \cQ_3^1$  \\such that  $C_t\in \spn {C_j,P}$
	\end{tabular}
	\right\rbrace \;,
	\\
	\cI_2^P &=\left\lbrace C_j \in \cQ_3^1 \;\middle|\;
	\begin{tabular}{@{}l@{}}
	$C_j,P$ satisfy \autoref{thm:structure}\ref{case:span}, and \\ there is no reducible polynomial in $\spn{C_j,P}$,  and\\ there is no $C_j\neq C_t \in \cQ_3^1$  such that  $C_t\in \spn {C_j,P}$
		\end{tabular}
	\right\rbrace \;, \nonumber
	\\
	\cI_3^P &=\left\lbrace C_j \in \cQ_3^1 \;\middle|\;
	\begin{tabular}{@{}l@{}}
	$C_j,P$ satisfy \autoref{thm:structure}\ref{case:2}
	\nonumber
	\end{tabular}
	\right\rbrace\;. 
	\end{align}
	It is clear that for every $P\in\cQ_2\cup\cQ_3$ it holds that $\cQ_3^1 = \cI_1^{P} \cup \cI_2^{P}\cup \cI_3^{P}$. 

\begin{proof}[Proof of \autoref{cla:Q_2-strac-noV-rk3}]
	To prove the claim we shall assume for a contradiction that there is a polynomial $B\in \cQ_2$ that does not satisfy it.  We shall construct a linear space of linear forms, $\tilde{V}$, such that $\dim(\tilde{V}) \leq \dimV$ and $V = V_\MS(\tilde{V}, Q_0)$ satisfies that $|\cQ_3^1 \cap \spn{Q_0,\C[V]_2}| > 0.8 m_3$ in contradiction to the assumption of \autoref{lem:noV-main}.

		Let $\tilde{V}_1=\{\vec{0}\}$ and set $V_1 = V_{\MS}(\tilde{V}_1 ,Q_0)$. From the assumption in \autoref{lem:noV-main}, we have that $|\cQ_3^1 \cap \spn{Q_0,C[V_1]_2}| \leq 0.8 m_3$. 
	
	If every $B\in \cQ_2$ satisfies $B \in \cQ_2^1$ or, $\rank_s(Q_0) < \rkq$ and $B \in \ideal{V_1}$, then \autoref{lem:noV-main} trivially holds. Thus, from now on consider only  $B\in \cQ_2^2$ that does not satisfy the claim. We show that the existence of this $B$ leads to a contradiction.
	
	Consider the partition defined in Equation~\eqref{eq:defI^P} for $B$. We first analyze $\cI_1^{B} \cup \cI_3^B$ before proving the claim for $B$.
If $\cI_1^{B} \cup \cI_3^B \neq \emptyset$ then there is a polynomial $L \in \spn{Q_0,B}$ such that $\rank_s (L) \leq 2$. Since $B$ does not satisfy the claim it must holds that $\rank_s (L) = 2$. 
Set $\tilde{V}_2= \tilde{V}_1+\MS(L)$ and set  $V_2 = V_\MS(\tilde{V}_2,Q_0)$.

\begin{claim}\label{cla:I1I3-CV}
	We have that \[\cI_1^{B} \cup \cI_3^B \subset \spn{Q_0, \C[V_2]_2}\;.\]
\end{claim}
\begin{proof}
Let $C_j= Q_0 + a_j b_j \in \cI_1^B$ (recall \autoref{eq:defQ^1}). By definition of $\cI_1^B$ it follows that there is a nontrivial linear combination of $B$ and $Q_0$ that is equal to $a_j b_j +ab$ for some linear forms $a$ and $b$. From \autoref{cla:lin-comb-V} it follows that $a_j,b_j \in V_2$ and hence $C_j\in  \spn{Q_0, \C[V_2]_2}$.

We now turn to  $\cI_3^B$. \autoref{obs:case3} implies that if $\rank_s(Q_0) \geq \rkq$ then $\cI_3^B = \emptyset$ so we only have to consider the case where  $\rank_s(Q_0)<\rkq$. In particular, $\MS(Q_0) \subseteq V_1$. Let $C_j \in \cI_3^B$.  Let $c$ and $d$ be linear forms such that $B, C_j \in \ideal{c,d}$. Thus, $c,d \in \MS(B)\subseteq V_2$. The conclusion of \autoref{cla: Q_0,L,2} (when applied to $Q_0$, $C_j=Q_0 + a_j b_j$, $B$, $c$ and $d$) combined with our assumption that  $B \notin \ideal{V_1}$ implies that $a_j,b_j \in \MS(Q_0)+ \MS(B) \subseteq V_2$.
\end{proof}
	
	By our construction we have that $\dim(\tilde{V}_2)\leq 4 <\dimV$. As we just proved that $\cI_1^B\cup \cI_3^B\subset \spn{Q_0,\C[V_2]_2}$, the assumption in \autoref{lem:noV-main} implies that  $n_1 \eqdef |\cI_1^B+ \cI_3^B| \leq 0.8 m_3$. Denote  $n_2 \eqdef |\cI_2^B| = |\cQ_3^1| - n_1 $. 
	
	Our assumption that every nonzero linear combination of $Q_0$ and $B_0$ is of $\rank_s\geq 3$ imply that $B_0$ can only satisfy Case~\ref{case:span} of \autoref{thm:structure} with polynomials in $\cQ_3^1$. 
	For every $C_k\in\cQ_3^1$ pick exactly one polynomial $A_k\in\spn{ B_0,C_k}\cap \cQ_1$ and let $\cA^{B_0}\subset \cQ_1$ be the set containing all these polynomials.  Note that by our assumptions, each such $A_k$ can be associated with at most one such $C_k$ (as otherwise $B_0$ and $Q_0$ will have a nonzero linear combination whose $\rank_s$ is at most $2$).
	Hence, there is a natural one-to-one correspondence $\pi_{B_0}: \cQ_3^1\rightarrow \cA^{B_0}$.
	As $|\cQ_3^1|\geq (1-2\delta)m_3 $ (recall Equation~\eqref{eq:size-Q3}) we get  that $|\cA^{B_0}|\geq  (1-2\delta)m_3$.

	 	A similar argument shows  that there are at least $n_2$ polynomials  in $\cQ_1$ that are of the form $A_j = \alpha_j B + \beta_j C_j$ for some $C_j \in \cI_2^B$. We similarly define the set  $\cA^{B}$ containing such polynomials, where for every $C_j\in\cI_2^B$ we pick exactly one polynomial $A_j\in\spn{ B,C_j}$. The definition of the set $\cI_2^B$ implies that each $A_j$ is associated with at most one $C_j$. Let  $\pi_{B}: \cI_2^B\rightarrow \cA^B$ be the natural one-to-one correspondence between the sets.
	 	  
	 	As $m_1\leq m_3$ we get that $|\cA^B \cap \cA^{B_0}| \geq n_2-2\delta m_3\geq (0.2-4\delta)m_3>0 $. In particular, $\cA^B \cap \cA^{B_0}\neq \emptyset$.

	Let $\tilde{V}_3'$ be the linear space guaranteed by \autoref{cla:lin-rank-r} for $B_0$, $Q_0$ and $r=4$. Recall that $\tilde{V}_3'$ has the property that if $P\in\spn{B_0,Q_0}$ is such that  $\rank_s(P)\leq 4$ then $\MS(P)\subseteq \tilde{V}_3'$.
	Set 	$\tilde{V}_3= \tilde{V}_2 + \tilde{V}_3'$ and $V_3 = V_\MS(\tilde{V}_3,Q_0)$. From \autoref{cla:lin-rank-r}  we get that $\dim(\tilde{V}_3') \leq 32$, and thus $\dim(\tilde{V}_3) \leq 36 <\dimV$.

	Let $A\in \cA^B \cap \cA^{B_0}$. Without loss of generality we can assume, from the definition of $\cA^B \cap \cA^{B_0}$, that $A$ can be represented as both $A = \alpha_1 B_0 + \beta_1 C_1 $ and $A = B + \beta_2 C_2$. Therefore: 
\begin{equation}\label{eq:B1}
	B = A- \beta_2 C_2=\alpha_1 B_0 + \beta_1 C_1 -\beta_2 C_2 = \alpha_1 B_0 + (\beta_1-\beta_2)Q_0 + \beta_1 a_1b_1 - \beta_2 a_2b_2.
\end{equation}	
	If $\rank_s(\beta_1 a_1b_1 - \beta_2 a_2b_2) = 1$ then the claim holds for $B$. Thus, assume that $\rank_s(\beta_1 a_1b_1- \beta_2 a_2b_2)=2$. 
	Let $\tilde{V}_4 = \tilde{V}_3 + \spn{a_1,b_1,a_2,b_2}$ and $V_4 = V_\MS(\tilde{V}_4,Q_0)$.
We are now done with  preparations and  ready to prove that $B$ does satisfy the claim. Specifically, we prove  that if $B$ is not of the form $\alpha B_0 + \beta Q_0 + \tilde{B}$, for some polynomial $\tilde{B}$ such that $\rank_s(\tilde{B})= 1$, then at least $n_2-2\delta m_3$ of the polynomials in $\cI_2^{B}$ (those in $\pi_{B}^{-1}(\cA^{B} \cap \cA^{B_0})$) belong to $\spn{Q_0,\C[V_4]_2}$. This  implies that 
\begin{align*}
|\spn{Q_0,\C[V_4]_2}\cap \cQ_3^1| &\geq |\cI_1^{B}\cup\cI_3^{B}|+|\pi_{B}^{-1}(\cA^{B} \cap \cA^{B_0})|\\ &\geq  n_1 + n_2-2\delta m_3 \geq (1-4\delta)m_3 >0.8m_3
\end{align*}
in contradiction to the assumption in \autoref{lem:noV-main} that such $\tilde{V}_4$ does not exist.
 
Let $A_j \in \cA^{B} \cap \cA^{B_0}$. As before, there is $C_k =Q_0 + a_k b_k \in \cI_2^{B}$ and $C_t=Q_0 + a_t b_t \in \cQ_3^1$ such that $A_j=\alpha_t B_0 + \beta_t C_t = B + \beta_k C_k$. Hence,
\begin{equation}\label{eq:B2}
B = \alpha_t B_0 + (\beta_t-\beta_k)Q_0 + \beta_t a_tb_t - \beta_k a_kb_k \;.
\end{equation}
As before, we can assume that $\rank_s(\beta_t a_tb_t - \beta_k a_kb_k)=2$.  Combining Equations~\eqref{eq:B1} and~\eqref{eq:B2} we get
	\[\alpha_1 B_0 + (\beta_1-\beta_2)Q_0 + \beta_1 a_1b_1 - \beta_2 a_2b_2 = B = \alpha_t B_0+ (\beta_t-\beta_k)Q_0 + \beta_t a_tb_t - \beta_k a_kb_k.\]
	Hence, 
	\[(\alpha_1-\alpha_t) B_0 + (\beta_1-\beta_2 -\beta_t+\beta_k)Q_0 = \beta_t a_tb_t - \beta_k a_kb_k -\beta_1 a_1b_1 + \beta_2 a_2b_2 .\]
	Since 
	\[\rank_s((\alpha_1-\alpha_t) B_0 + (\beta_1-\beta_2 -\beta_t+\beta_k)Q_0)=\rank_s( \beta_t a_tb_t - \beta_k a_kb_k -\beta_1 a_1b_1 + \beta_2 a_2b_2) \leq 4\;,
	\] 
	\autoref{cla:lin-rank-r} implies that $\MS((\alpha_1-\alpha_t) B_0 + (\beta_1-\beta_2 -\beta_t+\beta_k)Q_0) \subseteq \tilde{V}_3\subseteq V_4$. 
	As $a_1,b_1,a_2,b_2\in \tilde{V}_4$, we get that  
	\[\MS(\beta_t a_tb_t - \beta_k a_kb_k )=\MS((\alpha_1-\alpha_t) B_0 + (\beta_1-\beta_2 -\beta_t+\beta_k)Q_0 + \beta_1 a_1b_1 -\beta_2 a_2b_2)\subseteq V_4\;.
	\] 
	Since $\rank_s(\beta_t a_tb_t - \beta_k a_kb_k) = 2$ we conclude that $a_k,b_k,a_t,b_t  \in V_4$ and in particular that $C_k = Q_0 + a_kb_k \in \spn{Q_0,\C[V_4]_2}$. 
	
	As there are at least $n_2-2\delta m_3$ different such $A_j\in  \cA^B \cap \cA^{B_0}$, we conclude that $|\pi_{B}^{-1}( \cA^B \cap \cA^{B_0})| \geq n_2-2\delta m_3$, which leads to the desired contradiction. This completes the proof of  \autoref{cla:Q_2-strac-noV-rk3}.
\end{proof}


We next handle the case that such $B_0$ does not exist.

\begin{proof}[Proof of \autoref{cla:Q_2-strac-noV-rk2}]
	 Set $\tilde{V}_0 = \spn{\vec{0}}$ and $V_0 = V_{\MS}(\tilde{V}_0,Q_0)$. The claim clearly holds for polynomials in $\cQ_2^1\cup \ideal{V_0}$. 
	 
We now describe an iterative process for constructing a linear space $\tilde{U}$ that will satisfy the requirements of the claim. 

Let $B_1 \in \cQ_2^2\setminus \ideal{V_0}$. By definition of $\cQ_2^2$ and the assumption in the claim we have $B_1=\gamma_1 Q_0 + L_1$, where $\rank_s(L)=2$. 
Set $\tilde{V}_1 = \MS(L_1)$ and ${V}_1 = V_0 + \tilde{V}_1$. As in the proof of \autoref{cla:Q_2-strac-noV-rk3}, we consider the partition $\cQ_3^1 = \cI_1^{B_1} \cup \cI_2^{B_1}\cup \cI_3^{B_1}$and conclude that $\cI_1^{B_1} \cup \cI_3^{B_1} \subset \spn{Q_0, \C[V_1]_2}$.
Consequently, $n_1^{B_1} \eqdef |\cI_1^{B_1} \cup \cI_3^{B_1} | \leq 0.8 m_3$ and  $n^{B_1}_2 \eqdef |\cI^{B_1}_2| = |\cQ_3^1| - n^{B_1}_1$. As before, the assumption in \autoref{lem:noV-main} implies that $n^{B_1}_2 > (0.2-2\delta) m_3$.	

As in the proof of \autoref{cla:Q_2-strac-noV-rk3}, it is not hard to see that ${B_1}$ satisfies Case~\ref{case:span} of \autoref{thm:structure} with every $C\in \cI^{B_1}_2$. Furthermore, if $A\in \spn{{B_1},C}\cap \spn{{B_1},C'}$, for $C,C'\in \cI^{B_1}_2$, then $C=C'$. Thus, we can find a set $\cA^{B_1}\subseteq \cQ_1$ of size $|\cA^{B_1}|=|\cI^{B_1}_2|$ such that for every $C\in\cI^{B_1}_2$ we have $\spn{{B_1},C}\cap \cA^{B_1}\neq \emptyset$. As before, there is a natural one-to-one correspondence $\pi_{B_1}: \cI^{B_1}_2\rightarrow \cA^{B_1}$. 
 

Assume that we already found $B_1,\ldots,B_{i-1}$ and constructed $\tilde{V}_1,\ldots,\tilde{V}_{i-1}$, for $i\geq 2$. Denote $V_{i-1} = V_0+\tilde{V}_1+\ldots+\tilde{V}_{i-1}$. Consider a polynomial, $B_i \in \cQ_2^2\setminus \ideal{V_{0}}$ such that $B_i$ is not of the form $B_i=B'_i + a_ib_i$, where $B'_i \in \spn{Q_0,\C[V_{i-1}]_2}$. If no such $B_i$ exists then $\cQ_2^2$ satisfies \autoref{lem:noV-main} with $\tilde{U}=\tilde{V}_1+\ldots +\tilde{V}_{i-1}$. As we shall soon see, $\dim(\tilde{U})<\dimV/2$. This is clearly the case for $i=2$.

As before  denote $B_i=\gamma_i Q_0 + L_i$ with $\rank_s(L_i)=2$, and set $\tilde{V}_i = \MS(L_i)$ and $V_i=\tilde{V}_i + V_{i-1}$. Consider the partition $\cQ_3^1= \cI^{B_i}_1\cup\cI^{B_i}_2\cup \cI^{B_i}_3$, and, as above, define the set $\cA^{B_i}\subseteq \cQ_1$ and denote by $\pi_{B_i}:\cI_2^{B_i}\rightarrow \cA^{B_i}$ the natural bijection. We again conclude that $n^{B_i}_1 = |\cI_1^{B_i}\cup \cI_3^{B_i}| \leq 0.8 m_3$ and that $n^{B_i}_2 =|\cA^{B_i}|= |\cI^{B_i}_2| = |\cQ_3^1| - n^{B_i}_1 $. 
The next claim shows that $\cA^{B_i}$ is far from being contained in $ \cup_{\ell \leq i-1}\cA^{B_\ell}$.
\begin{claim}\label{cla:Bi-B}
	For $i<10$, if the process reached step $i$ then   $\cA^{B_i}\setminus \cup_{\ell \leq i-1}\cA^{B_\ell} \geq (0.2-2\delta) m_3$. 
\end{claim}
\begin{proof}
	Consider $A_j \in \cA^{B_i}\cap (\cup_{\ell \leq i-1}\cA^{B_\ell})$.
	Assume $A_j\in \cA^{B_\ell}$, for $1\leq \ell\leq i-1$. Denote, without loss of generality, $A_j = \alpha_t {B_\ell} + \beta_t C_t = B_i + \beta_k C_k$ for $C_t\in \cI_2^{B_\ell}$ and $C_k\in \cI_2^{B_i}$. Using the notation of Equation~\eqref{eq:defQ^1} we get
	\[\gamma_i Q_0 + L_i=B_i = \alpha_t {B_\ell} + \beta_t C_t -\beta_k C_k = (\alpha_t +\beta_t - \beta_k)Q_0 +\alpha_t L_\ell + \beta_t a_tb_t -\beta_ka_kb_k.\]
	In particular,  $\rank_s(\beta_t a_tb_t -\alpha_ka_kb_k) = 2$, from the choice of $B_i$ (as $L_\ell\in\C[V_{i-1}]_2$). 
	
	If $\rank_s(Q_0)\geq \rkq$ then $\gamma_i-(\alpha_t +\beta_t - \beta_k) = 0$ and $L_i -\alpha_t L_\ell = \beta_t a_tb_t -\beta_ka_kb_k$, implying  $ a_t,b_t,a_k,b_k \in V_i$.
	
	If $\rank_s(Q_0)\leq \rkq$ then $\MS(Q_0)=V_0\subseteq V_i$ and  it again follows that  $ a_t,b_t,a_k,b_k \in V_i$. Thus, in either case $ a_t,b_t,a_k,b_k \in V_i$.
	
	To conclude, we just proved that every $C\in \cI^{B_i}_2$, such that  $\pi_{B_i}(C) \in \cA^{B_i}\cap(\cup_{\ell \leq i-1}\cA^{B_\ell})$, satisfies $C \in \spn{Q_0,\C[V_i]_2}$. Therefore, 
	 $\cI_1^{B_i} \cup \cI_3^{B_i} \cup \pi_{B_i}^{-1}(\cA^{B_i} \cap \cup_{\ell \leq i-1}\cA^{B_\ell}) \subseteq \spn{Q_0,\C[V_i]_2}$. 	 
	 As $V_i=V_\MS\left(\sum_{j=1}^{i}\tilde{V}_j,Q_0\right)$ and $\dim\left(\sum_{j=1}^{i}\tilde{V}_j \right)\leq 4i <40$, the assumption in  \autoref{lem:noV-main} implies that $|\cI_1^{B_i} \cup \cI_3^{B_i} \cup \pi_{B_i}^{-1}(\cA^{B_i} \cap (\cup_{\ell \leq i-1}\cA^{B_\ell}))|\leq 0.8 m_3$. 
	 
	 Denote $\cD_1 = \pi_{B_i}^{-1}(\cA^{B_i} \cap (\cup_{\ell \leq i-1}\cA^{B_\ell}))$ and $\cD_2 = \pi_{B_i}^{-1}(\cA^{B_i} \setminus (\cup_{\ell \leq i-1}\cA^{B_\ell}))$. Clearly $\cD_1 \cup \cD_2 = \cI^{B_i}_2$ and we have proved that $|\cD_1\cup\cI^{B_i}_1\cup \cI^{B_i}_3| \leq 0.8m_3$. Since $|\cD_2 \cup \cD_1\cup\cI^{B_i}_1\cup \cI^{B_i}_3| = |\cQ_3^1| \geq (1-2\delta)m_3$ we conclude that   
	 $|\cA^{B_i} \setminus \cup_{k \leq i-1}\cA^{B_k}|= |\cD_2|\geq (0.2-2\delta) m_3 $, as claimed.
\end{proof}

	It follows from the claim that by adding $\MS(L_i)$ to $V_{i-1}$ we covered at least $(0.2 -2\delta )m_3$ new polynomials in $\cQ_1$. That is, there are at least $(0.2-2\delta)m_3$ more polynomials in $\cup_{\ell \leq i}\cA^{B_\ell}$
	compared to $\cup_{\ell \leq i-1}\cA^{B_\ell}$ . Therefore, $|\cup_{\ell \leq i}\cA^{B_\ell}| \geq (0.2-2\delta)\cdot i \cdot m_3 \geq (0.2-2\delta)\cdot i \cdot m_1$. This implies that the process can run for at most $5$ steps (as $12\delta < 0.2$).
	
	When the process terminates we get a subspace $\tilde{U}=\left( \sum_{j=1}^{i}\tilde{V}_j\right)$ such that every $B_j$ satisfies   \autoref{lem:noV-main} with $U=V_\MS(\tilde{U},Q_0)$. Furthermore, as each $\dim(\tilde{V}_j) \leq 4$ it holds that $\dim(\tilde{U}) \leq 20 < \dimV/2$, as claimed. 
	\end{proof}

This also completes the proof of \autoref{lem:Q_2-strac-noV-gen}. 
We now prove that \autoref{lem:noV-main} holds for $\cQ_1$.


\begin{claim}\label{cla:Q_1-struc-noV}
	Let $\cup_{j\in[3]} \cT_j$, $\cQ_2^2$ and $\cQ_3^1$ be as in \autoref{lem:noV-main}. Then, the statement of \autoref{lem:noV-main} holds for every $A\in \cQ_1$.	
\end{claim}

\begin{proof}
	

	\autoref{lem:Q_2-strac-noV-gen} guarantees the existence of a  space of linear forms $\tilde{U}$, such that $\dim(\tilde{U})< \dimV/2$ and for $U = V_\MS(\tilde{U},Q_0)$, all the polynomials in $\cQ_2$ are either of the form  $B_j = B'_j + a_jb_j$ for $B_j \in \spn{Q_0, B_0, \C[U]_2}$ or, when $\rank_s(Q_0) \leq \rkq$, they can also satisfy $B_j\in \ideal{U}$. 
	
	 Let $\tilde{V_0}=\tilde{U}$ and set $V_0 = V_\MS(\tilde{V_0}, Q_0)=U$. The claim holds for polynomials in $\cQ_1^1$, and, if $\rank_s(Q_0)<\rkq$, then for polynomials that are in  $\ideal{V_0}$ as well.

	As in the proof of \autoref{cla:Q_2-strac-noV-rk3} 	we shall assume for a contradiction that there is a polynomial $A\in \cQ_1$ that does not satisfy it.  We shall construct a linear space of linear forms, $\tilde{V}$, such that $\dim(\tilde{V}) \leq \dimV$ and $V = V_\MS(\tilde{V}, Q_0)$ satisfies that $|\cQ_3^1 \cap \spn{Q_0,\C[V]_2}| > 0.8 m_3$ in contradiction to the assumption of \autoref{lem:noV-main}. 
As before, the construction of $V$ we depend on whether $\rank_s(Q_0) \leq \rkq$ or not.

	 So assume towards a contradiction that there exists $A\in \cQ_1$ that does not satisfy the conclusion of the claim. Similarly to the proof of \autoref{lem:Q_2-strac-noV-gen}, consider the partition $\cQ_3^1 = \cI^{A}_1 \cup \cI^{A}_2 \cup\cI^{A}_3$.
If $\cI^{A}_1\cup \cI^{A}_3\neq \emptyset$ then there is a polynomial $L \in \spn{Q_0,A}$ such that $\rank_s (L) = 2$. In this case we set $\tilde{V}_1= \tilde{U} +\MS(L)$ and $V_1 = V_\MS(\tilde{V}_1, Q_0)$. Thus, $\dim(\tilde{V}_1) < \dimV /2 + 4$. By the same arguments as  in the proof of \autoref{cla:I1I3-CV}, we deduce that  $\cI^{A}_1 \cup \cI^{A}_3  \subset \spn{Q_0,\C[V_1]_2}$. 

If  $\rank_s(Q_0) \leq \rkq$ then we need to be more careful and consider the case where  there are linear forms $c$ and $d$, and a polynomial $A' \in \ideal{U}$, such that $A = A' + cd$. In this case we set $\tilde{V}_2 = \tilde{V}_1 + \spn{c, d}$ and  $V_2 = V_\MS(\tilde{V}_2, Q_0)$. Our assumption that $A\notin \ideal{V_0} = \ideal{U}$ implies that $c,d\notin U$.  Observe that if $A=A''+ef$ is another such representation of $A$ then $cd-ef\in\ideal{U}$ and \autoref{cla:rank-2-in-V} implies that $\spn{c,d}+ U = \spn{e,f}+ U $, and in particular, $V_2$ is well defined, regardless of which representation we chose.

If $\rank_s(Q_0) > \rkq$ or there are no such $c,d$ and $A'$ then we let $\tilde{V}_2 = \tilde{V}_1$ and $V_2 = V_1$. In either case we have that  $\dim(\tilde{V}_2) \leq \dimV/2 + 6$.
	
It is not hard to see that ${A}$ satisfies Case~\ref{case:span} of \autoref{thm:structure} with every $C\in \cI^{A}_2$. Furthermore, if there is a polynomial $B \in \cQ_2$, such that $ B\in \spn{{A},C}\cap \spn{{A},C'}$, for some $C,C'\in \cI^{A}_2$, then $C=C'$. Hence, there is a set $\cB^{A}\subseteq \cQ_2$, of size $|\cB^{A}|=|\cI^{A}_2|$, such that every $C\in\cI^{A}_2$ satisfies $\spn{{A},C}\cap \cB^{A}\neq \emptyset$. As before, there is a natural one-to-one correspondence $\pi_{A}: \cI^{A}_2\rightarrow \cB^{A}$.

	We now analyze the structure of $ \cI^{A}_2$ based on which case of \autoref{lem:noV-main} polynomials in $\cB^{A}$ satisfy (as we proved \autoref{lem:Q_2-strac-noV-gen} we know they satisfy \autoref{lem:noV-main}). 
	
	We first consider the case $\rank_s(Q_0)\leq \rkq$ (and in particular, $\MS(Q_0)\subseteq U$) and denote with $\cB^{A}_U \eqdef \cB^{A} \cap \ideal{U}$ the set of all polynomials in $\cB^{A}$ that satisfy Case~\ref{case:low-rank-Q} of \autoref{lem:noV-main}.	
	Next we show that $\pi_{A}^{-1}(\cB^{A}_U) \subseteq \spn{Q_0,\C[V_2]_2}$.
Let $C \in \cI^{A}_2$ be such  that $B\eqdef \pi_{A}(C)  \in \cB^{A}_U$. Denote $C= Q_0 + c_1 c_2$ and $B=\alpha A + \beta C$.
	Then, 
\[B=\alpha A + \beta  C = \alpha A + \beta Q_0 + \beta c_1 c_2 \in \ideal{U}\;.\] 
As $\MS(Q_0)\subseteq U$, it follows that  $A''\eqdef \alpha A + \beta c_1 c_2  \in \ideal{U}$. By definition of $V_2$, and as $V_2$ is well defined,  it follows that $c_1,c_2\in V_2$ and thus $C \in \spn{Q_0,\C[V_2]_2}$, as we wanted to prove.  Moreover, from the fact that $V_1\subseteq V_2$, we conclude that $\pi_{A}^{-1}(\cB^{A}_U) \cup \cI^{A}_1 \cup \cI^{A}_3 \subseteq \spn{Q_0,\C[V_2]_2}$.
Hence,  as $\dim(\tilde{V}_2) <\dimV$,  the assumption of \autoref{cla:Q_1-struc-noV} implies that $|\pi_{A}^{-1}(\cB^{A}_U) \cup \cI^{A}_1 \cup \cI^{A}_3| \leq 0.8 m_3 < |\cQ_3^1|$.



We next consider polynomials in $\cB^A$  satisfying Case~\ref{case:general-Q} of \autoref{lem:noV-main}. Let $C\in  \cI^{A}_2 \setminus \pi_{A}^{-1}(\cB^{A}_U)$. Let $\pi_{A}(C) = B$. Then $B$  satisfies Case~\ref{case:general-Q} of \autoref{lem:noV-main} and, by definition, there are linear forms $b_1$ and $b_2$, and a polynomial $B' \in\spn{Q_0,T_0, \C[U]_2}$ ($T_0$ is as in \autoref{lem:noV-main}, which we know holds for $\cQ_2$), such that $B= B' +b_1b_2$. Denote, without loss of generality, $B=  A + \beta C$, and $C = Q_0+c_1c_2$ (this holds as $C\in \cQ_3^1$). We have that
\begin{equation}\label{eq:A-C_1}
	 A  = B-\beta C=  (B' -\beta Q_0)+ b_1b_2 - \beta c_1c_2.
\end{equation}
Since $B' -\beta_1 Q_0 \in \spn{Q_0,T_0, \C[U]_2}$ and  we assumed that $A$ does not satisfy the conclusion of the claim, we conclude that  $\rank_s(b_1b_2 - \beta  c_1c_2) = 2$.

Set $\tilde{V}_3 = \tilde{V}_2 + \spn{b_1,b_2,c_1,c_2}$. Clearly, $\dim(\tilde{V}_3) \leq \dim(\tilde{V}_2)+4< \dimV/2 +10$.

If $T_0 \neq 0$ then let $V'$ be the linear space guaranteed by \autoref{cla:lin-rank-r-U} for $T_0$, $Q_0$, $U$ and $r=4$.
Set $\tilde{V}= V' + \tilde{V}_3$. From the bound in \autoref{cla:lin-rank-r-U}  we get that $\dim(V') \leq 32$, and thus $\dim(\tilde{V})  <\dimV/2 + 42 < \dimV$. Set $V = V_\MS(\tilde{V}, Q_0)$. If $T_0=0$ then we let $V=V_\MS(\tilde{V}_3, Q_0)$.

We next show that $Q_3^1 \subseteq \spn{Q_0,\C[V]_2}$ in contradiction to the assumption in \autoref{lem:noV-main}. As $V_2 \subseteq V$, it  suffices to  show that $\cI^{A}_2 \setminus \pi_{A}^{-1}(\cB^{A}_U)\subseteq \spn{Q_0,\C[V]_2}.$

Let $C_1 \in \cI^{A}_2 \setminus \pi_{A}^{-1}(\cB^{A}_U)$. As before we denote  $C_1 = Q_0+d_1d_2$ and let $B_1=\pi_A(C_1)$. By \autoref{lem:Q_2-strac-noV-gen}, there are linear forms $e_1$ and $e_2$, and a polynomial $B'_1 \in\spn{Q_0,T_0, \C[U]_2}$, such that $B_1 = B'_1 +e_1e_2$. Also denote, without loss of generality, $B_1 = A + \beta_1 C_1$. We have that,
\[ A  = B_1 -  \beta_1 C_1 = B'_1 + e_1e_2 - \beta_1 (Q_0+d_1d_2) =  (B'_1 -\beta_1 Q_0)+ e_1e_2 - \beta_1 d_1d_2 \;.\]
As before, since $A$ does not satisfy the conclusion of the claim, it follows that $\rank_s(e_1e_2 - \beta_1 d_1d_2) = 2$. Combining with \autoref{eq:A-C_1}, we obtain that
\[(B' -\beta Q_0)+ b_1b_2 - \beta c_1c_2= A  =(B'_1 -\beta_1 Q_0)+ e_1e_2 - \beta_1 d_1d_2,\]
and therefore,
\begin{equation}\label{eq:B,B'}
(B'- B'_1 +\beta_1  Q_0-\beta Q_0) =  e_1e_2 - \beta_1 d_1d_2 -(b_1b_2 - \beta c_1c_2)\;.
\end{equation}
Accordingly, $\rank_s(B'- B'_1 +\beta_1 Q_0-\beta Q_0)\leq 4$. We next show that $d_1,d_2\in V$, which implies $C_1\in \spn{Q_0, \C[V]_2}$, as we wanted to prove.
For this we first prove that $\MS(B'- B'_1 +\beta_1 Q_0-\beta Q_0) \subseteq V$.

Recall that in the case where $T_0 \neq 0$ we defined a subspace $V' \subseteq V$ using \autoref{cla:lin-rank-r-U}. As  $B'- B'_1 +\beta_1 Q_0-\beta Q_0\in \spn{Q_0,T_0,\C[U]_2}$, it follows that $\MS(B'- B'_1 +\beta_1 Q_0-\beta Q_0) \subseteq V' \subseteq V$.

If $T_0 =0$ then we consider two cases based on $\rank_s(Q_0)$. If 
$\rank_s(Q_0) > \rkq$ then, as $B',B'_1\in \spn{Q_0,\C[U]_2}$, Equation~\eqref{eq:B,B'} implies that that $B'- B'_1 +\beta_1 Q_0-\beta Q_0 \in \C[U]_2\subseteq \C[V]_2$.
Finally, if $T_0 =0$ and $\rank_s(Q_0) \leq \rkq$, then $\MS(Q_0)\subseteq V$ and thus $\MS(B'- B'_1 +\beta_1 Q_0-\beta Q_0)\subseteq V$. To conclude, in all possible cases it holds that 
\begin{equation}\label{eq:B-B'inV}
\MS(B'- B'_1 +\beta_1 Q_0-\beta Q_0) \subseteq V\;.
\end{equation}
As $b_1,b_2, c_1,c_2 \in V$, Equations~\ref{eq:B,B'} and \ref{eq:B-B'inV} imply that
 \[e_1e_2 - \beta_1 d_1d_2 = (B'- B'_1 +\beta_1 Q_0-\beta Q_0)  +(b_1b_2 - \beta c_1c_2) \in \C[V]_2 \;.\]
As $\rank_s(e_1e_2 - \beta_1 d_1d_2) = 2$,  we conclude that $d_1,d_2 \in V$. In particular, this proves that $C_1 \in \spn{Q_0, \C[V]_2}$, as we wanted to show. This completes the proof of \autoref{cla:Q_1-struc-noV}.
\end{proof}

So far we proved that $\cQ_1$ and $\cQ_2$ satisfy \autoref{lem:noV-main} and we now show it for $\cQ_3$. This will conclude the proof of  \autoref{lem:noV-main} and with it the last case left in the proof of \autoref{thm:main}.

\begin{claim}\label{cla:Q_3-struc-noV}
	Let $\cup_{j\in[3]} \cT_j$, $\cQ_2^2$ and $\cQ_3^1$ be as in \autoref{lem:noV-main}. Then, the statement of \autoref{lem:noV-main} holds for every $C\in \cQ_3$.	
\end{claim}

\begin{proof}
	From \autoref{lem:Q_2-strac-noV-gen} and 	\autoref{cla:Q_1-struc-noV}, there is a linear space of linear forms, $U$ such that the statement of \autoref{lem:noV-main} holds for $\cQ_2$ and $\cQ_1$. Let $C \in \cQ_3$. It is not hard to verify that no matter which case of \autoref{thm:structure} $C$ and $Q_0$  satisfy, $C$ has the form stated in the claim.
\end{proof}

\begin{proof}[Proof of \autoref{lem:noV-main}]
	The lemma follows immediately from \autoref{lem:Q_2-strac-noV-gen}, \autoref{cla:Q_1-struc-noV}, and \autoref{cla:Q_3-struc-noV}. 
\end{proof}

This concludes the proof of \autoref{thm:main}.

\fi

\section{Missing proofs from Section~\ref{sec:robust-EK}}\label{sec:EK-proofs}

We first give the proof of \autoref{thm:ek-not-disjoint}. The proof follows the lines of the proof in \cite{EdelsteinKelly66}.

We shall use the following notation.

\begin{definition}\label{def:line}
For two points $p_1\neq p_2\in \C^n$ we denote by $\line{p_1,p_2}$ the set of points on the unique line that passes though $p_1$ and $p_2$. Similarly, for three non colinear points $p_1\neq p_2\neq p_3 \in \C^n$ we denote by $\plane{p_1,p_2,p_3} = \plane{\line{p_1,p_2},p_3}$ the set of points on the unique plane determined by $p_1$, $p_2$ and $p_3$. 
\end{definition}

For the proof of \autoref{thm:ek-not-disjoint} we require the following simple corollary of the Sylvester-Gallai theorem.

	\begin{lemma}\label{lem:pen-of-lines}
		Let $\{\MVar{\ell}{k}\}$ be a set of lines, passing through the same point $p$, that are not all on the same plane. Then, there is a plane that contains exactly two of the lines.
	\end{lemma}
\begin{proof}
	Let $H$ be an hyperplane in general position. I.e. $H$ does not contain $p$ and every line from $\MVar{\ell}{k}$ intersects $H$ at a single point, which we denote $p_i = H \cap \ell_i$. By our assumption the points $\{p_i\}_1^k$ are not colinear. Thus, by the Sylvester-Gallai theorem, there is an ordinary line that passes through exactly two of the points, without loss of generality $p_1$ and $p_2$.  It follows that the plane through $\ell_1,\ell_2$ does not contain any other line.
\end{proof}

\begin{proof}[Proof of \autoref{thm:ek-not-disjoint}]
Assume towards a contradiction that the dimension of $\cup_{i\in[k]} \cS_i$ is at least $4$. Without loss of generality, let $p_1\in \cS_1$ and $p_2\in \cS_2$ be such that $p_1\neq p_2$. For every $p \in \cup_{i\in[k]} \cS_i \setminus \line{p_1,p_2}$ let $H_p = \plane{p,\line{p_1,p_2}}$. By the assumption in the statement of the theorem, the set $H_p \setminus \line{p_1,p_2}$ contains points from at least two different sets. 

By the assumption on the dimension of $\cup_{i\in[k]} \cS_i$ and \autoref{lem:pen-of-lines} it follows that there is an affine subspace of dimension $3$, $\Gamma$, such that $\Gamma$ contains exactly two of the planes in $\{H_p\}$ (indeed, we can achieve this by applying the lemma on the intersection of the planes with an hyperplane in general position). Denote these planes by $H_1$ and $H_2$. Thus, $\Gamma \cap (\cup_{i\in[k]} \cS_i )\subset H_1\cup H_2$. 

Let $p \in H_1 \setminus \line{p_1,p_2}$ and $p'\in H_2 \setminus \line{p_1,p_2}$ be such that $p$ and $p'$ are from different sets. Such points exist as each set $H_p \setminus \line{p_1,p_2}$ contains points from two different sets.
Observe that $\line{p,p'}\subset \Gamma$ but $\line{p,p'} \cap H_1 = p$ and $\line{p,p'} \cap H_2 = p'$ and thus $\line{p,p'}$ does not contain a third point different from $p$ and $p'$, in contradiction to our assumption on $\MVar{\cS}{k}$. 
\end{proof}

Next we prove \autoref{thm:partial-EK-robust}.
The proof is almost identical to the proof of Theorem 1.9 in \cite{DBLP:conf/stoc/Shpilka19}.

In this proof we use following version of Chernoff bound. See e.g. Theorem 4.5 in \cite{MU05-book}.

\begin{theorem}[Chernoff bound]\label{thm:chernoff}
	Suppose $X_1,\ldots, X_n$ are independent indicator random variables. Let $\mu = E[X_i]$ be the expectation of $X_i$. Then,
	$$\Pr\left[\sum_{i=1}^{n}X_i < \frac{1}{2}n\mu\right] < \exp(-\frac{1}{8}n\mu).$$
\end{theorem}

\begin{proof}[Proof of \autoref{thm:partial-EK-robust}]
	Denote $|\cT_i|=m_i$. Assume w.l.o.g. that $|\cT_1| \geq   |\cT_2| \geq |\cT_3|$. The proof distinguishes two cases. The first is when  $|\cT_3|$ is not too small and the second case is when it is much smaller than the largest set. 
	
	\begin{enumerate}
		\item {\bf Case $m_3 > m_1^{1/3}$: } \hfill
		
		Let $\cT'_1\subset \cT_1$ be a random subset, where each element is sampled with probability $m_2/m_1 = |\cT_2|/|\cT_1|$. By the Chernoff bound (\autoref{thm:chernoff}) we get that, w.h.p., the size of the set is at most, say, $2m_2$. Further, the Chernoff bound also implies that for every $p\in \cT_2$ there are at least $(\delta/2)\cdot m_2$ points in $\cT'_1$ that together with $p$ span a point in $\cT_3$. Similarly, for every $p\in \cT_3$ there are at least $(\delta/2)\cdot m_2$ points in $\cT'_1$ that together with $p$ span a point in $\cT_2$. Clearly, we also have that for every point $p\in\cT'_1$ there are $\delta m_2$ points in $\cT_2$ that together with $p$ span a point in $\cT_3$. Thus, the set $\cT'_1\cup \cT_2\cup \cT_3$ is a $(\delta/8)$-SG configuration and hence has dimension $O(1/\delta)$ by \autoref{thm:robustSG}. 
		
		Let $V$ be a subspace of dimension $O(1/\delta)$ containing all these points. Note that in particular, $\cT_2,\cT_3\subset V$. As every point $p\in \cT_1$ is a linear combination of points in $\cT_2\cup\cT_3$ it follows that the whole set has dimension $O(1/\delta)$.
	
		\item {\bf Case $  m_3  \leq m_1^{1/3}$:}\hfill
		
		In this case we may not be able to use the sampling approach from earlier as $m_2$ can be too small and the Chernoff argument from above will not hold. 
		
		We say that a point $p_1\in \cT_1$ is a neighbor of a point $p\in \cT_2\cup \cT_3$ if the space spaned by $p$ and $p_1$ intersects the third set. Denote with  $\Gamma_1(p)$ the neighborhood of  a point $p\in\cT_2\cup\cT_3$ in $\cT_1$.

		\sloppy
		Every two points $p\in\cT_2$  and $q\in\cT_3$ define a two-dimensional space that we denote  $V(p,q)=\spn{p,q}$.\footnote{We can assume without loss of generality that $\vec{0}$ is not one of our points.} 
		
		Fix $p\in \cT_2$ and consider those spaces $V(p,q)$ that contain points from $\cT_1$. Clearly there are at most $|\cT_3|$ such spaces. Any two different subspaces $V(p,q_1)$ and $V(p,q_2)$ have intersection of dimension $1$ (it is $\spn{p}$) and by the assumption in the theorem the union $\cup_{q\in\cT_3}V(p,q)$ covers at least $\delta m_1$ points of $\cT_1$.  Indeed, $\delta m_1$ points  $q_1\in \cT_1$ span a point in $\cT_3$ together with $p$. As our points are pairwise independent, it is not hard to see that if $q_3 \in \spn{p,q_1}$ then $q_1 \in \spn{p,q_3}=V(p,q_3)$
		
		For each  subspace $V(p,q)$ consider the set $V(p,q)_1 = V(p,q)  \cap \cT_1$. 
		
		\begin{claim}\label{cla:T1-intersect}
			Any two such spaces $V(p,q_1)$ and $V(p,q_2)$ satisfy that either $V(p,q_1)_1= V(p,q_2)_1$ or $V(p,q_1)_1\cap V(p,q_2)_1=\emptyset$. 
		\end{claim}
		
		\begin{proof}
			If there was a point $p'\in V(p,q_1)_1\cap V(p,q_2)_1$ then both $V(p,q_1)$ and $V(p,q_2)$ would contain $p,p'$ and as $p$ and $p'$ are linearly independent (since they belong to  $\cT_i$'s they are not the same point) we get that $\spn{p,p'}=V(p,q_1)=V(p,q_2)$. In particular, $V(p,q_1)_1= V(p,q_2)_1$.
		\end{proof}

		As a conclusion we see that at most $100/\delta^2$ different spaces $\{V(p,q)\}_q$ have intersection of size at least $\delta^2/100 \cdot m_1$ with $\cT_1$. Let $\cI$ contain $p$ and a point from each of the sets $\{V(p,q)_1\}$ that have size at least $\delta^2/100 \cdot m_1$. Clearly $|\cI| = O(1/\delta^2)$. We now repeat the following process. As long as $\cT_2 \not\subset \spn{\cI}$ we pick a point  $p'\in \cT_2 \setminus \spn{\cI}$ and add it to $\cI$ along with a point from every large set $V(p',q)_1$. I.e., we add a point, different from $p'$, from each subset satisfying $|V(p',q)_1|\geq \delta^2/100\cdot m_1$. We repeat this process until no such $p'$ exists.
		
		We next show that this process must terminate after $O(1/\delta)$ steps and that at the end $|\cI| = O(1/\delta^3)$. To show that the process terminates quickly we prove that if $p_k\in \cT_2$ is the point that was picked at the $k$'th step then $|\Gamma_1(p_k)\setminus \cup_{i\in[k-1]} \Gamma_1(p_i)| \geq (\delta/2)m_1$. Thus, every step covers at least $\delta/2$ fraction of new points in $\cT_1$ and thus the process must end after at most $O(1/\delta)$ steps. 
		
		\begin{claim}\label{cla:one-large-intersect}
			Let $p_i\in\cT_2$, for $i\in [k-1]$ be the point picked at the $i$th step. If  the intersection of $V(p_k,q)_1$ with $V(p_i,q')_1$, for any $q,q'\in \cT_3$, has size larger than $1$ then $V(p_k,q)= V(p_i,q')$ (and in particular, $V(p_k,q)_1= V(p_i,q')_1$) and $|V(p_k,q)_1| \leq \delta^2/100 \cdot m_1$.
			
			Moreover, if there is another pair of points  $q'',q'''\in\cT^3$ satisfying  $|V(p_k,q'')_1\cap  V(p_i,q''')_1|> 1$ then it must be the case that $V(p_i,q')=V(p_i,q''')$.
		\end{claim}
		
		\begin{proof}
			If the intersection of $V(p_k,q)_1$ with $V(p_i,q')_1$ has size at least $2$ then by an argument similar to the proof of \autoref{cla:T1-intersect} we would get that $V(p_k,q) = V(p_i,q')$. To see that in this case the size of $V(p_i,q')_1$ is not too large we note that by our process, if $|V(p_i,q')_1|\geq \delta^2/100 \cdot m_1$ then $\cI$ contains at least two points from $V(p_i,q')_1$. Hence, $p_k\in V(p_i,q')\subset \spn{\cI}$ in contradiction to the choice of $p_k$.
			
			To prove the moreover part we note that in the case of large intersection, since $V(p_k,q) = V(p_i,q')$, we have that $p_k,p_i\in V(p_i,q')$. If there was another pair $(q'',q''')$ so that $|V(p_k,q'')_1 \cap V(p_i,q''')_1|>1$ then we would similarly get that $p_k,p_i\in V(p_i,q''')$. By pairwise linear independence of the points in our sets this implies that $V(p_i,q')=V(p_i,q''')$.
		\end{proof}


		
		\begin{corollary}\label{cor:neighbor-grow}
			Let $i\in[k-1]$ then 
			$$|\Gamma_1(p_k)\cap \Gamma_1(p_i)|\leq \delta^2/100 \cdot m_1 + m_3^2.$$ 
		\end{corollary}
		
		\begin{proof}
			The proof follows immediately from \autoref{cla:one-large-intersect}. Indeed, the claim assures that there is at most one subspace $V(p_k,q)$ that has intersection of size larger than $1$ with any $V(p_i,q')_1$ (and that there is at most one such subspace  $V(p_i,q')$) and that whenever the intersection size is larger than $1$ it is upper bounded by $\delta^2/100 \cdot m_1$. As there are at most $m_3^2$ pairs $(q,q')\in\cT_3^2$ the claim follows.
		\end{proof}
		
		
		The corollary implies that
		$$|\Gamma_1(p_k)\cap \left( \cup_{i\in[k-1]} \Gamma_1(p_i)\right) | \leq k((\delta^2/100)m_1 + m_3^2) < (\delta/2)\cdot m_1,$$ 
		where the last inequality holds for, say, $k<10/\delta$.\footnote{It is here that we use the fact that we are in the case $  m_3  \leq m_1^{1/3}$.}
		As $|\Gamma_1(p_k)|\geq \delta \cdot m_1$, for each $k$, it follows that after $k<10/\delta$ steps 
		$$|\cup_{i\in[k]} \Gamma_1(p_i)| > k(\delta/2)m_1.$$
		In particular, the process must end after at most $2/\delta$ steps. 
		
		As each steps adds to $\cI$ at most $O(1/\delta^2)$ vectors, at the end we have that $|\cI| = O(1/\delta^3)$ and every $p\in\cT_2$ is in the span of $\cI$. 
		
		Now that we have proved that $\cT_2$ has small dimension we conclude as follows. We find a maximal subset of $\cT_3$ whose neighborhoods inside $\cT_1$ are disjoint. As each neighborhood has size at least $ \delta \cdot m_1$ it follows there the size of the subset is at most $O(1/\delta)$. We add those $O(1/\delta)$ points to $\cI$ and let $V=\spn{\cI}$. Clearly $\dim(V) = O(1/\delta^3)$.
		
		\begin{claim}
			$\cup_i \cT_i \subset V$.
		\end{claim}
		
		\begin{proof}
			We first note that if $p\in \cT_1$ is in the neighborhood of some $p'\in\cI\cap \cT_3$ then $p\in V$. Indeed, the subspace spanned by $p'$ and $p$ intersects $\cT_2$. I.e. there is $q\in \cT_2$ that is equal to $\alpha p + \beta p'$, where from pairwise independence both $\alpha\neq0$ and $\beta\neq 0$. As both $p'\in V$ and $\cT_2\subset V$ we get that $p\in V$ as well.
			
			We now have that the neighborhood of every $p\in \cT_3\setminus \cI$ intersects the neighborhood of some $p'\in\cI\cap \cT_3$. Thus, there is some point $q\in \cT_1$ that is in $V$ (by the argument above as it is a neighbor of $p'$) and is also a neighbor of $p$. It follows that also $p\in V$ as the subspace spanned by $q$ and $p$ contains some point in $\cT_2$ and both $\{q\},\cT_2\subset V$ (and we use pairwise independence again). Hence all the points in $\cT_3$ are in $V$. As $\cT_2\cup\cT_3 \subset V$ it follows that also $\cT_1\subset V$.
		\end{proof}
		This concludes the proof of the case $ m_3  \leq m_1^{1/3}$.
	\end{enumerate}
\end{proof}

		


\bibliographystyle{customurlbst/alphaurlpp}
\bibliography{main}

\begin{thebibliography}{BDWY13}

\bibitem[Agr05]{DBLP:conf/fsttcs/Agrawal05}
Manindra Agrawal.
\newblock \href {http://dx.doi.org/10.1007/11590156\_6} {Proving Lower Bounds
  Via Pseudo-random Generators}.
\newblock In Ramaswamy Ramanujam and Sandeep Sen, editors, {\em {FSTTCS} 2005:
  Foundations of Software Technology and Theoretical Computer Science, 25th
  International Conference, Hyderabad, India, December 15-18, 2005,
  Proceedings}, volume 3821 of {\em Lecture Notes in Computer Science}, pages
  92--105. Springer, 2005.

\bibitem[AV08]{DBLP:conf/focs/AgrawalV08}
Manindra Agrawal and V.~Vinay.
\newblock \href {http://dx.doi.org/10.1109/FOCS.2008.32} {Arithmetic Circuits:
  {A} Chasm at Depth Four}.
\newblock In {\em 49th Annual {IEEE} Symposium on Foundations of Computer
  Science, {FOCS} 2008, October 25-28, 2008, Philadelphia, PA, {USA}}, pages
  67--75. {IEEE} Computer Society, 2008.

\bibitem[BDSS16]{DBLP:journals/combinatorica/BhattacharyyaDS16}
Arnab Bhattacharyya, Zeev Dvir, Shubhangi Saraf, and Amir Shpilka.
\newblock \href {http://dx.doi.org/10.1007/s00493-015-3024-z} {Tight lower
  bounds for linear 2-query LCCs over finite fields}.
\newblock {\em Combinatorica}, 36(1):1--36, 2016.

\bibitem[BDWY13]{barak2013fractional}
Boaz Barak, Zeev Dvir, Avi Wigderson, and Amir Yehudayoff.
\newblock Fractional Sylvester--Gallai theorems.
\newblock {\em Proceedings of the National Academy of Sciences},
  110(48):19213--19219, 2013.

\bibitem[BMS13]{DBLP:journals/iandc/BeeckenMS13}
Malte Beecken, Johannes Mittmann, and Nitin Saxena.
\newblock \href {http://dx.doi.org/10.1016/j.ic.2012.10.004} {Algebraic
  independence and blackbox identity testing}.
\newblock {\em Inf. Comput.}, 222:2--19, 2013.

\bibitem[CKS18]{DBLP:conf/coco/ChouKS18}
Chi{-}Ning Chou, Mrinal Kumar, and Noam Solomon.
\newblock \href {http://dx.doi.org/10.4230/LIPIcs.CCC.2018.13} {Hardness vs
  Randomness for Bounded Depth Arithmetic Circuits}.
\newblock In Rocco~A. Servedio, editor, {\em 33rd Computational Complexity
  Conference, {CCC} 2018, June 22-24, 2018, San Diego, CA, {USA}}, volume 102
  of {\em LIPIcs}, pages 13:1--13:17. Schloss Dagstuhl - Leibniz-Zentrum
  f{\"{u}}r Informatik, 2018.

\bibitem[CLO07]{CLO}
David~A. Cox, John Little, and Donal O'Shea.
\newblock {\em Ideals, Varieties, and Algorithms: An Introduction to
  Computational Algebraic Geometry and Commutative Algebra}.
\newblock Springer, 3rd edition, 2007.

\bibitem[DS07]{DBLP:journals/siamcomp/DvirS07}
Zeev Dvir and Amir Shpilka.
\newblock \href {http://dx.doi.org/10.1137/05063605X} {Locally Decodable Codes
  with Two Queries and Polynomial Identity Testing for Depth 3 Circuits}.
\newblock {\em {SIAM} J. Comput.}, 36(5):1404--1434, 2007.

\bibitem[DSW12]{DBLP:journals/corr/abs-1211-0330}
Zeev Dvir, Shubhangi Saraf, and Avi Wigderson.
\newblock \href {http://arxiv.org/abs/1211.0330} {Improved rank bounds for
  design matrices and a new proof of Kelly's theorem}.
\newblock {\em CoRR}, abs/1211.0330, 2012.
\newblock Pre-print available at \href {http://arxiv.org/abs/1211.0330}
  {\path{arXiv:1211.0330}}.

\bibitem[DSY09]{DBLP:journals/siamcomp/DvirSY09}
Zeev Dvir, Amir Shpilka, and Amir Yehudayoff.
\newblock \href {http://dx.doi.org/10.1137/080735850} {Hardness-Randomness
  Tradeoffs for Bounded Depth Arithmetic Circuits}.
\newblock {\em {SIAM} J. Comput.}, 39(4):1279--1293, 2009.

\bibitem[EK66]{EdelsteinKelly66}
Michael Edelstein and Leroy~M. Kelly.
\newblock Bisecants of finite collections of sets in linear spaces.
\newblock {\em Canadian Journal of Mathematics}, 18:375--280, 1966.

\bibitem[FGT19]{DBLP:journals/cacm/FennerGT19}
Stephen~A. Fenner, Rohit Gurjar, and Thomas Thierauf.
\newblock \href {http://dx.doi.org/10.1145/3306208} {A deterministic parallel
  algorithm for bipartite perfect matching}.
\newblock {\em Commun. {ACM}}, 62(3):109--115, 2019.

\bibitem[For14]{ForbesThesis}
Michael~A. Forbes.
\newblock {\em Polynomial identity testing of read-once oblivious algebraic
  branching programs}.
\newblock PhD thesis, Massachusetts Institute of Technology, 2014.

\bibitem[FS13]{DBLP:conf/approx/ForbesS13}
Michael~A. Forbes and Amir Shpilka.
\newblock \href {http://dx.doi.org/10.1007/978-3-642-40328-6_37} {Explicit
  Noether Normalization for Simultaneous Conjugation via Polynomial Identity
  Testing}.
\newblock In Prasad Raghavendra, Sofya Raskhodnikova, Klaus Jansen, and
  Jos{\'{e}} D.~P. Rolim, editors, {\em Approximation, Randomization, and
  Combinatorial Optimization. Algorithms and Techniques - 16th International
  Workshop, {APPROX} 2013, and 17th International Workshop, {RANDOM} 2013,
  Berkeley, CA, USA, August 21-23, 2013. Proceedings}, volume 8096 of {\em
  Lecture Notes in Computer Science}, pages 527--542. Springer, 2013.

\bibitem[FSV18]{DBLP:journals/toc/ForbesSV18}
Michael~A. Forbes, Amir Shpilka, and Ben~Lee Volk.
\newblock \href {http://dx.doi.org/10.4086/toc.2018.v014a018} {Succinct Hitting
  Sets and Barriers to Proving Lower Bounds for Algebraic Circuits}.
\newblock {\em Theory of Computing}, 14(1):1--45, 2018.

\bibitem[GKKS16]{DBLP:journals/siamcomp/0001KKS16}
Ankit Gupta, Pritish Kamath, Neeraj Kayal, and Ramprasad Saptharishi.
\newblock \href {http://dx.doi.org/10.1137/140957123} {Arithmetic Circuits: {A}
  Chasm at Depth 3}.
\newblock {\em {SIAM} J. Comput.}, 45(3):1064--1079, 2016.

\bibitem[GT17]{DBLP:conf/stoc/GurjarT17}
Rohit Gurjar and Thomas Thierauf.
\newblock \href {http://dx.doi.org/10.1145/3055399.3055440} {Linear matroid
  intersection is in quasi-NC}.
\newblock In Hamed Hatami, Pierre McKenzie, and Valerie King, editors, {\em
  Proceedings of the 49th Annual {ACM} {SIGACT} Symposium on Theory of
  Computing, {STOC} 2017, Montreal, QC, Canada, June 19-23, 2017}, pages
  821--830. {ACM}, 2017.

\bibitem[Gup14]{Gupta14}
Ankit Gupta.
\newblock \href {http://eccc.hpi-web.de/report/2014/130} {Algebraic Geometric
  Techniques for Depth-4 {PIT} {\&} Sylvester-Gallai Conjectures for
  Varieties}.
\newblock {\em Electronic Colloquium on Computational Complexity {(ECCC)}},
  21:130, 2014.

\bibitem[HS80]{DBLP:conf/stoc/HeintzS80}
Joos Heintz and Claus{-}Peter Schnorr.
\newblock \href {http://dx.doi.org/10.1145/800141.804674} {Testing Polynomials
  which Are Easy to Compute (Extended Abstract)}.
\newblock In Raymond~E. Miller, Seymour Ginsburg, Walter~A. Burkhard, and
  Richard~J. Lipton, editors, {\em Proceedings of the 12th Annual {ACM}
  Symposium on Theory of Computing, April 28-30, 1980, Los Angeles, California,
  {USA}}, pages 262--272. {ACM}, 1980.

\bibitem[KI04]{DBLP:journals/cc/KabanetsI04}
Valentine Kabanets and Russell Impagliazzo.
\newblock \href {http://dx.doi.org/10.1007/s00037-004-0182-6} {Derandomizing
  Polynomial Identity Tests Means Proving Circuit Lower Bounds}.
\newblock {\em Computational Complexity}, 13(1-2):1--46, 2004.

\bibitem[KMSV13]{DBLP:journals/siamcomp/KarninMSV13}
Zohar~S. Karnin, Partha Mukhopadhyay, Amir Shpilka, and Ilya Volkovich.
\newblock \href {http://dx.doi.org/10.1137/110824516} {Deterministic Identity
  Testing of Depth-4 Multilinear Circuits with Bounded Top Fan-in}.
\newblock {\em {SIAM} J. Comput.}, 42(6):2114--2131, 2013.

\bibitem[KS09]{DBLP:conf/focs/KayalS09}
Neeraj Kayal and Shubhangi Saraf.
\newblock \href {http://dx.doi.org/10.1109/FOCS.2009.67} {Blackbox Polynomial
  Identity Testing for Depth 3 Circuits}.
\newblock In {\em 50th Annual {IEEE} Symposium on Foundations of Computer
  Science, {FOCS} 2009, October 25-27, 2009, Atlanta, Georgia, {USA}}, pages
  198--207. {IEEE} Computer Society, 2009.

\bibitem[KS11]{DBLP:journals/combinatorica/KarninS11}
Zohar~S. Karnin and Amir Shpilka.
\newblock \href {http://dx.doi.org/10.1007/s00493-011-2537-3} {Black box
  polynomial identity testing of generalized depth-3 arithmetic circuits with
  bounded top fan-in}.
\newblock {\em Combinatorica}, 31(3):333--364, 2011.

\bibitem[KS17]{DBLP:journals/toc/0001S17}
Mrinal Kumar and Shubhangi Saraf.
\newblock \href {http://dx.doi.org/10.4086/toc.2017.v013a006} {Arithmetic
  Circuits with Locally Low Algebraic Rank}.
\newblock {\em Theory of Computing}, 13(1):1--33, 2017.

\bibitem[KSS15]{DBLP:journals/cc/KoppartySS15}
Swastik Kopparty, Shubhangi Saraf, and Amir Shpilka.
\newblock \href {http://dx.doi.org/10.1007/s00037-015-0102-y} {Equivalence of
  Polynomial Identity Testing and Polynomial Factorization}.
\newblock {\em Computational Complexity}, 24(2):295--331, 2015.

\bibitem[MU05]{MU05-book}
Michael Mitzenmacher and Eli Upfal.
\newblock {\em Probability and Computing -- Randomized Algorithms and
  Probabilistic Analysis}.
\newblock Cambridge University Press, 2005.

\bibitem[Mul17]{Mulmuley-GCT-V}
Ketan~D. Mulmuley.
\newblock {Geometric complexity theory V: Efficient algorithms for Noether
  normalization}.
\newblock {\em J. Amer. Math. Soc.}, 30(1):225--309, 2017.

\bibitem[PS20]{DBLP:journals/corr/abs-2003-05152}
Shir Peleg and Amir Shpilka.
\newblock \href {https://arxiv.org/abs/2003.05152} {A generalized
  Sylvester-Gallai type theorem for quadratic polynomials}.
\newblock {\em CoRR}, abs/2003.05152, 2020.
\newblock Pre-print available at \href {http://arxiv.org/abs/2003.05152}
  {\path{arXiv:2003.05152}}.

\bibitem[Sax09]{Saxena09}
Nitin Saxena.
\newblock \href {https://eccc.weizmann.ac.il/report/2009/101/} {Progress on
  polynomial identity testing}.
\newblock {\em Bulletin of EATCS}, 99:49--79, 2009.

\bibitem[Sax14]{Saxena14}
Nitin Saxena.
\newblock \href {https://books.google.co.il/books?id=U7ApBAAAQBAJ} {Progress on
  Polynomial Identity Testing-II}.
\newblock In M.~Agrawal and V.~Arvind, editors, {\em Perspectives in
  Computational Complexity: The Somenath Biswas Anniversary Volume}, Progress
  in Computer Science and Applied Logic, pages 131--146. Springer International
  Publishing, 2014.

\bibitem[Shp19]{DBLP:conf/stoc/Shpilka19}
Amir Shpilka.
\newblock \href {http://dx.doi.org/10.1145/3313276.3316341} {Sylvester-Gallai
  type theorems for quadratic polynomials}.
\newblock In Moses Charikar and Edith Cohen, editors, {\em Proceedings of the
  51st Annual {ACM} {SIGACT} Symposium on Theory of Computing, {STOC} 2019,
  Phoenix, AZ, USA, June 23-26, 2019.}, pages 1203--1214. {ACM}, 2019.

\bibitem[SS12]{DBLP:journals/siamcomp/SaxenaS12}
Nitin Saxena and Comandur Seshadhri.
\newblock \href {http://dx.doi.org/10.1137/10848232} {Blackbox Identity Testing
  for Bounded Top-Fanin Depth-3 Circuits: The Field Doesn't Matter}.
\newblock {\em {SIAM} J. Comput.}, 41(5):1285--1298, 2012.

\bibitem[ST17]{DBLP:conf/focs/SvenssonT17}
Ola Svensson and Jakub Tarnawski.
\newblock \href {http://dx.doi.org/10.1109/FOCS.2017.70} {The Matching Problem
  in General Graphs Is in Quasi-NC}.
\newblock In Chris Umans, editor, {\em 58th {IEEE} Annual Symposium on
  Foundations of Computer Science, {FOCS} 2017, Berkeley, CA, USA, October
  15-17, 2017}, pages 696--707. {IEEE} Computer Society, 2017.

\bibitem[SV18]{DBLP:journals/combinatorica/SarafV18}
Shubhangi Saraf and Ilya Volkovich.
\newblock \href {http://dx.doi.org/10.1007/s00493-016-3460-4} {Black-Box
  Identity Testing of Depth-4 Multilinear Circuits}.
\newblock {\em Combinatorica}, 38(5):1205--1238, 2018.

\bibitem[SY10]{DBLP:journals/fttcs/ShpilkaY10}
Amir Shpilka and Amir Yehudayoff.
\newblock \href {http://dx.doi.org/10.1561/0400000039} {Arithmetic Circuits:
  {A} survey of recent results and open questions}.
\newblock {\em Foundations and Trends in Theoretical Computer Science},
  5(3-4):207--388, 2010.

\end{thebibliography}


\end{document}